\documentclass[11pt,a4paper]{article}
\usepackage{authblk}
\usepackage[top=25truemm,bottom=25truemm,left=25truemm,right=25truemm,footskip=10truemm]{geometry}
\usepackage{booktabs}
\usepackage{amsmath}
\usepackage{amsthm}
\usepackage{float}
\usepackage{amssymb}
\usepackage{amsfonts}
\usepackage{scalefnt}
\usepackage{here}
\usepackage{indentfirst}
\usepackage{algorithmic}
\usepackage{algorithm}
\usepackage{setspace}
\usepackage{color}
\usepackage{graphicx}
\usepackage{subcaption}
\usepackage{lscape}
\usepackage{natbib}
\usepackage{url}
\usepackage{placeins}
\usepackage{tikz}
\usepackage{makecell}
\usetikzlibrary{arrows.meta, positioning, calc, fit}
\DeclareMathOperator*{\argmin}{arg\,min}

\newtheorem{theorem}{Theorem}[section]
\newtheorem{proposition}[theorem]{Proposition}
\newtheorem{corollary}[theorem]{Corollary}
\newtheorem{lemma}[theorem]{Lemma}

\newtheorem{definition}{Definition}[section]
\newtheorem{remark}{Remark}[section]

\title{
Estimating Consensus Epidemic Trajectories via a Constrained Power Fr\'echet Mean with Functional Registration
}
\author[1,*]{Yui Tomo}
\author[1,2]{Shu Tamano}
\author[1,3]{Daisuke Yoneoka}
\affil[1]{Department of Epidemiology, National Institute of Infectious Diseases, Japan Institute for Health Security, 1-23-1 Toyama, Shinjuku-ku, Tokyo 162-0052, Japan}
\affil[2]{Department of Multidisciplinary Sciences, Graduate School of Arts and Sciences, The University of Tokyo, 3-8-1 Komaba, Meguro-ku, Tokyo 153-8902, Japan}
\affil[3]{Department of Global Health Policy, Graduate School of Medicine, The University of Tokyo, 7-3-1 Hongo, Bunkyo-ku, Tokyo 113-0033, Japan}
\affil[*]{E-mail: tomo.y@jihs.go.jp}
\date{}

\begin{document}
\maketitle

\begin{abstract}
In infectious disease modeling during the early phase of a pandemic, SEIR-type compartmental models are standard tools, and they require epidemiological parameters as inputs.
Because these parameters are subject to uncertainty, different research groups often report different epidemic curves, and obtaining a representative curve that captures the characteristics of these trajectories is important for decision-making.
However, simple pointwise summaries can attenuate epidemic peaks under temporal misalignment and generally do not preserve the dynamical structure of the underlying compartmental models.
To address these limitations, we propose a method for summarizing multiple solutions to SEIR-type compartmental models on a functional space by computing a constrained power Fr\'echet mean with temporal shift registration.
In our method, we regard the pairs of exposed and infectious compartments as objects in a Hilbert space, and the consensus curve is defined as the solution to a constrained optimization problem.
Differential equation constraints and population constraints are incorporated in the optimization to preserve a partially mechanistic interpretation regarding the infectious compartment.
We develop an implementable block-optimization algorithm based on basis function expansion.
In simulation studies based on early COVID-19 parameter estimates, the proposed method produced consensus curves with a single epidemic peak, whereas pointwise summaries exhibited attenuated or multiple peaks.
We further applied the method to six literature-derived parameter sets from early COVID-19 studies and obtained a representative trajectory with interpretable epidemiological parameters.
The proposed approach provides a generalized trajectory-summarization method that includes mean- and median-type estimators and preserves mechanistic interpretability.
\end{abstract}


\section{Introduction}

Mathematical modeling of infectious diseases plays an important role in understanding epidemic dynamics, predicting outbreak trajectories, and evaluating the potential impact of interventions \citep{white2010mathematical}.
Among various modeling approaches, compartmental models, such as the susceptible--infectious--removed (SIR) model and its extensions, have been widely used due to their interpretability, analytical tractability, and applicability to diverse infectious diseases \citep{kermack1927contribution}.
These models describe the time evolution of subpopulations (or compartments) such as susceptible, infectious, and removed individuals, governed by systems of ordinary differential equations.
The parameters in these models reflect epidemiological quantities such as transmission rates, incubation periods, and recovery rates.
The susceptible--exposed--infectious--removed (SEIR) model extends the basic SIR model by incorporating an exposed compartment to model the latent period between infection and infectiousness, which is essential for diseases with incubation periods \citep{aron1984seasonality}.
Therefore, respiratory infectious diseases such as COVID-19 and influenza are typically modeled using SEIR-based approaches \citep{he2020seir,etbaigha2018seir}.

In practice, the compartmental models require input parameters such as transmission rates, incubation periods, and recovery rates.
These parameters are typically estimated from epidemiological data and hence are subject to uncertainty.
For example, Figure~\ref{fig:trajectories_I_naive_mean_median} displays five infectious population curves from the SEIR model with different parameter values, depicted as gray lines.
These trajectories were generated by sampling five realizations from probability distributions constructed based on the parameter estimates and confidence intervals reported in \cite{li2020early}.
Refer to Section~\ref{sec:simulation} for detailed simulation settings.
This figure illustrates that stochastic variations in parameter estimates can lead to large discrepancies across the resulting curves.
This variability poses a challenge for researchers and policymakers: how to summarize a set of heterogeneous epidemic trajectories into a single, consensus curve that facilitates decision-making.
Furthermore, in practice, the parameter estimates may depend on the availability and quality of data.
Therefore, depending on the data source and estimation method, different researchers or research groups may obtain different trajectories.
This leads to the problem of summarizing a set of such trajectories in a representative way.

\begin{figure}[htbp]
    \centering
    \includegraphics[width=0.8\textwidth]{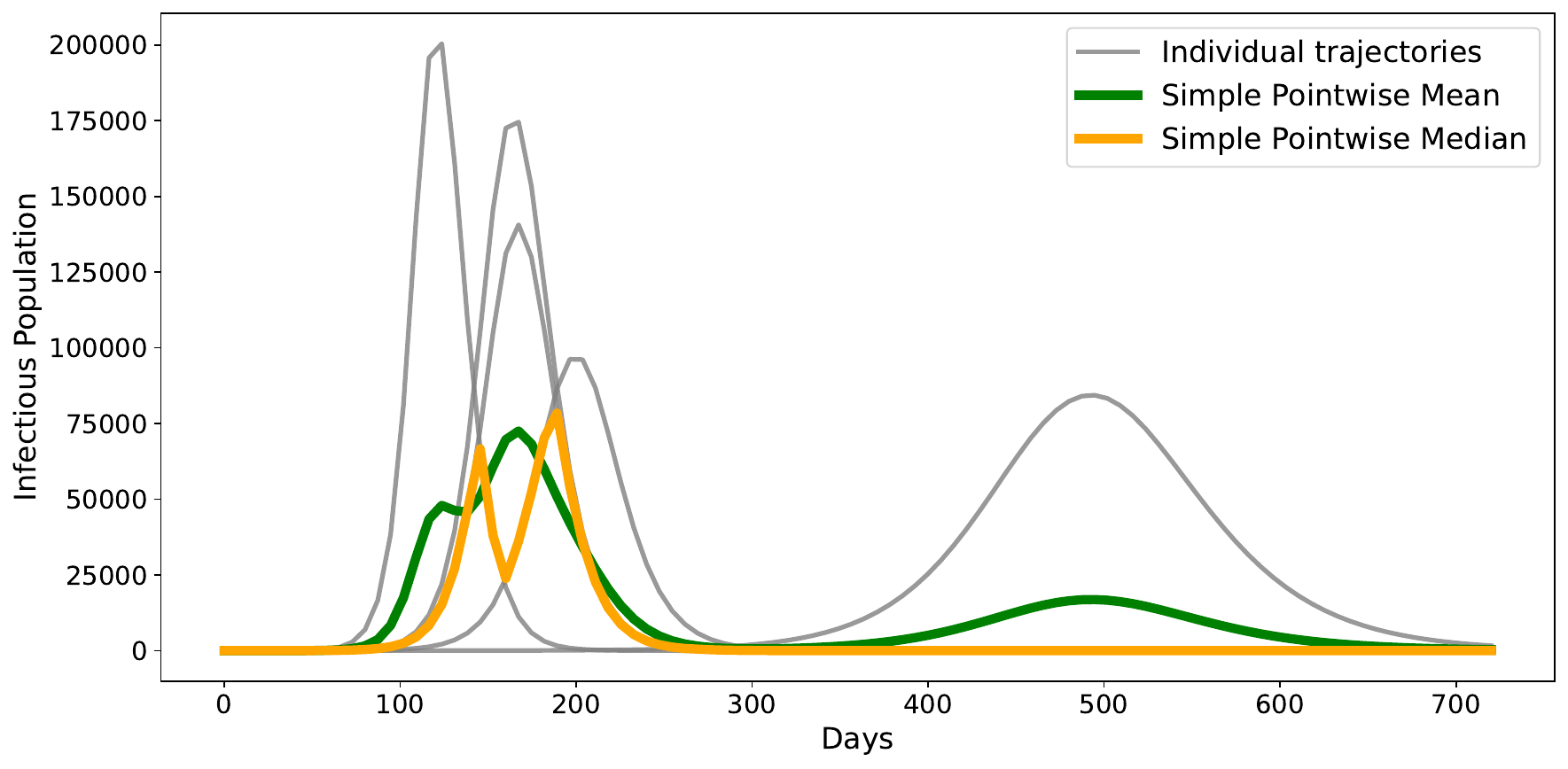}
    \caption{Simulated SEIR trajectories (gray) with the simple pointwise mean (green) and the median (orange).}
    \label{fig:trajectories_I_naive_mean_median}
\end{figure}

When the target of summarization is itself an epidemic trajectory, it is natural to evaluate a representative summary through a loss function defined directly on the space of trajectories.
From a decision-theoretic perspective, the representative trajectory can then be viewed as an action chosen to minimize its discrepancy from the collection of model outputs.
This perspective motivates performing the aggregation in the trajectory space.
Since the map from parameters to trajectories is nonlinear, the trajectory generated by the averaged parameters is not the average of the trajectories.
The former is not guaranteed to be representative in terms of the peak timing and magnitude that are of primary interest for decision-making.
Moreover, when different research groups report trajectories, their parameters are not necessarily comparable because the underlying model structures and parameterizations may differ.

A naive approach to summarizing curves is to compute the pointwise mean of the curves, which preserves the smoothness of the original curves.
However, this approach can be sensitive to outlier curves that take extreme values. 
Moreover, due to the misalignment of curves, peaks may be desynchronized across curves, causing the averaged peak intensity to be attenuated \citep{bigot2013frechet}. 
On the other hand, taking the pointwise median does not guarantee the smoothness of the resulting curve, since the median at each time point may come from different curves.
In Figure \ref{fig:trajectories_I_naive_mean_median}, the simple pointwise mean and median are depicted by the green and orange lines, respectively.
As shown, both naive summarized curves exhibit attenuated and multiple peaks.
Specifically, the pointwise mean displays a peak driven by the rightmost outlier trajectory, while the median lacks smoothness.
Crucially, neither the pointwise mean nor the median yields a curve that satisfies the underlying differential equations that describe the dynamics of infectious compartment.
Consequently, the resulting curve lacks mechanistic interpretability.

To address this issue, we formulate the summarization in a functional space of epidemic curves, while accounting for both temporal misalignment and the underlying disease dynamics.
\cite{bigot2013frechet} proposed an approach based on Fr\'echet means to summarize a collection of temporally misaligned curves by jointly estimating deformation parameters and a representative shape.
Inspired by this work, we propose integrating solutions to SEIR-type compartment models via a power Fr\'echet mean defined on a functional space \citep{schotz2022strong}.
Here, the power Fr\'echet mean is a generalization of the standard Fr\'echet mean, that includes Fr\'echet median and continuously intermediate barycenters.
In particular, we focus on the exposed and infectious compartments and define the target estimator as the solution to the power Fr\'echet mean-type optimization problem subject to the differential equations constraint that govern their dynamics.
We also incorporate a temporal functional registration to account for variability in the timing of epidemic peaks across trajectories by introducing temporal shift parameters.
We develop an estimation algorithm based on techniques from functional data analysis \citep{ramsay2005functional,hsing2015theoretical,wang2016functional}.
Based on our formulation and algorithm, we can jointly estimate the rate at which exposed individuals become infectious and the removal rate of infectious individuals.
From the estimated exposed and infectious trajectories, the susceptible and removed compartments can also be reconstructed.
Furthermore, we demonstrate the utility of our approach through numerical studies based on estimated parameter values from early epidemiological studies on the COVID-19 pandemic.

The remainder of this paper is organized as follows:
Section~\ref{sec:relatedworks} introduces the formulation of compartmental models and describes related works on ensemble approaches for epidemic curves and statistics on metric space.
Section~\ref{sec:cpfm-definition} introduces formulations of constrained power Fr\'echet mean and establishes several basic properties.
Section~\ref{sec:repcurve} introduces the formulation of our proposed summarization approach by defining a functional space for solutions from compartmental models and a consensus curve based on the constrained power Fr\'echet mean.
Section~\ref{sec:algorithm} develops an implementable algorithm to obtain the consensus curve, recover a full trajectory, and estimate epidemiological parameters based on basis function expansion.
Section~\ref{sec:simulation} illustrates the proposed methods via a simulation study using parameters of the COVID-19 epidemic reported during an early phase.
Section~\ref{sec:application} applies our approach to obtain consensus epidemic trajectories from literature-derived epidemiological parameter sets from multiple research groups.
Section~\ref{sec:discussion} discusses the advantages and limitations of our approach.

\section{Preliminaries and Related Works}
\label{sec:relatedworks}

\subsection{SEIR-type Compartmental Models}

The SIR model was originally proposed by \cite{kermack1927contribution} and has been widely used as a fundamental model for describing infectious disease dynamics.
However, for respiratory infectious diseases typically characterized by a latent incubation period, compartmental models that explicitly include an exposed population are often preferred to capture the delay between infection and infectiousness.
In this study, we focus on SEIR-type compartmental models, including such an exposed compartment.

We introduce notations.
Let $N \in \mathbb{N}$ denote the total population.
Let $S, E, I, R:[0, T] \to [0, N]$ ($T > 0$) be continuously differentiable functions representing the population sizes of the susceptible, exposed, infectious, and removed compartments, respectively.
The SEIR model is a standard formulation for models consisting of these four compartments.
Specifically, the SEIR model is defined as the following system of ordinary differential equations:
\begin{align}
    \begin{cases}
        \displaystyle \frac{dS}{dt} = -\beta\, \frac{S(t)I(t)}{N}
        ,
        \\[10pt]
        \displaystyle \frac{dE}{dt} = \beta\, \frac{S(t)I(t)}{N} - \sigma\, E(t)
        ,
        \\[10pt]
        \displaystyle \frac{dI}{dt} = \sigma\, E(t) - \gamma\, I(t)
        ,
        \\[10pt]
        \displaystyle \frac{dR}{dt} = \gamma\, I(t)
        ,
    \end{cases}
    \label{eq:SEIR}
\end{align}
where $\beta > 0$ is the disease transmission rate, $\sigma > 0$ is the rate at which exposed individuals become infectious, and $\gamma > 0$ is the removal rate.
The initial conditions are given by $(S(0),E(0),I(0),R(0))=(N-E_0-I_0,E_0,I_0,0)$, with $N-E_0-I_0>0$, $E_0 \geq 0$, $I_0 \geq 0$, and either $E_0 > 0$ or $I_0 > 0$.

Various extensions of the SEIR model have been proposed to incorporate more complex epidemiological features, such as the SEIUR model where unreported cases are considered, the SEIQR model that incorporates a quarantined compartment, or the SEIRD model that distinguishes between recovered and deceased individuals \citep{griette2021can,gerberry2009seiqr,weitz2015modeling}.
The schematic diagrams of these models are provided in Figure \ref{fig:diagram_compartmental_models}.
In any such model, the dynamics of the infectious compartment $I$ are commonly governed by a differential equation where the rate of increase is proportional to $E$ and the rate of decrease is proportional to $I$ itself.
Furthermore, it is assumed that the total population remains constant across all models.
Specifically, when considering the four compartments $S, E, I,$ and $R$, we impose the constraint that $S(t) + E(t) + I(t) + R(t) = N$ for all $t \in [0, T]$.

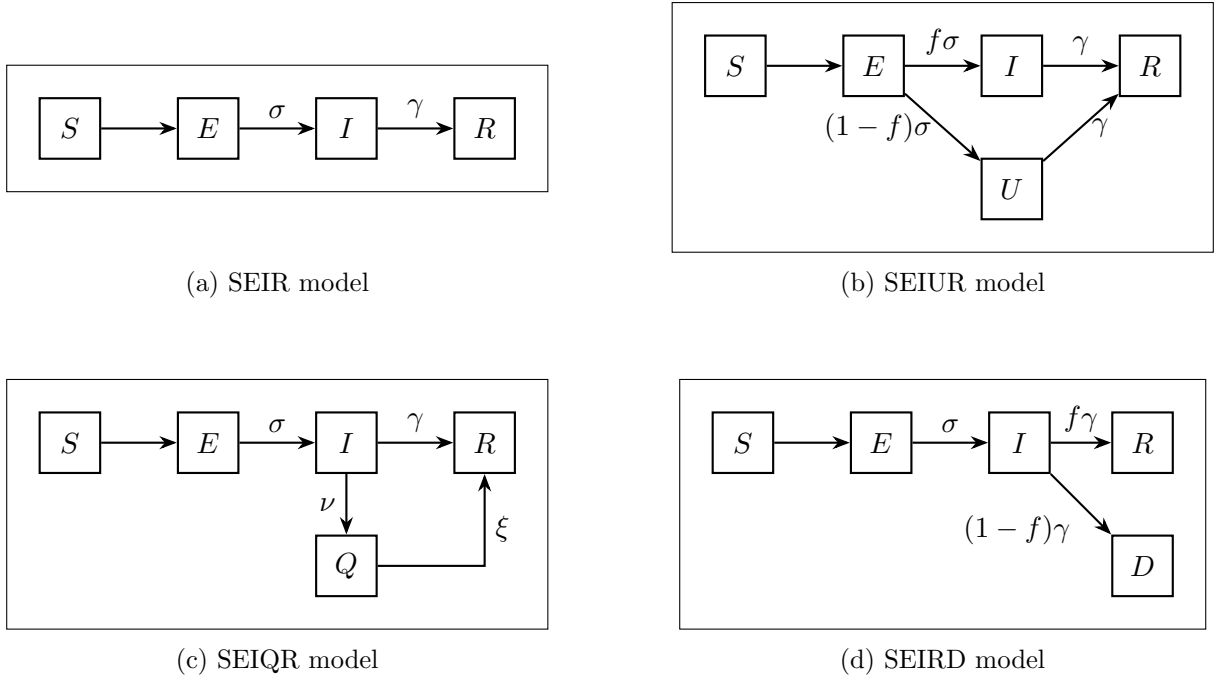
\begin{figure}[htbp]
    \centering
    \tikzset{
        node distance=1.2cm and 1.0cm,
        box/.style={rectangle, draw, minimum width=0.8cm, minimum height=0.8cm, thick},
        arrow/.style={-Stealth, thick},
        frame/.style={draw=black, thin, inner sep=12pt}
    }

    \begin{subfigure}[b]{0.45\textwidth}
        \centering
        \begin{tikzpicture}[baseline=(current bounding box.center)]
            \node[box] (S) {$S$};
            \node[box, right=of S] (E) {$E$};
            \node[box, right=of E] (I) {$I$};
            \node[box, right=of I] (R) {$R$};

            \draw[arrow] (S) -- (E);
            \draw[arrow] (E) -- node[above] {$\sigma$} (I);
            \draw[arrow] (I) -- node[above] {$\gamma$} (R);
            
            \node[frame, fit=(S) (R)] (F1) {};
            \path (F1.south) ++(0,-0.8cm);
            \path (F1.north) ++(0,0.8cm);
        \end{tikzpicture}
        \caption{SEIR model}
        \label{fig:model_seir}
    \end{subfigure}
    \hfill
    \begin{subfigure}[b]{0.45\textwidth}
        \centering
        \begin{tikzpicture}[baseline=(current bounding box.center)]
            \node[box] (S) {$S$};
            \node[box, right=of S] (E) {$E$};
            \node[box, right=of E] (I) {$I$};
            \node[box, right=of I] (R) {$R$};
            \node[box, below=0.8cm of I] (U) {$U$};

            \draw[arrow] (S) -- (E);
            \draw[arrow] (E) -- node[above] {$f \sigma$} (I);
            \draw[arrow] (I) -- node[above] {$\gamma$} (R);
            \draw[arrow] (E) -- node[left] {$(1-f)\sigma$} (U);
            \draw[arrow] (U) -- node[right] {$\gamma$} (R);
            
            \node[frame, fit=(S) (R) (U)] {};
        \end{tikzpicture}
        \caption{SEIUR model}
        \label{fig:model_seiur}
    \end{subfigure}

    \vspace{1cm}

    \begin{subfigure}[b]{0.45\textwidth}
        \centering
        \begin{tikzpicture}[baseline=(current bounding box.center)]
            \node[box] (S) {$S$};
            \node[box, right=of S] (E) {$E$};
            \node[box, right=of E] (I) {$I$};
            \node[box, right=of I] (R) {$R$};
            \node[box, below=0.8cm of I] (Q) {$Q$};

            \draw[arrow] (S) -- (E);
            \draw[arrow] (E) -- node[above] {$\sigma$} (I);
            \draw[arrow] (I) -- node[above] {$\gamma$} (R);
            \draw[arrow] (I) -- node[left] {$\nu$} (Q);
            \draw[arrow] (Q) -| node[pos=0.7, right] {$\xi$} (R);
            
            \node[frame, fit=(S) (R) (Q)] {};
        \end{tikzpicture}
        \caption{SEIQR model}
        \label{fig:model_seiqr}
    \end{subfigure}
    \hfill
    \begin{subfigure}[b]{0.45\textwidth}
        \centering
        \begin{tikzpicture}[baseline=(current bounding box.center)]
            \node[box] (S) {$S$};
            \node[box, right=of S] (E) {$E$};
            \node[box, right=of E] (I) {$I$};
            \node[box, right=0.8cm of I] (R) {$R$};
            \node[box, below=0.8cm of R] (D) {$D$};

            \draw[arrow] (S) -- (E);
            \draw[arrow] (E) -- node[above] {$\sigma$} (I);
            \draw[arrow] (I) -- node[above] {$f \gamma$} (R);
            \draw[arrow] (I) -- node[below left] {$(1-f) \gamma$} (D);
            
            \node[frame, fit=(S) (R) (D)] {};
        \end{tikzpicture}
        \caption{SEIRD model}
        \label{fig:model_seird}
    \end{subfigure}

    \caption{Schematic diagrams of the SEIR model and its extensions.
    (a) The standard SEIR model.
    (b) The SEIUR model, which incorporates the unreported compartment $U$ and a proportion $0 \leq f \leq 1$.
    (c) The SEIQR model, which includes the quarantined compartment $Q$ and the transition rates $0 < \nu$ and $0 < \xi$.
    (d) The SEIRD model, which considers the deceased compartment $D$ and a proportion $0 \leq f \leq 1$.
    }
    \label{fig:diagram_compartmental_models}
\end{figure}

\subsection{Model Aggregation Methods for Epidemic Curves}

The challenge of aggregating multiple epidemic trajectories has led to the development of various ensemble techniques.
A common strategy for integrating epidemiological models is to form a weighted ensemble of predictive distributions or trajectories.
Simple unweighted methods, such as taking the pointwise mean or median of submitted forecasts, have been used in influenza and COVID-19 forecasting \citep{mcgowan2019collaborative, ray2023comparing, fox2024optimizing}.
More sophisticated approaches assign weights based on historical performance (performance-based weights) or predictive density stacking \citep{ray2018prediction, reich2019accuracy}.
Recent studies have also explored adaptive weighting schemes where weights are updated over time as an epidemic progresses \citep{chowell2020real, mcandrew2021adaptively}.
For instance, \cite{taylor2022interval} and \cite{taylor2023combining} evaluated various combination methods for COVID-19 mortality forecasts, including trimmed means and inverse-score weighting, finding that weighted ensembles often outperform individual models.
Bayesian approaches provide a formal probabilistic method for model integration.
Bayesian Model Averaging and its variants have been applied to dengue and other infectious disease forecasts to reconcile differences between competing models \citep{yamana2016superensemble, yamana2017individual}.
While these methods are effective for point and interval predictions, they do not explicitly preserve the mathematical structure of the compartmental dynamics in the aggregated results.

\subsection{Statistics on Metric Space and Functional Data Analysis}

The field of statistics that deals with statistical methods and theory for objects taking values in a metric space is referred to as metric statistics \citep{muller2016peter,dubey2024metric}.
We introduce the power Fr\'echet mean \citep{schotz2022strong}.
For a given metric space $(\Omega, d_{\Omega})$, the power Fr\'echet mean of random objects $u$ in $\Omega$ is defined as $\bar{u}:=\argmin_{v \in \Omega} \mathbb{E}[d_{\Omega}(v, u)^{q}]$ with $q > 0$.
$\bar{u}$ is consistent with the standard Fr\'echet mean when $q=2$ and the Fr\'echet median when $q=1$.
Furthermore, we can interpret it as the mode and circumcenter when $q \to 0$ and $q \to \infty$, respectively.
In this study, we consider an empirical version of the power Fr\'echet mean.
Specifically, for random objects $u_1,\ldots , u_{J} \in \Omega$, an empirical power Fr\'echet mean is defined as 
\begin{align*}
    \hat{u}:=\argmin_{v \in \Omega} \frac{1}{J}\sum_{j=1}^{J} d_{\Omega}(v, u_{j})^{q}
\end{align*}
We call this empirical version simply the power Fr\'echet mean in the following parts of this paper.

Metric statistics encompasses general metric spaces beyond Euclidean spaces, providing a systematic approach for the statistical analysis of complex data such as compositional data or functional data.
In the context of functional data, an object is treated as a function when multiple observations are obtained over a continuous window, such as the passage of time.
By defining a metric space for functions, the analytical methods and theoretical results in metric statistics can be applied.
Specifically, functional data analysis provides a practical implementation for handling such functional data \citep{ramsay2005functional,hsing2015theoretical,wang2016functional}.

Since epidemic curves are inherently functional objects, functional data analysis offers a natural approach for handling them.
For example, \cite{boschi2021functional} used functional data analysis to extract patterns from COVID-19 mortality curves across different regions.
Building upon this functional perspective, we propose a method that integrates multiple solutions to compartmental models by solving a constrained optimization problem within a functional space.
Unlike existing ensemble approaches above, our approach incorporates a governing differential equation as constraints, thereby ensuring that the resulting consensus curves preserve a part of the mathematical structure of the compartmental models.
This also allows for the simultaneous estimation of the underlying epidemiological parameters alongside the integrated trajectories.

\section{Constrained Power Fr\'echet Mean}
\label{sec:cpfm-definition}

In this section, we introduce a constrained power Fr\'echet mean on a metric space.

\begin{definition}[Constrained power Fr\'echet mean]
\label{def:cpfm_metric}
Let $(\Omega,d_{\Omega})$ be a metric space, let $u_1,\ldots,u_J \in \Omega$, let $q>0$, and let $C \subset \Omega$ be a non-empty set.
We define a constrained power Fr\'echet mean $\hat{u}_{q}$ on the constraint set $C$ as
\begin{align}
    \hat{u}_{q}
    \in
    \argmin_{u \in C}
    \frac{1}{J}\sum_{j=1}^{J} d_{\Omega}(u,u_j)^{q}
    .
    \label{eq:const-power-frechet-q-mean}
\end{align}
\end{definition}

In our later application, the feasible set itself depends on additional parameters to be optimized.
This motivates the following extension of Definition \ref{def:cpfm_metric}.

\begin{definition}[An extended constrained power Fr\'echet mean]
\label{def:extended_cpfm_metric}
Let $(\Omega,d_{\Omega})$ be a metric space, let $\Theta \subset \mathbb{R}^{p}$ be a $p$-dimensional parameter space, let $u_1,\ldots,u_J : \Theta \to \Omega$, and let $q>0$.
For each $\theta \in \Theta$, let $\mathcal{C}(\theta) \subset \Omega$ be a nonempty constraint set.
We define an extended constrained power Fr\'echet mean as any minimizer $\hat{u}_q$ of 
\begin{align}
    (\hat{u}_q,\hat\theta)
    \in
    \argmin_{(u,\theta)\in \operatorname{Graph}(\mathcal{C})}
    \frac{1}{J}\sum_{j=1}^{J} d_{\Omega}(u,u_j(\theta))^q
    ,
    \label{eq:extended-const-power-frechet-q-mean}
\end{align}
where
$
    \operatorname{Graph}(\mathcal{C})
    :=
    \left\{
        (u,\theta)\in \Omega\times\Theta
        \mid
        u\in \mathcal{C}(\theta)
    \right\}
$
is the graph of the set-valued map $\mathcal{C}:\Theta \rightrightarrows \Omega$.
We note that Definition~\ref{def:cpfm_metric} is recovered when each $u_j$ is a constant map.
\end{definition}

For notational convenience, we write the objective functions in Definitions~\ref{def:cpfm_metric} and \ref{def:extended_cpfm_metric} as
\begin{align*}
    F_{q}(u)
    :=
    \frac{1}{J}\sum_{j=1}^{J} d_{\Omega}(u,u_j)^{q}
    ,
    \quad
    \text{and}
    \quad
    F_{q}(u,\theta)
    :=
    \frac{1}{J}\sum_{j=1}^{J} d_{\Omega}(u,u_j(\theta))^{q}
    ,
\end{align*}
respectively.
We then prove some basic properties such as existence, uniqueness under $q>1$, and robustness under $q=1$ of the constrained power Fr\'echet mean that hold in an arbitrary Hilbert space.
These properties are used in Sections~\ref{sec:repcurve} and \ref{sec:algorithm} to justify the proposed estimator and its algorithm.
See Appendix~\ref{sec:proof} for the proofs.

We first establish the existence of the constrained power Fr\'echet mean.
The three statements below correspond to a compact constraint set, a closed convex constraint set in a Hilbert space, and the extended problem of Definition~\ref{def:extended_cpfm_metric}, respectively.

\begin{proposition}[Existence]
\label{prop:existence_general}
Let $(\Omega,d_{\Omega})$ be a metric space and let $q>0$.
\begin{enumerate}
    \item[(i)] Let $u_1,\ldots,u_J\in\Omega$. If $C\subset\Omega$ is a non-empty and compact set, then a minimizer $\hat{u}_q$ of \eqref{eq:const-power-frechet-q-mean} exists.
    \item[(ii)] Suppose that $\Omega$ is a Hilbert space with $d_{\Omega}(u,v)=\|u-v\|_{\Omega}$.
    Let $u_1,\ldots,u_J\in\Omega$.
    If $C\subset\Omega$ is a non-empty, closed, and convex set, then a minimizer $\hat{u}_q$ of \eqref{eq:const-power-frechet-q-mean} exists.
    \item[(iii)] Suppose that $\Omega$ is a Hilbert space with $d_{\Omega}(u,v)=\|u-v\|_{\Omega}$.
    Let $\Theta\subset\mathbb R^{p}$ be a compact and non-empty set, let
    $u_1,\ldots,u_J:\Theta\to\Omega$ be continuous maps, and let
    $\mathcal{C}:\Theta\rightrightarrows\Omega$ have non-empty values.
    If $\operatorname{Graph}(\mathcal{C})$ is weakly sequentially closed as defined in
    Definition~\ref{def:wsc_graph}, then a minimizer $(\hat{u}_q,\hat\theta)$ of
    \eqref{eq:extended-const-power-frechet-q-mean} exists.
\end{enumerate}
\end{proposition}

We next consider the uniqueness of the constrained power Fr\'echet mean.
The following proposition shows that the power $q>1$ yields a unique minimizer.

\begin{proposition}[Uniqueness]
\label{prop:uniqueness}
Suppose that $\Omega$ is a Hilbert space with $d_{\Omega}(u,v)=\|u-v\|_{\Omega}$.
If $C\subset\Omega$ is a non-empty, closed, and convex set and $q>1$, then the minimizer of \eqref{eq:const-power-frechet-q-mean} is unique.
\end{proposition}

\begin{remark}
We note that $\mathcal{F}$ in Section~\ref{sec:joint_frechet_q_mean} is not convex because of the bilinear term $\sigma E$ in the differential equation constraint.
Proposition~\ref{prop:uniqueness} therefore applies to the inner problem for a fixed $(\sigma,\gamma,\boldsymbol\delta)$, which is exactly the convex program solved in Appendix~\ref{sec:simplified_optimization}, and not to \eqref{eq:joint_shift_frechet} as a whole.
\end{remark}

The following proposition shows that the minimizer under $q=1$ remains bounded even when one of the objects diverges.

\begin{proposition}[Robustness to an outlier object under $q=1$]
\label{prop:robustness}
Suppose that $\Omega$ is a Hilbert space with $d_{\Omega}(u,v)=\|u-v\|_{\Omega}$.
Suppose that $C\subset\Omega$ is a non-empty, closed, and convex set.
Fix $u_1,\ldots,u_{J-1} \in \Omega$ with $J\geq 3$ and take a sequence $(u_J^{(n)}|~n\in\mathbb N)\subset\Omega$.
For each $n$, define
\begin{align*}
    \hat{u}_q^{(n)}
    \in
    \argmin_{u \in C}
    \frac{1}{J}\left(d_{\Omega}(u,u_J^{(n)})^{q} + \sum_{j=1}^{J-1} d_{\Omega}(u,u_j)^{q}\right)
    .
\end{align*}
Under these settings, if $q=1$, then it holds that $\sup_{n\in\mathbb N}\|\hat{u}_1^{(n)}\|_{\Omega}<\infty$ even when $\|u_J^{(n)}\|_{\Omega} \to \infty$.
\end{proposition}

\section{Formulation of Proposed Curve Summarization}
\label{sec:repcurve}

In this section, we define a functional space for $(E, I)$ and define a consensus curve as a constrained power Fr\'echet mean.

\subsection{Functional Space}
\label{sec:deffuncspace}

We assume that each compartment function $S$, $E$, $I$, and $R$ belongs to the Sobolev space $H^1([0,T])$.
Here, the Sobolev space $H^1([0, T])$ is the space of functions $u \in L^2([0, T])$ such that their weak derivatives $u^{\prime}$ exist and also belong to $L^2([0, T])$.
From Theorem 8.2 of \cite{brezis2011functional}, there exists a function $\tilde{u} \in C([0,T])$ such that $u = \tilde{u}$ a.e. on $[0,T]$ and $\tilde{u}(t_2) - \tilde{u}(t_1) = \int_{t_1}^{t_2} u^{\prime}(t) dt$ for all $t_1, t_2 \in [0, T]$ for each $u \in H^1([0, T])$.
Here, $C([0, T])$ is a function space of continuous functions defined on $[0,T]$.
Consequently, we identify the compartment functions with their continuous representatives and consider them to be defined for all $t \in [0, T]$.

We focus on the exposed and infectious compartments, $E$ and $I$, respectively, and the vector function: 
\begin{align*}
    y
    :=
    (E,I)^{\top} \in \mathcal{Y}
    :=
    \left\{u:[0, T] \rightarrow \mathbb{R}^2 \mid u_1, u_2 \in H^1([0, T])\right\}
    .
\end{align*}
Given $E$, $I$, and a fixed parameter $\gamma$, the remaining compartment functions $S$ and $R$ are uniquely determined by defining the removal process as ${dR}/{dt} = \gamma I$, and the conservation law $S+E+I+R=N$.
Therefore, it is sufficient to focus on the pair $(E(t), I(t))$ when constructing consensus trajectories or defining distances between them.
It should be noted that the trajectories $(S, E, I, R)$ constructed in this manner do not necessarily satisfy the transmission dynamics specified by the SEIR model, where additional differential equations are imposed on the susceptible and exposed populations.

Let $y_1(t)=(E_1(t),I_1(t))^{\top}$ and $y_2(t)=(E_2(t),I_2(t))^{\top}$ be two trajectories in $\mathcal Y$, where $E_i(t)$ and $I_i(t)$ denote the exposed and infectious compartments for the sample $i=1,2$, respectively, and $E_i^{\prime}(t)$, $I_i^{\prime}(t)$ are their time derivatives.
Since $H^{1}([0, T])$ is a Hilbert space (Proposition 8.1 in \cite{brezis2011functional}), the product space $\mathcal{Y}$ equipped with the following inner product is also a Hilbert space (Section II.1, Example 5 of \cite{reed2012methods}):
\begin{align*}
    \langle y_1, y_2 \rangle_{\mathcal Y}
    :=
    \int_{0}^{T} \Bigl( E_1 E_2 + \rho E_1^{\prime} E_2^{\prime} + I_1 I_2 + \rho I_1^{\prime} I_2^{\prime} \Bigr)\,dt
    ,
\end{align*}
where $\rho > 0$ is a scaling constant.
The induced norm is $\|y\|_{\mathcal Y} := \sqrt{\langle y, y \rangle_{\mathcal Y}}$.
We measure the dissimilarity between trajectories using the metric derived from this norm:
\begin{align*}
    d(y_1,y_2)&:=\|y_1-y_2\|_{\mathcal Y} \\
    &=
    \biggl(
    \int_{0}^{T}\Bigl[|E_1-E_2|^{2}+ \rho |E_1^{\prime}-E_2^{\prime}|^{2}+|I_1-I_2|^{2}+ \rho |I_1^{\prime}-I_2^{\prime}|^{2}\Bigr]\,dt\biggr)^{1/2}
    .
\end{align*}
We formulate all subsequent optimization problems in this Hilbert space $(\mathcal Y, \langle \cdot, \cdot \rangle_{\mathcal Y})$.

\subsection{Shift Parameter for Registration}
\label{sec:shift_registration}

We introduce a shift parameter $\delta\in\mathbb R$ to align curves, and define the shift operator
\begin{align*}
    (\mathcal T_{\delta} y_j)(t) := y_j^{\mathrm{ext}}(t+\delta)
    ,
    \quad\text{for}~
    t\in[0,T]
    ,
\end{align*}
where $y_j^{\mathrm{ext}}(t)$ is the constant extension for each trajectory $y_j=(E_j,I_j)^\top\in\mathcal Y$:
\begin{align*}
    y_j^{\mathrm{ext}}(t):=
    \begin{cases}
        y_j(0), & t<0
        ,\\
        y_j(t), & 0\le t\le T
        ,\\
        y_j(T), & t>T
        .
    \end{cases}
\end{align*}
We note that $\mathcal T_{\delta}y_j\in \mathcal Y$ for any $\delta\in\mathbb R$.

\subsection{Constrained Power Fr\'echet Mean with Functional Registration}
\label{sec:joint_frechet_q_mean}

Assume that we have $J$ trajectories $y_1, \ldots, y_J \in \mathcal Y$.
Let $\boldsymbol{\delta}:=(\delta_{1},\ldots,\delta_{J}) \in \mathbb{R}^{J}$ be the vector of shift parameters for each trajectory and define
\begin{align*}
    \Delta
    :=
    \left\{
        \boldsymbol{\delta}
        =
        (\delta_1,\ldots,\delta_J)
        \in[-\delta_{\max},\delta_{\max}]^J
        \left|
        \sum_{j=1}^J\delta_j=0
        \right.
    \right\}
    .
\end{align*}
Define the feasible set
\begin{align*}
    \mathcal F
    :=
    \left\{
        \begin{aligned}
            &(y,\sigma,\gamma,\boldsymbol{\delta}) \\
            &\in \mathcal Y \times [\sigma_{\min},\sigma_{\max}] \times [\gamma_{\min},\gamma_{\max}] \times \Delta
        \end{aligned}
        ~\left|~
        \begin{array}{l}
        y = (E,~I)^\top,\\
        E\geq 0,\quad I\geq 0,\\
        N - E - I - \gamma\displaystyle\int_0^\cdot I(s)\,ds\geq 0
        ,\\
        I' = \sigma E - \gamma I
        ,
        \quad\text{on}~[0,T]
        \end{array}
        \right.
    \right\}.
\end{align*}
We then define the consensus curve as the solution to the following constrained optimization problem:
\begin{align}
    (\hat y_{q}, \hat{\sigma}, \hat{\gamma}, \hat{\boldsymbol{\delta}})
    &:=
    \argmin_{\substack{
    (y,\sigma,\gamma, \boldsymbol{\delta})\in \mathcal F
    }}
    \frac{1}{J}\sum_{j=1}^{J} d\left(y,\mathcal T_{\delta_j} y_j\right)^{q}
    \label{eq:joint_shift_frechet}
\end{align}
where $q > 0$, $0<\sigma_{\min}<\sigma_{\max}<\infty$, $0<\gamma_{\min}<\gamma_{\max}<\infty$.
The objective function is formulated to yield the power Fr\'echet mean.
The constraints are constructed based on the dynamical constraints regarding the infectious population $I(t)$ and the constraints on the total population.
Additionally, the constraint of the shift parameter $\sum_{j=1}^J \delta_j = 0$ is imposed to ensure identifiability against global translation in $t$, and each $\delta_j$ is restricted to a bounded interval $[-\delta_{\max},\delta_{\max}]$ to prevent the curves from shifting outside the observation window, which may result in a pathological solution.

By estimating the representative trajectory while simultaneously aligning the curves, this approach avoids the attenuation of peaks caused by the superposition of misaligned curves, which is a common issue in simple methods such as the pointwise mean or median.
Furthermore, by jointly estimating the trajectories $E$, $I$, and the parameters $\sigma$ and $\gamma$ under these dynamical constraints, we can uniquely reconstruct the full epidemiological state $(S, E, I, R)$, as discussed in Section \ref{sec:deffuncspace}.
Additionally, under the assumption of the transmission dynamics equation ${dE}/{dt} = \beta {S I}/{N} - \sigma E$, the parameter $\beta$ can be estimated by fitting to the reconstructed trajectories $(S, E, I)$.

The constrained optimization problem \eqref{eq:joint_shift_frechet} is formulated as the extended constrained power Fr\'echet mean given in Definition \ref{def:extended_cpfm_metric}.
Therefore, we establish the existence of the minimizer of \eqref{eq:joint_shift_frechet} as a corollary of Proposition~\ref{prop:existence_general}(iii).

\begin{corollary}[Existence of the consensus curve]
\label{cor:existence_consensus_curve}
Let $q > 0$.
Then, we have $\mathcal F\neq\emptyset$ and there exists a minimizer $(\hat y_q,\hat\sigma,\hat\gamma,\hat{\boldsymbol\delta})\in\mathcal F$ of \eqref{eq:joint_shift_frechet}.
\end{corollary}

Similarly, the robustness result of Proposition~\ref{prop:robustness} can be extended to the consensus curve.

\begin{corollary}[Robustness of the consensus curve to an outlier curve under $q=1$]
\label{cor:robustness_consensus_curve}
Fix $y_1,\ldots,y_{J-1} \in \mathcal Y$ with $J\geq 3$ and take a sequence $(y_J^{(n)}|~n\in\mathbb N)\subset\mathcal Y$.
For each $n$, let $(\hat y^{(n)},\hat\sigma^{(n)},\hat\gamma^{(n)},\hat{\boldsymbol\delta}^{(n)})$ be a minimizer of \eqref{eq:joint_shift_frechet} with $q=1$ for $y_1,\ldots,y_{J-1},y_J^{(n)}$.
Then, it holds that $\sup_{n\in\mathbb N}\|\hat y^{(n)}\|_{\mathcal Y}<\infty$ even when $\|y_J^{(n)}\|_{\mathcal Y} \to \infty$.
\end{corollary}

See Appendix~\ref{sec:proof} for the Proofs of these corollaries.

\section{Algorithm Based on Techniques from Functional Data Analysis}
\label{sec:algorithm}

In this section, we develop an algorithm for \eqref{eq:joint_shift_frechet}.  
Using basis function expansion, we construct an implementable algorithm by representing trajectories.
Here, we introduce the main algorithm in 4 steps for numerically solving \eqref{eq:joint_shift_frechet} and an optional step for recovering the full trajectory.

\paragraph{Step 1: Basis Function Expansion.}
For each shifted trajectory $(\mathcal T_{\delta_j} y_j)(t)=(\mathcal T_{\delta_j} E_j,\mathcal T_{\delta_j} I_j)^\top$ with any shift parameter $\delta_j$, we approximate
\begin{align*}
    (\mathcal T_{\delta_j} E_j)(t)\approx \sum_{k=1}^K c_{jk}^{(E)}(\delta_j)\phi_k(t),
    \qquad
    (\mathcal T_{\delta_j} I_j)(t)\approx \sum_{k=1}^K c_{jk}^{(I)}(\delta_j)\phi_k(t),
\end{align*}
using basis functions $\{\phi_k\}_{k=1}^K$, and $\boldsymbol{c}_{j}^{(E)}(\delta_j) := (c_{j 1}^{(E)}(\delta_j),\ldots, c_{j K}^{(E)}(\delta_j))^{\top} \in \mathbb{R}^{K}$ and $\boldsymbol{c}_{j}^{(I)}(\delta_j) := (c_{j 1}^{(I)}(\delta_j),\ldots, c_{j K}^{(I)}(\delta_j))^{\top} \in \mathbb{R}^{K}$ are coefficients vectors.
Here, let $\boldsymbol{c}_{j}(\delta_j) := ((\boldsymbol{{c}}_{j}^{(E)}(\delta_j))^{\top}, (\boldsymbol{{c}}_{j}^{(I)}(\delta_j))^{\top})^{\top}$.
We also represent the unknown target trajectory $y=(E,I)^\top$ by
\begin{align*}
    E(t)\approx \sum_{k=1}^K c_{k}^{(E)}\phi_k(t),
    \qquad
    I(t)\approx \sum_{k=1}^K c_{k}^{(I)}\phi_k(t),
\end{align*}
with a coefficient vector $\boldsymbol{c} := ((\boldsymbol{{c}}^{(E)})^{\top}, (\boldsymbol{{c}}^{(I)})^{\top})^{\top}$, where $\boldsymbol{c}^{(E)},\boldsymbol{c}^{(I)}\in\mathbb{R}^K$.

Using this representation, the $H^1$-norm between the two trajectories is approximated as
\begin{align*}
    &d(y, \mathcal{T}_{\delta_j}y_j)
    \\
    &\quad
    \approx
    \left\{
        (\boldsymbol{c}^{(E)} - \boldsymbol{c}_{j}^{(E)}(\delta_j))^\top \mathbf{G}^{H^1} (\boldsymbol{c}^{(E)} - \boldsymbol{c}_{j}^{(E)}(\delta_j)) +  (\boldsymbol{c}^{(I)} - \boldsymbol{c}_{j}^{(I)}(\delta_j))^\top \mathbf{G}^{H^1} (\boldsymbol{c}^{(I)} - \boldsymbol{c}_{j}^{(I)}(\delta_j))
    \right\}^{1/2}
    \\
    &\quad
    = \left\{
        (\boldsymbol{c}-\boldsymbol{c}_j(\delta_j))^\top
        \mathbf{G}^{\mathcal{Y}}
        (\boldsymbol{c}-\boldsymbol{c}_j(\delta_j))
    \right\}^{1/2}
    \\
    &\quad
    =: D(\boldsymbol{c},\boldsymbol{c}_{j}(\delta_j))
    ,
\end{align*}
where $\mathbf{G}^{H^1}$ is the Gram matrix of the basis system under the $H^1$ inner product:
\begin{align*}
    \mathbf{G}^{H^1}_{k\ell}
    =
    \int_0^T \phi_k(t)\phi_{\ell}(t)\,dt + \rho \int_0^T {\phi'}_k(t){\phi'}_{\ell}(t)\,dt
    ,
\end{align*}
and
\begin{align*}
    \mathbf{G}^{\mathcal{Y}}
    :=
    \begin{pmatrix}
        \mathbf{G}^{H^1} & \mathbf{0}\\
        \mathbf{0} & \mathbf{G}^{H^1}
    \end{pmatrix}
    \in\mathbb{R}^{2K\times 2K}
    .
\end{align*}

Furthermore, the constraints in \eqref{eq:joint_shift_frechet} can be expressed as follows:
\begin{align}
    \begin{cases}
        \displaystyle \sum_{k=1}^K c_{k}^{(E)} \phi_k(t) \geq 0, \\
        \displaystyle \sum_{k=1}^K c_{k}^{(I)} \phi_k(t) \geq 0, \\
        \displaystyle N - \sum_{k=1}^K c_{k}^{(E)} \phi_k(t) - \sum_{k=1}^K c_{k}^{(I)} \phi_k(t) - \gamma \sum_{k=1}^K c_{k}^{(I)} \Phi_k(t) \geq 0, \\
        \displaystyle\sum_{k=1}^K c_{k}^{(I)} {\phi'}_k(t) = \sigma \sum_{k=1}^K c_{k}^{(E)} \phi_k(t) - \gamma \sum_{k=1}^K c_{k}^{(I)} \phi_k(t),
    \end{cases}
    \forall t \in [0,T]
    ,
    \label{eq:constraints_joint_shift_frechet}
\end{align}
where ${\phi'}_k(t) := {d\phi_k(t)}/{dt}$ and $\Phi_k(t) = \int_{0}^{t} \phi_{k}(s) ds$.

When B-splines are adopted as the basis functions, $\phi_k(t)$, their first derivatives ${\phi'}_k(t)$ and integrals $\Phi_k(t)$ are available in closed form.
Furthermore, the integrals of products of basis functions and their derivatives, and consequently the Gram matrix, can be computed analytically.
This prevents approximation errors associated with numerical differentiation and integration.

\paragraph{Step 2: Discretization of Evaluation Time.}
For implementation, we discretize the evaluation time.
Specifically, let $0=t_0<t_1<\cdots<t_M=T$ be an equally spaced grid on $[0,T]$.

Now we prepare matrix forms on this evaluation grid.
For $m=0,1,\ldots,M,\ k=1,\ldots,K$, define the design matrices
\begin{align}
    \begin{cases}
    \mathbf{B} &:= \{\phi_k(t_m)\} \in \mathbb{R}^{(M+1) \times K},
    \\
    \mathbf{B}' &:= \{{\phi'}_k(t_m)\} \in \mathbb{R}^{(M+1) \times K},
    \\
    \mathbf{\Phi} &:= \{\Phi_k(t_m)\} = \left\{\int_0^{t_m}\phi_k(s)\,ds\right\} \in \mathbb{R}^{(M+1) \times K}.
    \end{cases}
    \label{eq:bspline_design_matrices}
\end{align}
For the coefficient vector $\boldsymbol{c}=((\boldsymbol{c}^{(E)})^\top,(\boldsymbol{c}^{(I)})^\top)^\top\in\mathbb{R}^{2K}$, the grid evaluations satisfy
\begin{align*}
    \bigl(E(t_m)\bigr)_{m=0}^M
    =
    \mathbf{B}\boldsymbol{c}^{(E)},
    \quad
    \bigl(I(t_m)\bigr)_{m=0}^M
    =
    \mathbf{B}\boldsymbol{c}^{(I)},
    \quad
    \left(\int_0^{t_m} I(s)\,ds\right)_{m=0}^M
    =
    \mathbf{\Phi}\boldsymbol{c}^{(I)}.
\end{align*}

Using the matrix formulation \eqref{eq:bspline_design_matrices}, the constraints in \eqref{eq:constraints_joint_shift_frechet} evaluated on the grid can be expressed as follows:
\begin{align}
    \begin{cases}
        \displaystyle \mathbf{B}\boldsymbol{c}^{(E)} \geq \boldsymbol{0}, \\
        \displaystyle \mathbf{B}\boldsymbol{c}^{(I)} \geq \boldsymbol{0}, \\
        \displaystyle N\boldsymbol{1} - \mathbf{B}\boldsymbol{c}^{(E)} - \mathbf{B}\boldsymbol{c}^{(I)} - \gamma\,\mathbf{\Phi}\boldsymbol{c}^{(I)} \geq \boldsymbol{0}, \\
        \displaystyle \mathbf{B}'\boldsymbol{c}^{(I)} = \sigma\,\mathbf{B}\boldsymbol{c}^{(E)} - \gamma\,\mathbf{B}\boldsymbol{c}^{(I)}.
    \end{cases}
    \label{eq:bspline_constraints}
\end{align}

Therefore, we solve the following optimization problem:
\begin{align}
    (\hat{\boldsymbol{c}}, \hat{\sigma}, \hat{\gamma}, \hat{\boldsymbol{\delta}})
    &:=
    \argmin_{({\boldsymbol{c}},\sigma,\gamma,\boldsymbol{\delta}) \in \mathcal {F}_K}
    \frac{1}{J}\sum_{j=1}^{J}D(\boldsymbol{c},\boldsymbol{c}_{j}(\delta_j))^{q}
    \label{eq:joint_shift_frechet_discretized}
    ,
\end{align}
where $\mathcal {F}_K$ is the finite-dimensional feasible set defined as
\begin{align*}
    \mathcal {F}_K
    :=
    \left\{
        (\boldsymbol{c},\sigma,\gamma,\boldsymbol{\delta})
        \in
        \mathbb R^{2K}
        \times
        [\sigma_{\min},\sigma_{\max}]
        \times
        [\gamma_{\min},\gamma_{\max}]
        \times
        \Delta
        ~ \left| ~
        \eqref{eq:bspline_constraints}
        \right.
    \right\}
    .
\end{align*}

\paragraph{Step 3: Alternating Minimization.}
We solve \eqref{eq:joint_shift_frechet_discretized} by alternating updates of $(\boldsymbol{c},\sigma,\gamma)$ and $\boldsymbol{\delta}:=(\delta_{1},\ldots,\delta_{J})$.
Specifically, iterate the following (I) and (II) steps alternately until convergence.

\begin{itemize}
    \item[(I)] \emph{Update of $\boldsymbol{c},\sigma,\gamma$.}
    Given shifts $\hat{\boldsymbol{\delta}}^{(r)}$ ($r=0,1,\ldots,r_{\text{max}}$) with the maximum iteration number $r_{\text{max}}$, we obtain $\boldsymbol{c}_1(\hat{\delta}_{1}^{(r)}),\ldots,\boldsymbol{c}_J(\hat{\delta}_{J}^{(r)})$ by the basis expansion via e.g. a least-square method for $\mathcal{T}_{\hat{\delta}^{(r)}_1}y_1,\ldots,\mathcal{T}_{\hat{\delta}^{(r)}_J}y_J$, respectively.
    We then update $(\boldsymbol{c},\sigma,\gamma)$ by solving
    \begin{align}
        \label{eq:rep_update_given_shift}
        &(\hat{\boldsymbol{c}}, \hat{\sigma}, \hat{\gamma})^{(r+1)}
        :=
        \argmin_{(\boldsymbol{c}, \sigma, \gamma) \in \mathbb{R}^{2K} \times [\sigma_{\min},\sigma_{\max}] \times [\gamma_{\min},\gamma_{\max}]} \left\{ \frac{1}{J} \sum_{j=1}^J D(\boldsymbol{c}, \boldsymbol{c}_j(\hat{\delta}_{j}^{(r)}))^q \right\}
        \quad\text{s.t. \eqref{eq:bspline_constraints}} 
        .
    \end{align}
    This constrained optimization problem can be solved by a constrained nonlinear optimization method, such as sequential quadratic programming, if $q>1$.
    Otherwise, the formulation can be separated into a block optimization process for $\boldsymbol{c}$ and ($\sigma,\gamma$).
    We note that for a fixed ($\sigma,\gamma$), the constrained optimization problem for ($\boldsymbol{c}^{(E)},\boldsymbol{c}^{(I)}$) is reduced to a convex optimization problem.
    We describe details of this approach in Appendix \ref{sec:simplified_optimization}.
    \item[(II)] \emph{Update of $\boldsymbol{\delta}$.}
    Given $(\hat{\boldsymbol{c}}^{(E)}, \hat{\boldsymbol{c}}^{(I)}, \hat{\sigma}, \hat{\gamma})^{(r+1)}$, update each shift independently by the one-dimensional optimization problem
    \begin{align}
        \tilde{\delta}_j
        &:=
        \argmin_{\delta\in[-\delta_{\max},\delta_{\max}]}
        D(\hat{\boldsymbol{c}}^{(r+1)},\boldsymbol{c}_{j}({\delta}))^{q}
        ,
        \qquad j=1,\ldots,J,
        \label{eq:shift_update_1d}
    \end{align}
    which can be solved by standard 1-dimensional optimization methods such as Brent's method.
    At each evaluation point of $\delta_j$, we obtain $\boldsymbol{c}_{j}({\delta})$ by the basis expansion.
    To pose the constraint $\sum_{j=1}^J \delta_j = 0$ while satisfying $\delta_j \in [-\delta_{\max}, \delta_{\max}]$, we find an adjustment parameter $\mu_{\delta} \in \mathbb{R}$ that satisfies
    \begin{align*}
        \sum_{j=1}^{J} \max\left( -\delta_{\max}, \min\left( \delta_{\max}, \tilde{\delta}_j - \mu_{\delta} \right) \right)
        =
        0
        .
    \end{align*}
    Since the function $f(\mu_{\delta}) := \sum_{j=1}^{J} \max( -\delta_{\max}, \min( \delta_{\max}, \tilde{\delta}_j - \mu_{\delta} ) )$ is monotonically decreasing and continuous, the unique root $\mu_{\delta}$ can be efficiently found via a root-finding algorithm such as bisection method.
    Finally, we update the shift parameters by
    \begin{align*}
        \hat{\delta}_j^{(r+1)}
        \leftarrow
        \max\left( -\delta_{\max}, \min\left( \delta_{\max}, \tilde{\delta}_j - \mu_{\delta} \right) \right)
        .
    \end{align*}
\end{itemize}    

For initialization of the optimization problem \eqref{eq:rep_update_given_shift} in the first iteration ($r=0$), we solve the following reduced problem
\begin{align*}
    &\hat{\boldsymbol{c}}_{\text{(init)}}
    :=
    \argmin_{\boldsymbol{c} \in \mathbb{R}^{2K}} \frac{1}{J} \sum_{j=1}^J D(\boldsymbol{c}, \boldsymbol{c}_j(\hat{\delta}_j^{(0)}))^q
    \quad
    \text{s.t.}
    \quad
    \begin{cases}
        \displaystyle \mathbf{B}\boldsymbol{c}^{(E)} \geq \boldsymbol{0}, \\
        \displaystyle \mathbf{B}\boldsymbol{c}^{(I)} \geq \boldsymbol{0}, \\
        \displaystyle N\boldsymbol{1} - \mathbf{B}\boldsymbol{c}^{(E)} - \mathbf{B}\boldsymbol{c}^{(I)} \geq \boldsymbol{0}
        .
    \end{cases}
\end{align*}
Given these coefficients, we obtain the initial estimates of $(\hat{\sigma}_{\text{(init)}}, \hat{\gamma}_{\text{(init)}})$ as a solution to a least-squares problem for the differential equation ${dI}/{dt} = \sigma E - \gamma I$ over $t \in [0, T]$, i.e.,
\begin{align*}
    &(\hat{\sigma}_{\text{(init)}}, \hat{\gamma}_{\text{(init)}})
    \\
    &\quad
    :=
    \argmin_{(\sigma, \gamma) \in [\sigma_{\min},\sigma_{\max}] \times [\gamma_{\min},\gamma_{\max}]}
    \sum_{m=0}^{M}
    \left(\sum_{k=1}^K \hat{c}_{k \text{(init)}}^{(I)} {\phi'}_k(t_m) - \sigma \sum_{k=1}^K \hat{c}_{k \text{(init)}}^{(E)} \phi_k(t_m) + \gamma \sum_{k=1}^K \hat{c}_{k \text{(init)}}^{(I)} \phi_k(t_m)\right)^2
    .
\end{align*}
We take the initial shifts $\hat{\delta}_j^{(0)}$ for \eqref{eq:shift_update_1d} by the alignment of curves on the first most intensive peak and center them to satisfy
$\sum_{j=1}^J\hat{\delta}_j^{(0)}=0$.

\paragraph{Step 4: Recovery of the Full Trajectory.}

Let $(\hat{\boldsymbol{c}}^{(E)}, \hat{\boldsymbol{c}}^{(I)}, \hat{\sigma}, \hat{\gamma})$ be the final solution.
The exposed and infectious compartments are recovered from the estimated coefficients as:
\begin{align*}
    \hat{E}(t) = \sum_{k=1}^K \hat{c}_{k}^{(E)} \, \phi_k(t)
    ,
    \quad
    \hat{I}(t) = \sum_{k=1}^K \hat{c}_{k}^{(I)} \, \phi_k(t)
    .
\end{align*}
The full trajectory is recovered from $(\hat{E}, \hat{I})$ and $\hat{\gamma}$ via
\begin{align*}
    \hat{R}(t)
    = \hat{\gamma} \int_0^t \hat{I}(s)\,ds
    ,
    \quad
    \hat{S}(t) = N - \hat{E}(t) - \hat{I}(t) - \hat{R}(t)
    .
\end{align*}

\paragraph{Optional Step: Estimation of Transmission Rate.}

Furthermore, by assuming the SEIR model and constant transmission rate $\beta$, we estimate $\beta$ using the integral form of the differential equation to avoid numerical instability.
Integrating the equation
\begin{align*}
    \frac{dE(t)}{dt} = \beta \frac{S(t) I(t)}{N} - \sigma E(t).
\end{align*}
yields:
\begin{align*}
    \hat{E}(t) - \hat{E}(0)
    =
    \beta \int_{0}^{t} \frac{\hat{S}(s) \hat{I}(s)}{N} ds
    - \hat{\sigma} \int_{0}^{t} \hat{E}(s) ds.
\end{align*}
Accordingly, we estimate $\beta$ by solving the following optimization problem:
\begin{align*}
    \hat{\beta} := \argmin_{\beta\in [\beta_{\text{min}},\beta_{\text{max}}]} \sum_{m=0}^{M} \left(
    \hat{E}(t_m) - \hat{E}(0)
    + \hat{\sigma} \int_{0}^{t_m} \hat{E}(\tau) d\tau
    - \beta \int_{0}^{t_m} \frac{\hat{S}(\tau)\hat{I}(\tau)}{N} d\tau
    \right)^2
    ,
\end{align*}
where $0<\beta_{\text{min}}<\beta_{\text{max}}<\infty$.

\section{Simulation Study with Parameters under Uncertainty}
\label{sec:simulation}

In this section, we conduct a simulation to illustrate the behavior of our estimator using SEIR trajectories generated under varying model parameters.

\subsection{Simulation Settings}

We simulated $J = 5$, $10$, and $50$ SEIR trajectories over the interval $[0, T]$ with $T = 720$.
We set $N=1,000,000$.
For each trajectory, the incubation period $D_E$, infectious period $D_I$, and basic reproductive number $R_0$ were sampled independently from normal distributions derived from published 95\% confidence intervals reported by \cite{li2020early}.
The estimated incubation period was 5.2 days (95\% confidence interval (CI): 4.1--7.0 days), the estimated infectious period was 7.5 days (95\% CI: 5.3--19.0 days), and the estimated basic reproductive number $R_0$ was 2.2 (95\% CI: 1.4--3.9).

Let $\hat\theta$ denote a reported point estimate with a 95\% confidence interval $[L,U]$.
To accommodate asymmetric uncertainty, we define left and right standard deviations
\begin{align*}
    S^{L}
    := \frac{\hat\theta - L}{1.96}
    ,
    \quad
    S^{U}
    := \frac{U - \hat\theta}{1.96},
\end{align*}
and sampled $\theta$ as the two-piece normal random numbers generated as
\begin{align*}
    Z \sim \mathcal{N}(0,1)
    ,
    \quad
    \theta
    :=
    \begin{cases}
        \hat\theta + S^{L} Z, & Z<0,\\
        \hat\theta + S^{U} Z, & Z\geq 0.
    \end{cases}
\end{align*}
By this formulation, the 2.5th and 97.5th percentiles of $\theta$ match $L$ and $U$, respectively, while the median was fixed at $\hat\theta$.
Using the reported confidence intervals, we set:
\begin{align*}
    D_E &\sim \mathrm{TPN}\left(5.2;\, S_{E}^{L}, S_{E}^{U}\right),
    \quad (S_{E}^{L}, S_{E}^{U}) = (0.5612,\, 0.9184),
    \\
    D_I &\sim \mathrm{TPN}\left(7.5;\, S_{I}^{L}, S_{I}^{U}\right),
    \quad (S_{I}^{L}, S_{I}^{U}) = (1.1224,\, 5.8673),
    \\
    R_0 &\sim \mathrm{TPN}\left(2.2;\, S_{R}^{L}, S_{R}^{U}\right),
    \quad (S_{R}^{L}, S_{R}^{U}) = (0.4082,\, 0.8673),
\end{align*}
where $\mathrm{TPN}(\mu; S_L, S_U)$ denotes the above two-piece normal sampling rule. 
If a negative value was sampled, it was replaced with $10^{-4}$.
Each parameter set $(\beta_j, \sigma_j, \gamma_j)$ was then computed deterministically from these sampled epidemiological quantities as
\begin{align*}
    \sigma_j = \frac{1}{D_{E,j}}
    ,
    \quad
    \gamma_j = \frac{1}{D_{I,j}}
    ,
    \quad
    \beta_j = R_{0,j} \cdot \gamma_j
    .
\end{align*}
The initial conditions $(E_0, I_0) = (1,\,0)$ were fixed across all simulations. Each trajectory is obtained by numerically integrating the SEIR model over the time interval $[0, T]$ using an adaptive-step Runge--Kutta method via \texttt{scipy}'s \texttt{solve\_ivp} routine.
For each set of sampled parameters $(\beta_j, \sigma_j, \gamma_j)$, the full state vector $(S_j(t), E_j(t), I_j(t), R_j(t))$ was computed over a common grid of evaluation times.

We then applied the power Fr\'echet mean estimation algorithm described in the previous section, employing a B-spline basis with $K=30$ basis functions of degree $3$ for smoothing and representing each trajectory in functional form.
The alternating minimization algorithm was terminated when the maximum relative change in parameters dropped below $10^{-3}$.
The consensus trajectory was estimated in the reduced space $\mathcal{Y}$, and subsequently reconstructed to recover the full system for the case of $J=10$, based on the Fr\'echet median curves.

The experiments were repeated $1,000$ times.
We show typical figures of consensus curves among the repeated experiments and report the mean and Monte Carlo standard deviations for the estimated parameters.
We provide sensitivity analyses results for other settings of hyperparameters $K$ and $\rho$ in Appendix \ref{sec:sensanal}.

\subsection{Results}

Figure~\ref{fig:trajectories_EI} shows the individual SEIR trajectories (gray lines) together with the Fr\'echet mean ($q=2$, blue line), Fr\'echet $1.5$-mean ($q=1.5$, purple line), and the Fr\'echet median ($q=1$, red line), along with the pointwise mean (green) and the pointwise median (orange).
Note that the pointwise methods show multiple peaks and attenuation caused by misalignment peaks, whereas the Fr\'echet approaches capture a single peak with plausible magnitude and timing.

In Figure~\ref{fig:trajectories_FULL}, the curve reconstructed from the median ($q=1$) of $E$ and $I$ is compared with the SEIR curves fitted using point estimates reported in the literature (dashed lines): $\beta = 2.2/7.5 \simeq 0.2933$, $\sigma = 1/5.2 \simeq 0.1923$, and $\gamma = 1/7.5 \simeq 0.1333$.
Furthermore, as shown in Table~\ref{tab:estimated_params_extended}, the estimated epidemiological quantities from the consensus trajectories were plausible.
This demonstrates the utility of our method in producing interpretable summaries.
The input infectious period $D_I$ followed a right-skewed distribution characterized by an asymmetric confidence interval (5.3 -- 19.0 days).
As shown in Table~\ref{tab:estimated_params_extended}, the standard Fr\'echet mean ($q=2$) tended to yield a longer estimate for the infectious period ($1/\gamma$) than other values of $q$.
In contrast, the Fr\'echet median ($q=1$) provides estimates closer to the input median of 7.5 days.
This suggests that for real-world data containing outliers or asymmetric distributions, the $q=1$ approach may be more robust and appropriate for determining a representative value.
Furthermore, the estimated consensus curves and estimated parameters demonstrate stability regardless of the sample size ($J$).
This implies that our method is expected to provide useful summaries even in situations where the number of available or reliable curves is limited, such as during the early stages of an emerging infectious disease outbreak.
Overall, the simulation study shows that our method successfully integrates trajectories in a representative manner.

\begin{figure}[htbp]
    \centering
    \begin{subfigure}[t]{0.48\textwidth}
        \centering
        \includegraphics[width=\textwidth]{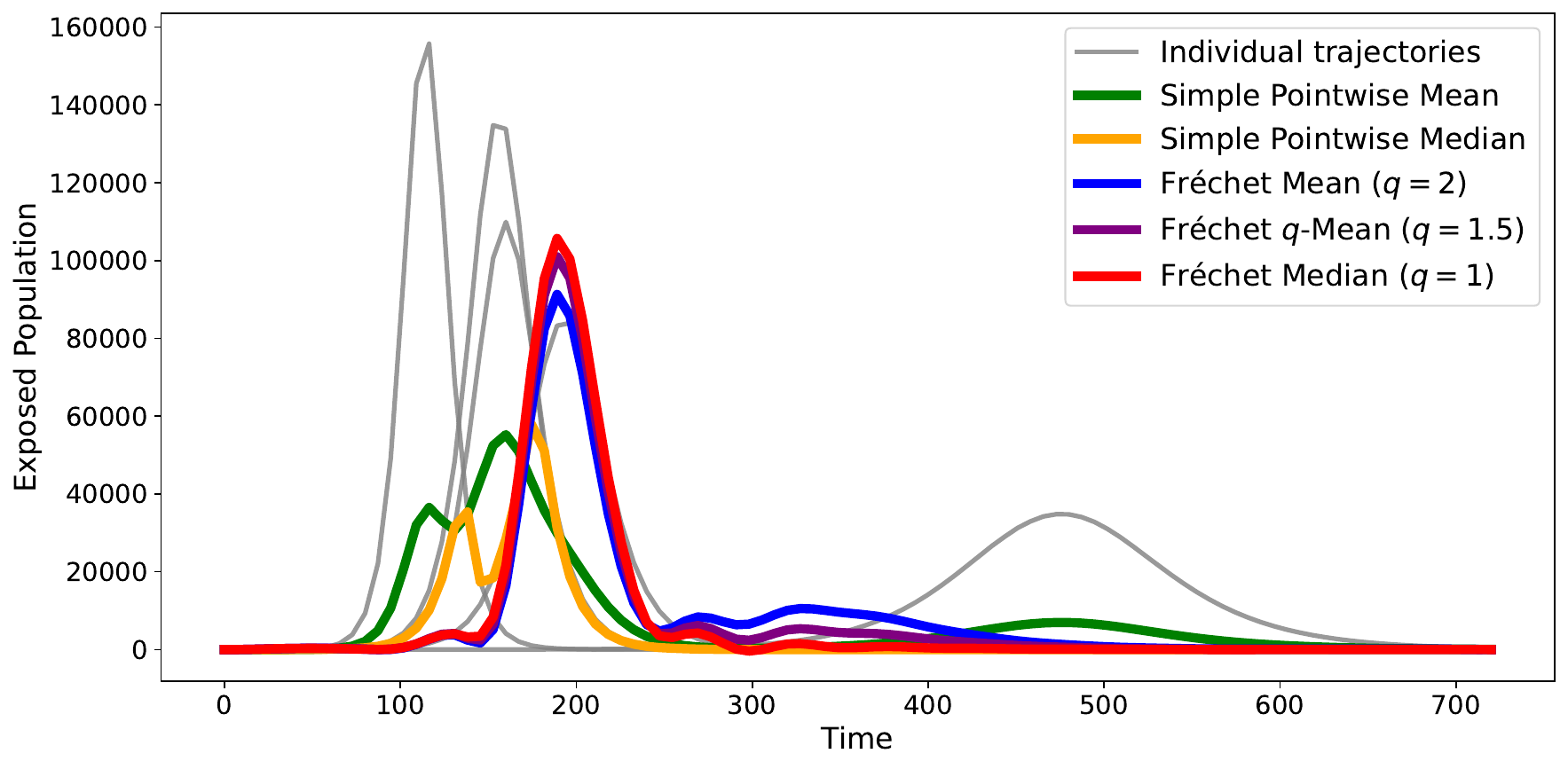}
        \caption{Exposed population ($J=5$)}
        \label{fig:trajectories_E}
    \end{subfigure}
    \hfill
    \begin{subfigure}[t]{0.48\textwidth}
        \centering
        \includegraphics[width=\textwidth]{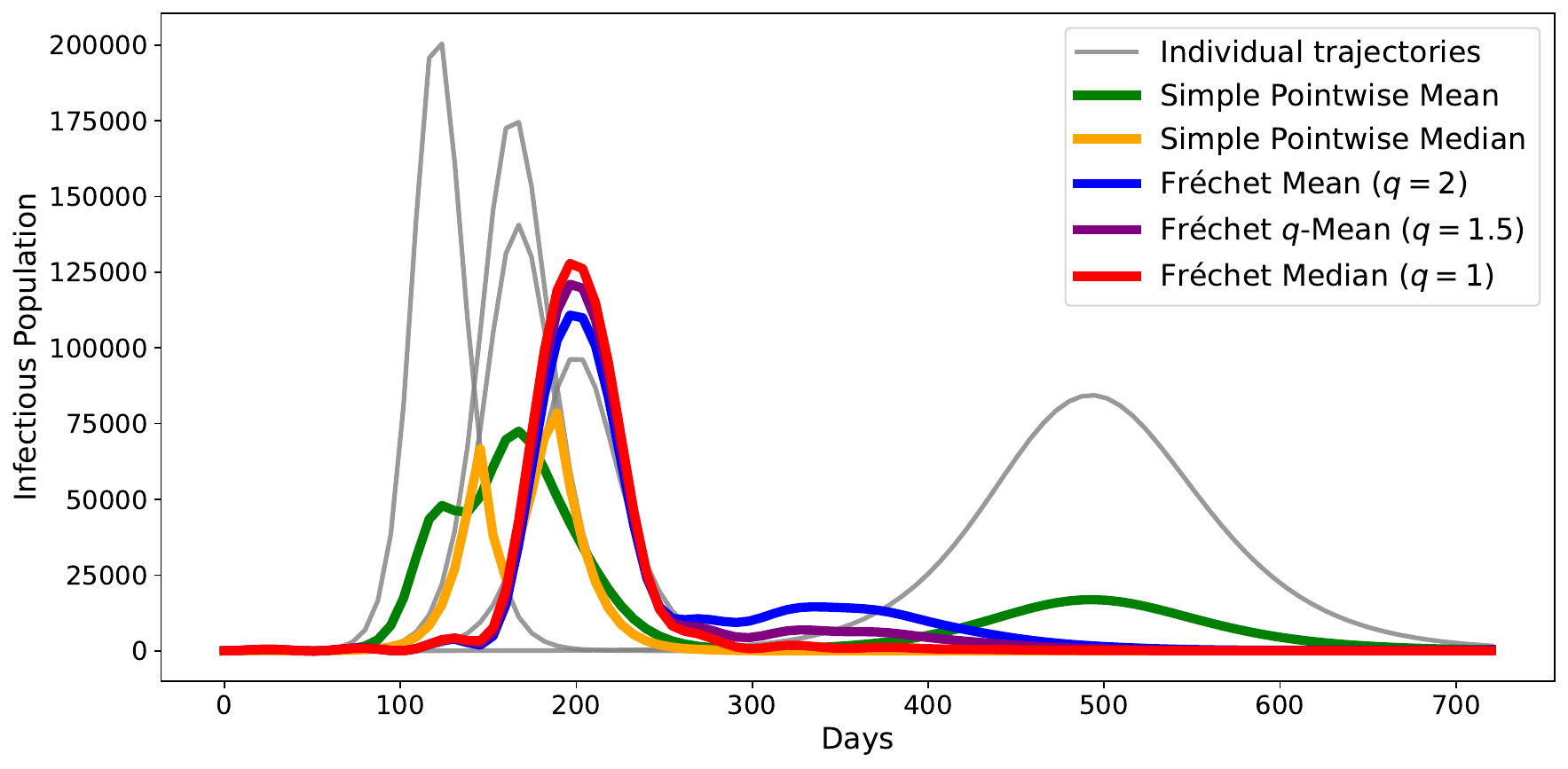}
        \caption{Infectious population ($J=5$)}
        \label{fig:trajectories_I}
    \end{subfigure}
    \hfill
    \begin{subfigure}[t]{0.48\textwidth}
        \centering
        \includegraphics[width=\textwidth]{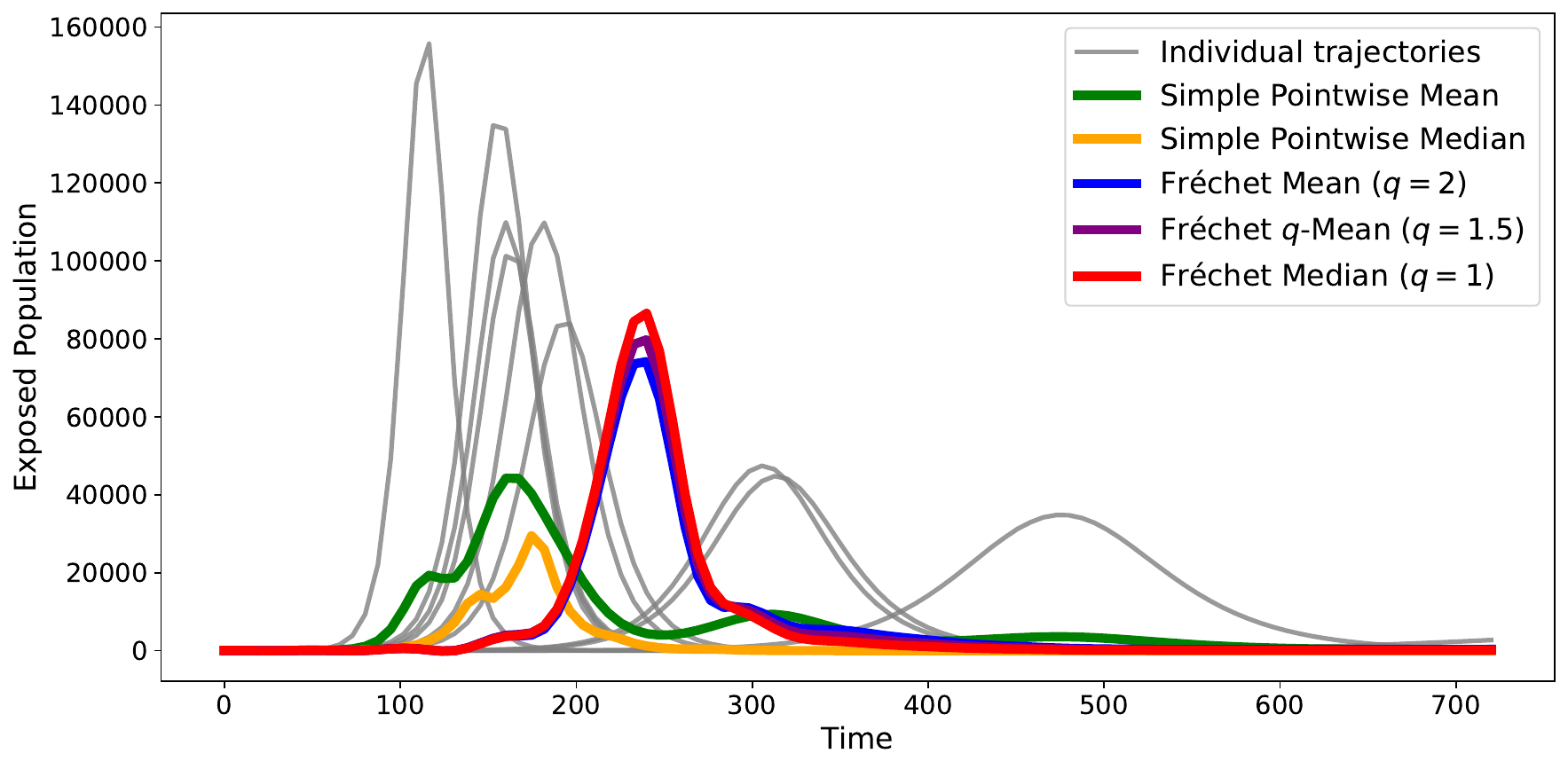}
        \caption{Exposed population ($J=10$)}
        \label{fig:trajectories_E_10}
    \end{subfigure}
    \hfill
    \begin{subfigure}[t]{0.48\textwidth}
        \centering
        \includegraphics[width=\textwidth]{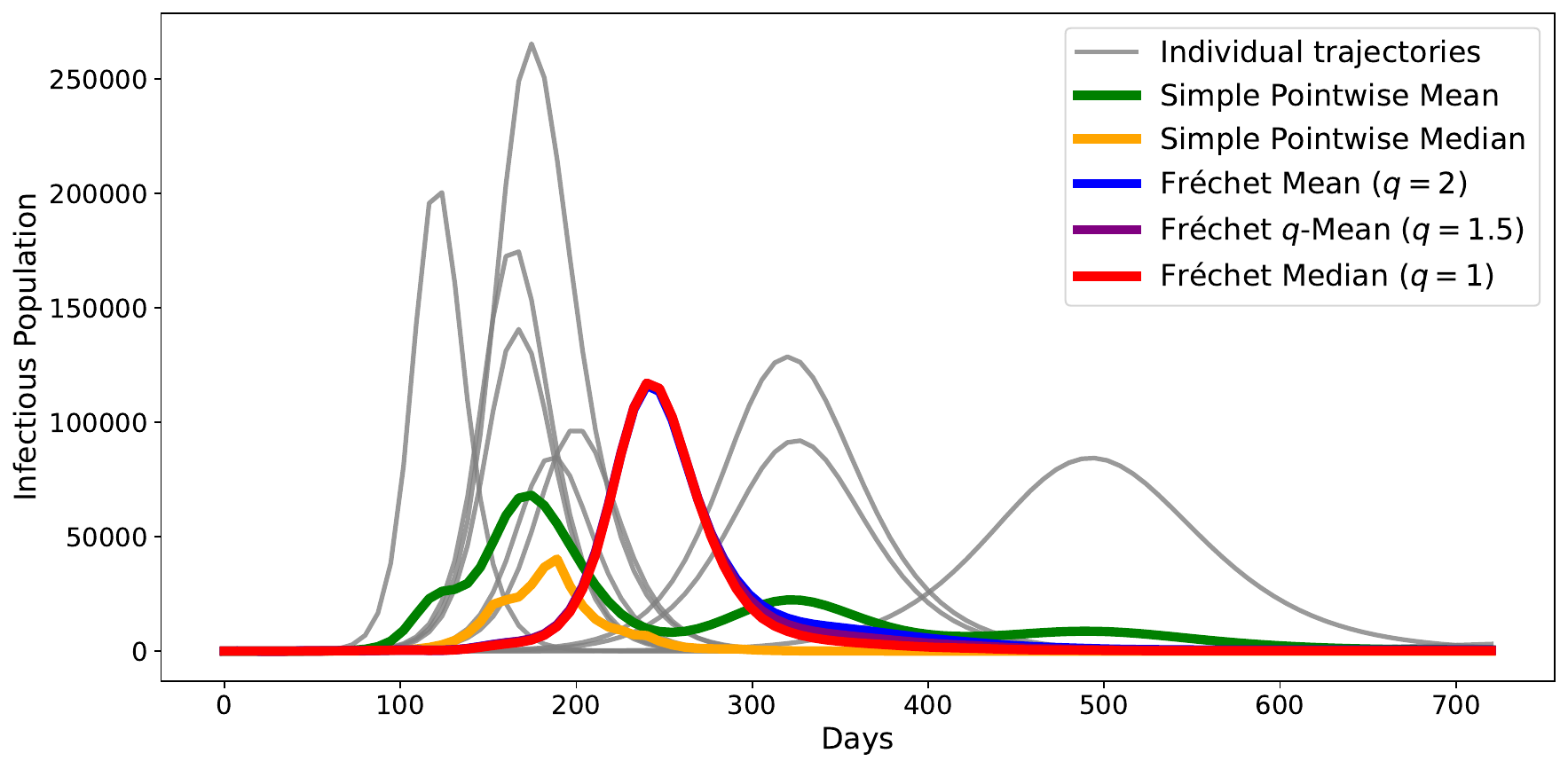}
        \caption{Infectious population ($J=10$)}
        \label{fig:trajectories_I_10}
    \end{subfigure}
    \hfill
    \begin{subfigure}[t]{0.48\textwidth}
        \centering
        \includegraphics[width=\textwidth]{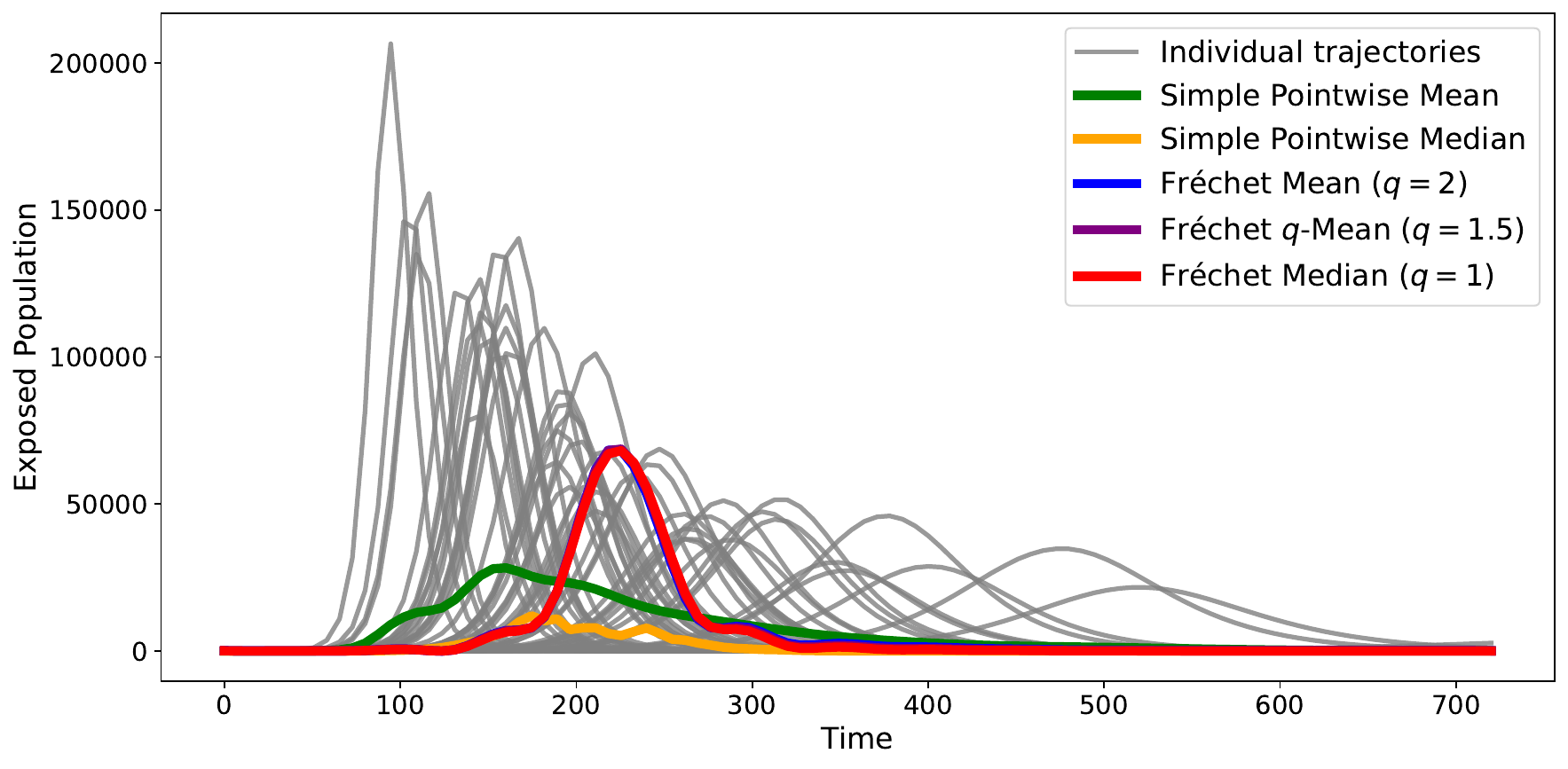}
        \caption{Exposed population ($J=50$)}
        \label{fig:trajectories_E_50}
    \end{subfigure}
    \hfill
    \begin{subfigure}[t]{0.48\textwidth}
        \centering
        \includegraphics[width=\textwidth]{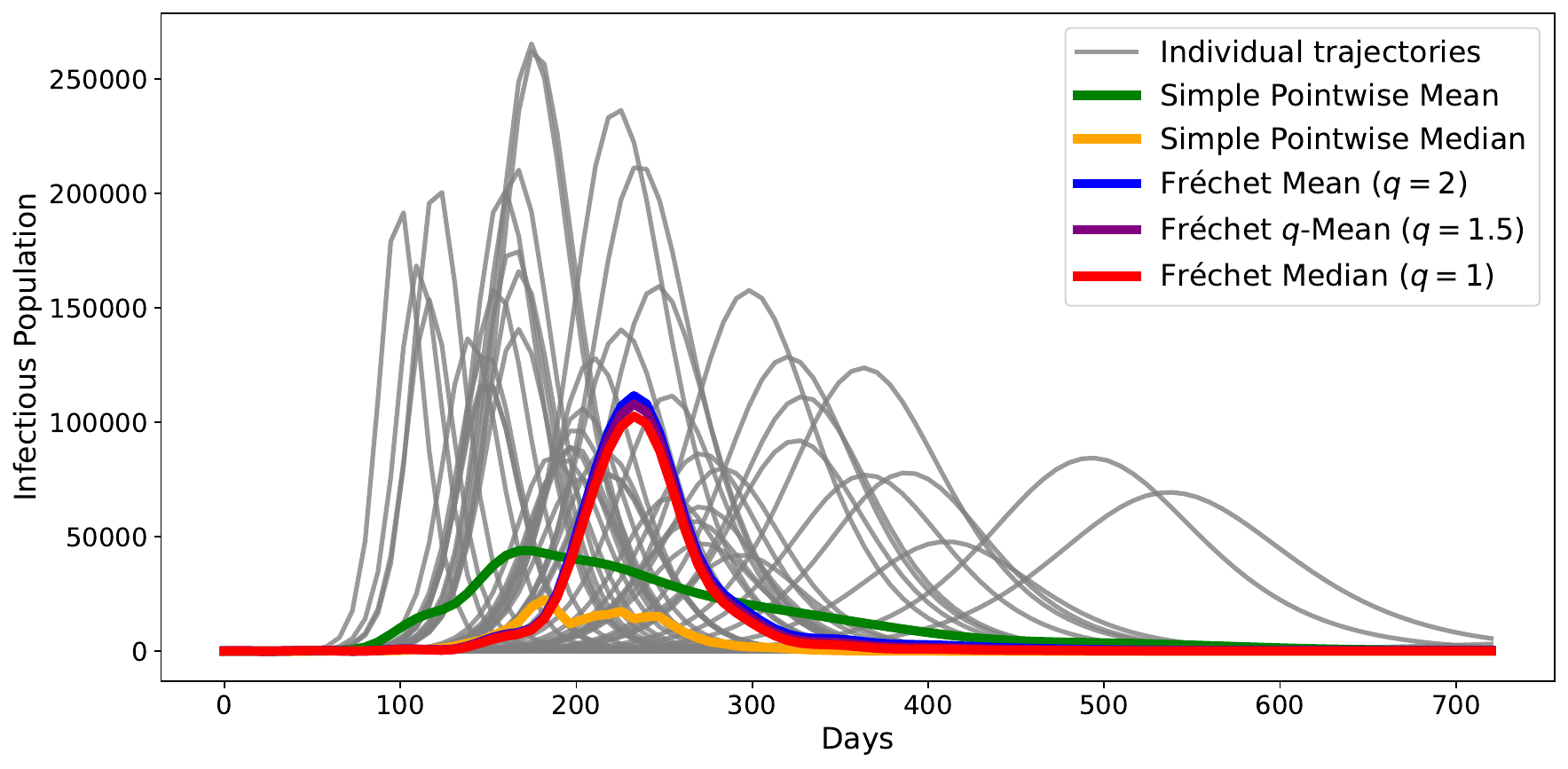}
        \caption{Infectious population ($J=50$)}
        \label{fig:trajectories_I_50}
    \end{subfigure}
    \caption{Simulated SEIR trajectories (gray) with the Fr\'echet mean (blue), the Fr\'echet $1.5$-mean (purple), the Fr\'echet median (red), the simple pointwise mean (green), and the simple pointwise median (orange). The power Fr\'echet means were calculated under $K=30$ and $\rho=1$.}
    \label{fig:trajectories_EI}
\end{figure}

\begin{figure}[htbp]
    \centering
    \includegraphics[width=0.8\textwidth]{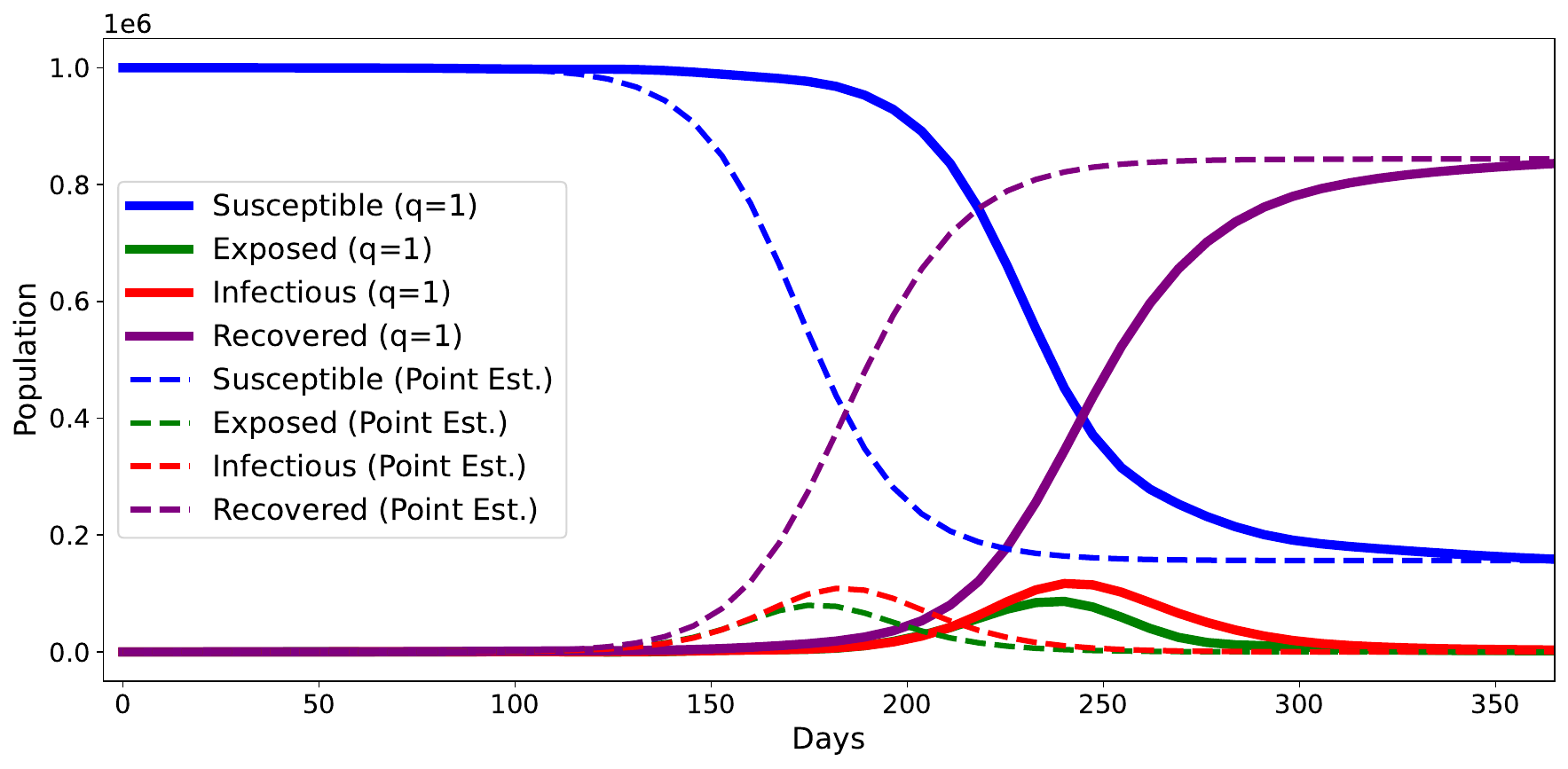}
    \caption{Full trajectory via Fr\'echet median (bold lines) and the ordinary SEIR model using point estimates of the parameters (dashed lines) for $J=10$. The power Fr\'echet means were calculated under $K=30$ and $\rho=1$.}
    \label{fig:trajectories_FULL}
\end{figure}

\begin{table}[htbp]
\centering
\caption{Estimated parameters and transmission rate under the SEIR model assumption from the consensus trajectories under different settings of $J$ and $q$. The power Fr\'echet means were calculated under $K=30$ and $\rho=1$.}
\label{tab:estimated_params_extended}
{\small
\begin{tabular}{cccccccc}
\toprule
$J$ & $q$ & $\sigma$ & $\gamma$ & Incubation ($1/\sigma$) & Infectious ($1/\gamma$) & $\beta$ & $R_0 = \beta/\gamma$\\
\midrule
5 & 1 & \makecell{0.198~(0.023)} & \makecell{0.128~(0.028)} & \makecell{5.114~(0.681)} & \makecell{8.235~(2.252)} & \makecell{0.305~(0.092)} & \makecell{2.349~(0.404)} \\
5 & 1.5 & \makecell{0.197~(0.024)} & \makecell{0.122~(0.027)} & \makecell{5.172~(0.743)} & \makecell{8.647~(2.304)} & \makecell{0.290~(0.091)} & \makecell{2.333~(0.393)} \\
5 & 2 & \makecell{0.195~(0.027)} & \makecell{0.118~(0.028)} & \makecell{5.227~(0.855)} & \makecell{9.017~(2.469)} & \makecell{0.278~(0.091)} & \makecell{2.317~(0.396)} \\
10 & 1 & \makecell{0.200~(0.018)} & \makecell{0.128~(0.021)} & \makecell{5.049~(0.502)} & \makecell{8.048~(1.524)} & \makecell{0.302~(0.075)} & \makecell{2.334~(0.317)} \\
10 & 1.5 & \makecell{0.197~(0.019)} & \makecell{0.121~(0.021)} & \makecell{5.118~(0.557)} & \makecell{8.507~(1.630)} & \makecell{0.284~(0.073)} & \makecell{2.316~(0.310)} \\
10 & 2 & \makecell{0.196~(0.021)} & \makecell{0.116~(0.021)} & \makecell{5.177~(0.642)} & \makecell{8.941~(1.793)} & \makecell{0.269~(0.072)} & \makecell{2.291~(0.312)} \\
50 & 1 & \makecell{0.202~(0.010)} & \makecell{0.131~(0.011)} & \makecell{4.952~(0.240)} & \makecell{7.687~(0.644)} & \makecell{0.307~(0.043)} & \makecell{2.333~(0.177)} \\
50 & 1.5 & \makecell{0.200~(0.010)} & \makecell{0.123~(0.010)} & \makecell{5.021~(0.253)} & \makecell{8.182~(0.693)} & \makecell{0.286~(0.041)} & \makecell{2.314~(0.171)} \\
50 & 2 & \makecell{0.197~(0.011)} & \makecell{0.116~(0.010)} & \makecell{5.084~(0.278)} & \makecell{8.676~(0.759)} & \makecell{0.266~(0.039)} & \makecell{2.280~(0.171)} \\
\bottomrule
\end{tabular}
}
\end{table}

\section{Application: Summarizing Trajectories Based on Parameters from Multiple Research Groups}
\label{sec:application}

In this section, we applied the proposed method to obtain consensus epidemic trajectories from multiple research groups, using parameter sets reported by different research groups during the early phase of the COVID-19 pandemic in 2020.
The aim is to show that, when an analyst is faced with different parameter sets chosen by independent groups, the proposed method can summarize them into a single, mechanistically interpretable consensus trajectory.

\subsection{Literature Selection and Parameter Extraction}

We first searched PubMed on April 30, 2026, using the following query:
\begin{quote}\small
\texttt{((COVID-19[TIAB] OR SARS-CoV-2[TIAB] OR 2019-nCoV[TIAB])
AND (SEIR[TIAB] OR SEIUR[TIAB] OR SEIQR[TIAB] OR SEIRD[TIAB]
OR "susceptible exposed infectious"[TIAB])
AND (model[TIAB] OR modelling[TIAB] OR modeling[TIAB]))
AND ("2020/01/01"[Date - Publication] : "2020/7/31"[Date - Publication]).}
\end{quote}
We then selected $J=6$ studies that meet the following criteria: (i) the study employs an SEIR-type compartmental model with at least the exposed and infectious compartments; (ii) the study targets COVID-19 transmission dynamics in human populations during the early phase of the pandemic; (iii) the incubation period $D_E=1/\sigma$, the infectious period $D_I=1/\gamma$, and the basic reproduction number $R_0$ are explicitly written or can be derived from the reported parameter values; and (iv) the study does not share the authors with any other selected study.
It should be noted that these $J=6$ studies do not necessarily represent an exhaustive collection of all research satisfying these criteria.
The selected studies are summarized in Table~\ref{tab:realdata_studies}.
 
For each study $j=1,\ldots,J$, we computed the rate parameters $\sigma_j = 1/D_{E,j}$, $\gamma_j = 1/D_{I,j}$, and $\beta_j = R_{0,j}\cdot\gamma_j$, and then numerically derived the solution to the corresponding SEIR solutions on $[0, T]$ with $T = 720$ days, total population $N=1,000,000$, and common initial conditions $(E_0, I_0) = (1, 0)$, using \texttt{scipy}'s \texttt{solve\_ivp} routine.

\begin{table}[htbp]
\centering
\caption{Parameter sets extracted from studies for early-phase COVID-19 (2020). $D_E = 1/\sigma$ is the incubation period, $D_I = 1/\gamma$ is the infectious period, and $R_0$ is the basic reproduction number reported by the cited study.}
\label{tab:realdata_studies}
{\small
\setlength{\tabcolsep}{4pt}
\begin{tabular}{ccccccc}
\toprule
Study & \makecell[c]{Published\\month} & Region & \makecell[c]{Model} & \makecell[c]{$D_E$\\(days)} & \makecell[c]{$D_I$\\(days)} & $R_0$ \\
\midrule
\cite{wu2020nowcasting}& February       & Wuhan, China    & SEIR & 6.00 &  2.40 & 2.68 \\
\cite{tang2020estimation}& February     & Wuhan, China    & \makecell[l]{SEIR with quarantine} & 7.00 &  2.16$^{\star}$ & 6.47 \\
\cite{peirlinck2020outbreak}& April  & United States    & SEIR                   & 2.56 & 17.82 & 5.30 \\
\cite{carcione2020simulation}& May & Lombardy, Italy & SEIRD                  & 4.25 &  4.02 & 3.00$^{\dagger}$ \\
\cite{wan2020seir}& June             & Wuhan, China    & SEIR                   & 3.00 & 14.00 & 1.44 \\
\cite{dagpunar2020sensitivity}& July & United Kingdom  & SEIR                   & 4.50 &  3.80 & 3.18 \\
\bottomrule
\end{tabular}
\\
$\star$\,2.16 was calculated by the inverse of the total removal rates from the symptomatic infectious population; $\dagger$\,$3.00$ was obtained as the initial $R_0$ value of the scheduling in simulation in the study.
}
\end{table}

\subsection{Application of the Proposed Method}

We computed the constrained power Fr\'echet means described in Section~\ref{sec:algorithm} for the $J=6$ trajectories.
We employed a B-spline basis with $K=60$ basis functions of degree $3$, the metric scaling $\rho=1$, and a maximum shift of $\delta_{\max}=120$ days, which is approximately $17\%$ of the observation window.
The alternating minimization algorithm was terminated when the maximum relative change in parameters dropped below $10^{-3}$.
We computed the consensus trajectory for $q\in\{1, 1.5, 2\}$, corresponding to the Fr\'echet median, an intermediate $q$-mean, and the standard Fr\'echet mean, respectively.
The full $(S, E, I, R)$ trajectory was then reconstructed from the estimated as described in Section~\ref{sec:algorithm}.

\subsection{Results}

Figure~\ref{fig:realdata_EI} shows the $J=6$ individual SEIR trajectories (gray) for the exposed and infectious compartments, together with the simple pointwise mean (green), the simple pointwise median (orange), the Fr\'echet mean ($q=2$, blue), the Fr\'echet $1.5$-mean ($q=1.5$, purple), and the Fr\'echet median ($q=1$, red).
The individual trajectories were highly heterogeneous: peak times ranged from approximately $50$ days for high-$R_0$ studies (e.g., \cite{tang2020estimation} with $R_0=6.47$) to more than $500$ days for the slowly evolving scenario ($R_0=1.44$) of \cite{wan2020seir}.
Furthermore, the peak heights varied by an order of magnitude across studies.
As in the simulation example, the simple pointwise mean and median were attenuated and exhibited multiple peaks, because the peaks across studies occur at very different times.
In contrast, the three Fr\'echet representative curves exhibited single epidemic waves with the peaks around day $130$, with magnitudes that were much closer to those of the typical individual trajectories.

The estimated representative parameters are reported in Table~\ref{tab:realdata_estimated}.
For the Fr\'echet median ($q=1$), the implied incubation period was approximately $4.43$ days and the implied infectious period was approximately $3.61$ days, yielding $R_0 \simeq 3.25$.
These values were consistent with the estimates from a systematic review of \cite{allen2021early}, who reported a mean incubation period of $4$--$7$ days across early COVID-19 studies and median $R_0$ estimates clustered around $2$--$3$.
As in the simulation, the standard Fr\'echet mean ($q=2$) yielded a longer estimate of the infectious period (approximately $7.11$ days), possibly reflecting the influence of the right-skewed input distribution induced by the long $D_I=17.82$ days reported by \cite{peirlinck2020outbreak}.
The Fr\'echet median may therefore arguably be a more robust summary in this setting, and the full trajectory reconstructed from the Fr\'echet median curves is shown in Figure~\ref{fig:realdata_FULL}.

\begin{table}[htbp]
\centering
\caption{Estimated parameters and the implied transmission rate from the SEIR trajectories from literature-derived parameter values ($J=6$) of early-phase COVID-19 studies, for different choices of the power $q$. The power Fr\'echet means were calculated with $K=60$ and $\rho=1$.}
\label{tab:realdata_estimated}
{\small
\begin{tabular}{cccccccc}
\toprule
$J$ & $q$ & $\sigma$ & $\gamma$ & Incubation ($1/\sigma$) & Infectious ($1/\gamma$) & $\beta$ & $R_0 = \beta/\gamma$\\
\midrule
6 & 1.0    &  0.2257   & 0.2771   & 4.43     & 3.61 &   0.9013   & 3.25      \\
6 & 1.5    &  0.1872   & 0.2249   & 5.34     & 4.45 &   0.5779   & 2.57      \\
6 & 2.0    &  0.1486   & 0.1406   & 6.73     & 7.11 &   0.2889   & 2.06      \\
\bottomrule
\end{tabular}
}
\end{table}

\begin{figure}[htbp]
    \centering
    \begin{subfigure}[t]{0.8\textwidth}
        \centering
        \includegraphics[width=\textwidth]{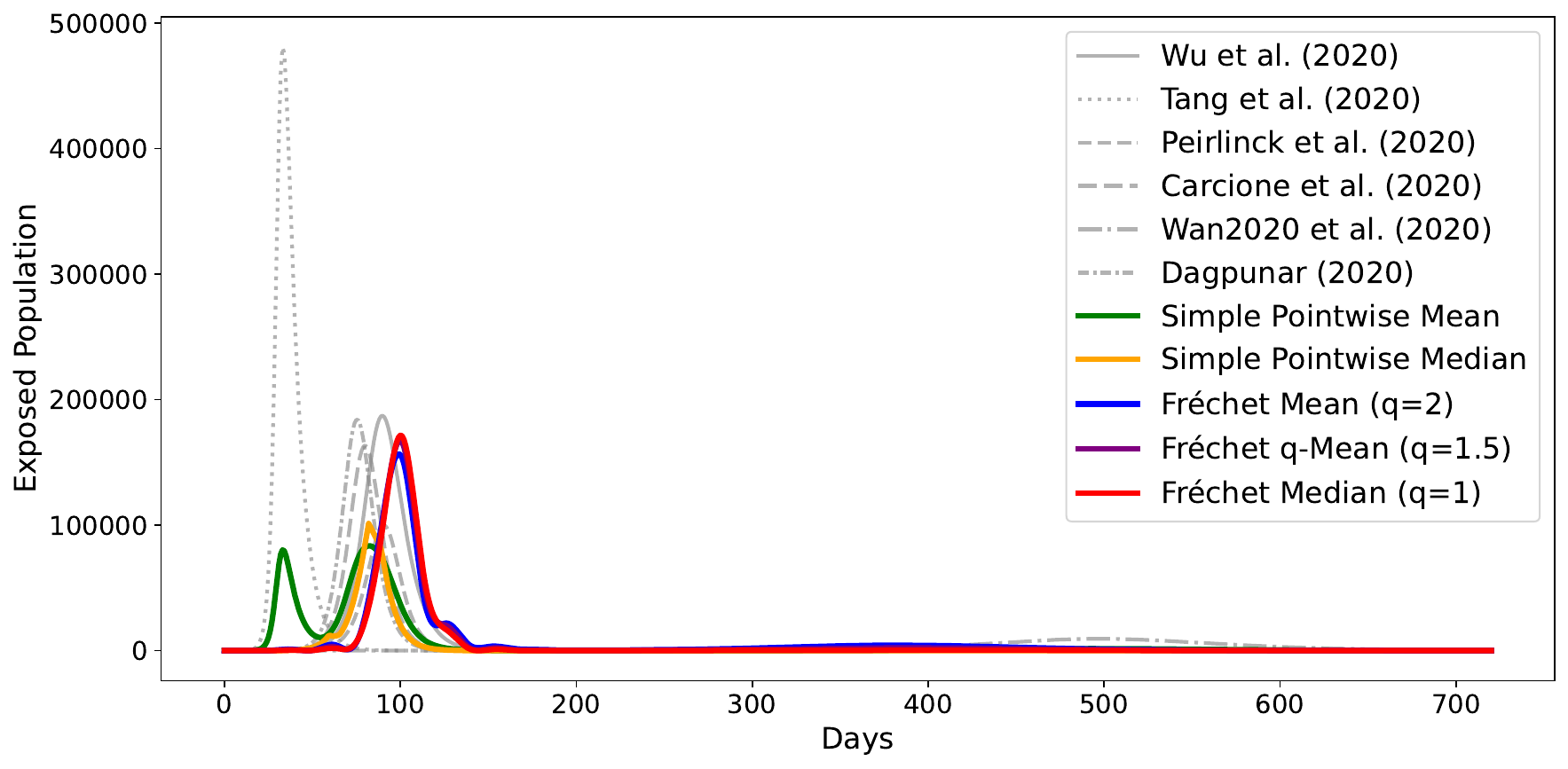}
        \caption{Exposed population}
        \label{fig:realdata_E}
    \end{subfigure}
    \\
    \begin{subfigure}[t]{0.8\textwidth}
        \centering
        \includegraphics[width=\textwidth]{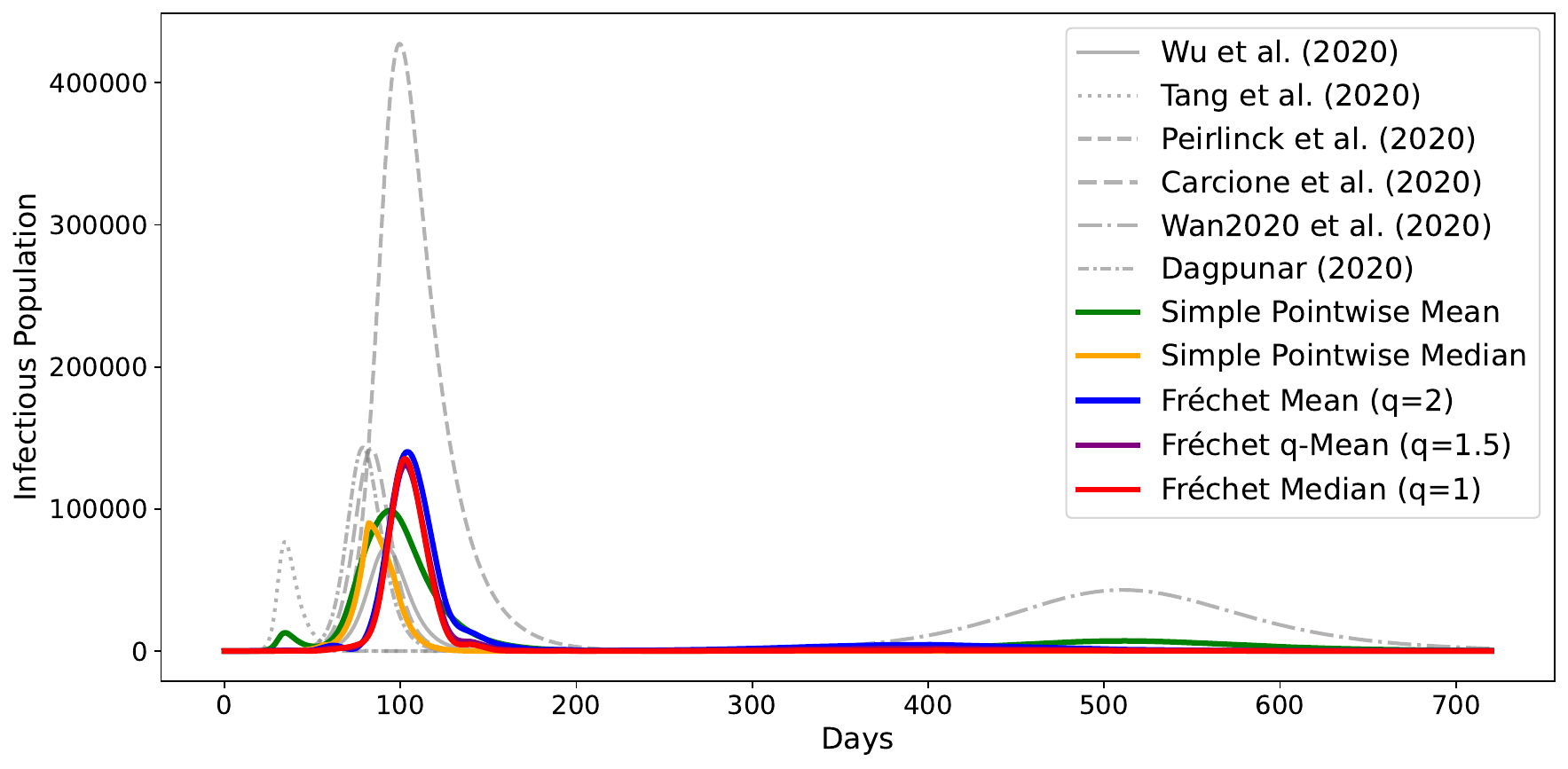}
        \caption{Infectious population}
        \label{fig:realdata_I}
    \end{subfigure}
    \caption{SEIR trajectories from literature-derived parameter values (gray, $J=6$) for the early phase of COVID-19 in 2020, together with the Fr\'echet mean (blue), the Fr\'echet $1.5$-mean (purple), the Fr\'echet median (red), the simple pointwise mean (green), and the simple pointwise median (orange). The power Fr\'echet means were computed under $K=60$ and $\rho=1$.}
    \label{fig:realdata_EI}
\end{figure}
 
\begin{figure}[htbp]
    \centering
    \includegraphics[width=0.8\textwidth]{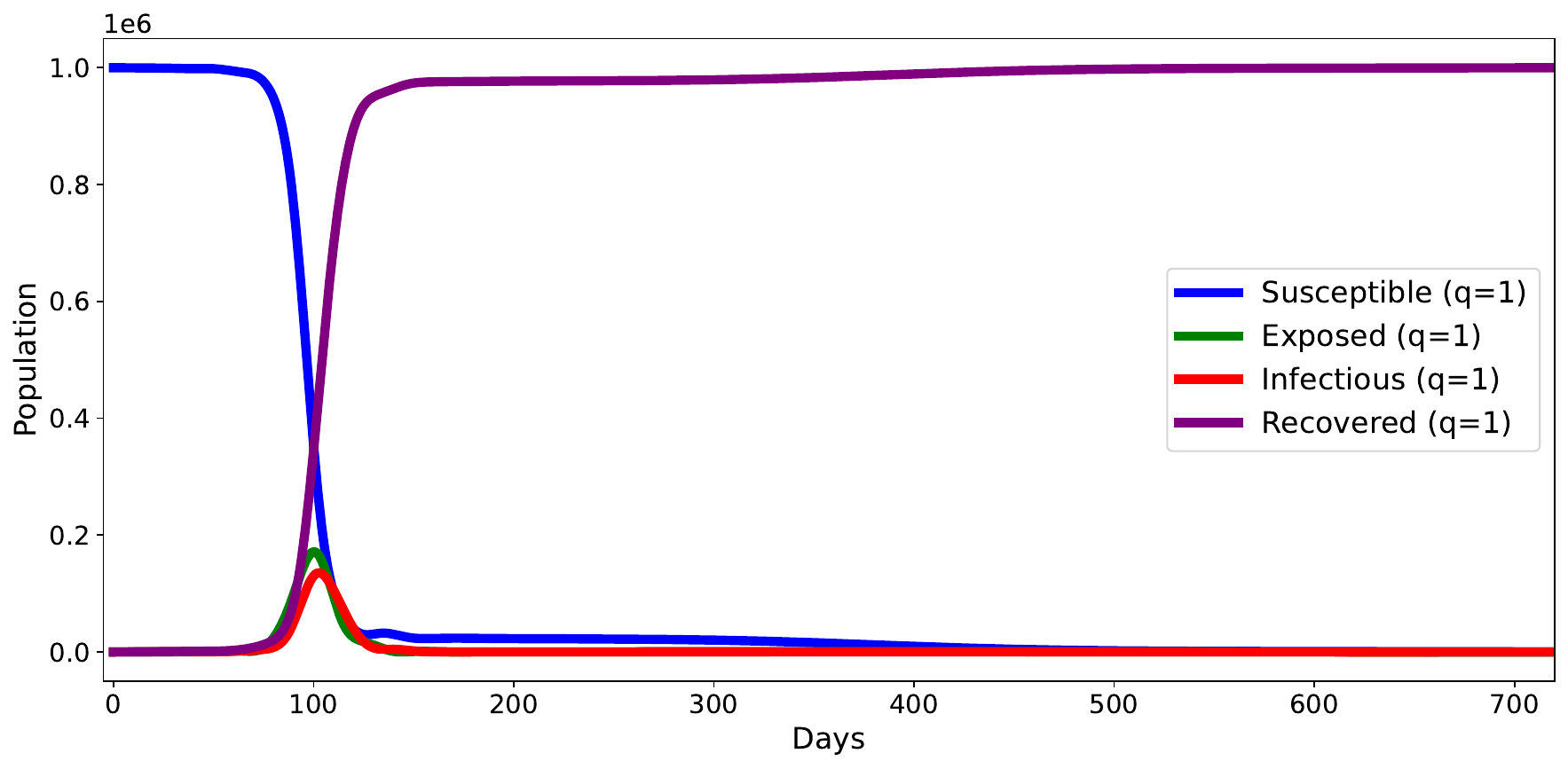}
    \caption{Reconstructed full trajectory based on the Fr\'echet median ($q=1$) ($J=6$, $K=60$, $\rho=1$).}
    \label{fig:realdata_FULL}
\end{figure}

\section{Discussion}
\label{sec:discussion}

We propose a method for summarizing multiple solutions to SEIR-type compartmental models on a functional space and constructing consensus epidemic trajectories capturing a part of the underlying dynamics of the system.
By mapping the solution curves into a functional space equipped with a metric, our approach enables direct comparison and summarization.
Specifically, solution curves are embedded into a Hilbert space.
This functional setting ensures that the dissimilarity between curves is measured by the norm induced by the inner product.
Under these constructions, we establish the existence of the consensus curve.
We further establish the robustness to an outlier curve under $q=1$.
The resulting optimization problem can be implemented using established techniques, such as basis function expansion and continuous optimization algorithms.

Our method is valuable in contexts where a unified perspective is required despite modeling uncertainties.
For instance, during the early stages of an emerging infectious disease when information is scarce, multiple research groups often provide disparate predictions based on SEIR-type models with varying assumptions.
Our method allows a national or regional researcher to integrate these diverse outputs into a single, representative curve.
Similarly, our approach is useful when seeking to obtain characteristic epidemic features from trajectories across multiple geographical regions.
By summarizing these localized trajectories, our method can construct global representative trajectories that preserve a part of the underlying dynamical features.

Our method is inspired by the integrating approach developed by \cite{bigot2013frechet}, which aims to compute a representative signal from temporally misaligned curves, such as ECG recordings.
While both approaches share the common goal of alignment and averaging in functional spaces, there are several differences.
\cite{bigot2013frechet} uses diffeomorphic transformations to achieve precise temporal alignment of signals.
In contrast, our method introduces only one shift parameter for each curve and focuses on translation.
Moreover, whereas \cite{bigot2013frechet} focuses on geometric variability in signal shape without any governing dynamics, our approach explicitly incorporates epidemiological constraints, such as the differential equation, as well as nonnegativity and population conservation conditions.
This enables not only the estimation of a representative trajectory but also the recovery of full dynamics and interpretable epidemiological parameters.

A key modeling choice in our approach is the relaxation of the full SEIR dynamics.
That is, we do not impose all four differential equations of the SEIR model as constraints.
Instead, we focus on enforcing only a subset of the model structure, namely, the differential equation governing the infectious compartment and the total population size.
This relaxation is motivated by the fact that various SEIR-type compartmental models, such as SEIUR and SEIRD, share a common structure: the infectious population increases at a rate proportional to the exposed population size and decreases at a rate proportional to its own size.
These transitions, representing the natural history of latency to infectiousness ($\sigma$) and removal ($\gamma$), are fundamentally intrinsic to the disease process across different model specifications.

Furthermore, this relaxation is advantageous in terms of the complexity of the optimization problem.
Enforcing the full SEIR system, including the nonlinear dynamics of the susceptible, infectious, and exposed compartments, such as ${d S}/{dt} = -\beta S I/N$ and ${d E}/{dt} = \beta S I/N - \sigma E$, would introduce quadratic or higher-order polynomial constraints with respect to the basis coefficients due to the multiplicative term $S(t)I(t)$.
This would make the optimization problem non-convex, thereby compromising computational efficiency and stability in solving the inner convex programs in the constructed algorithm.
Our approach thus provides a tractable compromise between dynamical fidelity and computational feasibility.
As a result, the consensus trajectory derived from our method does not necessarily correspond to an exact SEIR solution, but rather to a smoothed and constrained summary that is epidemiologically interpretable.
Moreover, once a representative curve is obtained, it is possible to perform parameter estimation in reverse, such as the transmission rate, under the assumption of the SEIR model.
This enables a form of parameter estimation, which synthesizes information across multiple trajectory sources and can yield interpretable summaries of complex or uncertain model outputs.

The choice of the power $q$ affects the summaries.
As established in Section~\ref{sec:cpfm-definition}, the minimizer under $q=1$ remains bounded even when one of the input object diverges, whereas no such bound is available for $q>1$.
For the consensus curve, we establish robustness in this sense in Section~\ref{sec:repcurve}.
This is consistent with the results in Sections~\ref{sec:simulation} and~\ref{sec:application}, where the standard Fr\'echet mean ($q=2$) yielded a longer estimate of the infectious period under the right-skewed distribution of $D_I$, while the Fr\'echet median ($q=1$) provided estimates closer to the input median.
Therefore, the Fr\'echet median may be a preferable choice when the collected trajectories are expected to be heterogeneous, such as in the early phase of an emerging infectious disease.

The use of the power Fr\'echet mean is not limited to trajectories arising from a single epidemic model.
Since our method is based on trajectories rather than parameter values, it is possible to integrate and summarize outputs from different types of infectious disease models, as long as their solutions can be defined in a function space.
This flexibility enables principled integration across model classes.

Finally, the introduction of a metric into the function space enables the use of tools from metric statistics.
Our formulation suggests that advanced tools such as Fr\'echet regression can be applied to the epidemic trajectories \citep{petersen2019frechet}.
This represents a notable expansion in the methodological toolkit available for infectious disease epidemiology.

In conclusion, we propose a method for summarizing multiple solutions from SEIR-type compartmental models on a functional space with partial mechanistic interpretability.
Our approach is valuable for integrating individual predictions from various research groups during the early phase of an emerging infectious disease, or for estimating a representative curve from trajectories across different geographical regions. 
Based on functional data analysis, we develop an optimization algorithm for constrained optimization.
Our method offers a novel approach to overcome modeling heterogeneity and parameter uncertainty in infectious disease modeling.


\section*{Funding}

This research was supported by AMED under Grant Number JP223fa627001 (UTOPIA AI Research Discovery Program) and JSPS KAKENHI Grant Number 26K02664.

\section*{Code Availability}

The Python implementation of the proposed method is available from the following GitHub repository: \url{https://github.com/t-yui/ConsensusEpidemicTrajectory}.

\section*{Disclosure statement}

The authors declare no conflicts of interest.
The authors used ChatGPT (OpenAI), Claude (Anthropic), and Gemini (Google) to assist with developing scripts for simulation and application, and editing the English language during the preparation of this manuscript.
The authors checked and edited the content and take full responsibility for this manuscript.

\appendix

\section{Proofs}
\label{sec:proof}

\subsection{Proof of Proposition~\ref{prop:existence_general}}

Before we begin the proof, we clarify a closedness condition on the set-valued map that is used in Proposition~\ref{prop:existence_general} (iii).

\begin{definition}[Weakly sequentially closed graph]
\label{def:wsc_graph}
Let $\Omega$ be a Hilbert space, let $\Theta\subset\mathbb R^{p}$, and let $\mathcal C:\Theta\rightrightarrows\Omega$ be a set-valued map.
We say that $\operatorname{Graph}(\mathcal C)$ is weakly sequentially closed if it is sequentially closed with respect to the product of the weak topology $\sigma(\Omega,\Omega^{\star})$ on $\Omega$ and the usual topology on $\Theta$.
Namely, if a sequence $((v_n,\theta_n) \mid n\in\mathbb N)\subset\operatorname{Graph}(\mathcal C)$ satisfies $v_n\rightharpoonup v$ in $\Omega$ and $\theta_n\to\theta$ in $\Theta$, then $(v,\theta)\in\operatorname{Graph}(\mathcal C)$.
\end{definition}

Here, $\Omega^{\star}$ is the dual space of $\Omega$, and $\rightharpoonup$ denotes weak convergence, namely, $f(v_n) \rightarrow f(v)$, for all linear functionals $f \in \Omega^{\star}$ (Proposition 3.5 of \cite{brezis2011functional}).

We then prove Proposition~\ref{prop:existence_general} as follows.

\begin{proof}[Proof of Proposition~\ref{prop:existence_general}.]
We prove (i), (ii), and (iii) separately.

\paragraph{Proof of (i)}
For each $j=1,\ldots,J$, the map $u \mapsto d_{\Omega}(u,u_j)^q$ is continuous on $\Omega$.
Therefore, $F_q$ is continuous.
Since $C$ is non-empty and compact, $F_q$ attains its minimum on $C$.
Thus, a minimizer $\hat{u}_q$ of \eqref{eq:const-power-frechet-q-mean} exists.

\paragraph{Proof of (ii)}
Define $m_q:=\inf_{u\in C} F_q(u)$ ($m_q < \infty$ because $C \neq \emptyset$) and take a minimizing sequence $(v_n |~n\in\mathbb N) \subset C$ satisfying $F_q(v_n) \to m_q$.
Since the convergent sequence $(F_q(v_n) |~ n \in \mathbb{N})$ is bounded, there exists $V<\infty$ such that $F_q(v_n)\leq V$ for all $n\in\mathbb N$.
Additionally, we have
\begin{align*}
    F_q(v_n)
    &\geq
    \frac{1}{J}\|v_n-u_1\|_{\Omega}^q
    \geq
    \frac{1}{J}
    \left(
        \|v_n\|_{\Omega}
        -
        \|u_1\|_{\Omega}
    \right)_+^q
    ,
\end{align*}
where $(x)_+:=\max(x,0)$.
Therefore, we obtain $\left(\|v_n\|_{\Omega}-\|u_1\|_{\Omega}\right)_{+}^{q} \leq J V$
and, hence,
\begin{align*}
    \left\|v_n\right\|_{\Omega}
    \leq
    \left\|u_1\right\|_{\Omega}
    +
    \left(J V\right)^{1/q}
    \quad
    \text{for all}
    ~
    n \in \mathbb{N}
    .
\end{align*}
This means $(v_n \mid n\in\mathbb N)$ is bounded in $\Omega$.
Since Hilbert space is reflexive, by Theorem 3.18 of \cite{brezis2011functional}, there exist a subsequence $(v_{n_{k}} |~k\in\mathbb N)$ and some $\hat{u}_q\in\Omega$ such that
\begin{align*}
    v_{n_{k}}
    \rightharpoonup
    \hat{u}_q
    .
\end{align*}
Since $C$ is closed and convex, it is weakly closed (Theorem 3.7 of \cite{brezis2011functional}).
Therefore, since $v_{n_k}\in C$ for every $k\in\mathbb N$, it follows from
$v_{n_k}\rightharpoonup\hat{u}_q$ that $\hat{u}_q \in C$.

For each $j=1,\ldots,J$, we have $v_{n_k}-u_j \rightharpoonup \hat{u}_q - u_j$. 
Then, the norm $\|\cdot\|_{\Omega}$ is weakly sequentially lower semicontinuous
(Proposition 3.5 of \cite{brezis2011functional}).
Thus, it follows that
\begin{align*}
    \left\|\hat{u}_q - u_j\right\|_{\Omega}
    \leq
    \liminf_{k \to \infty}\left\|v_{n_k} - u_j\right\|_{\Omega}
    .
\end{align*}
Since $t\mapsto t^{q}$ is continuous and nondecreasing on $t \in [0,\infty)$ for $q>0$, we have $\|\hat{u}_q-u_j\|_{\Omega}^{q}\leq\liminf_{k\to\infty}\|v_{n_k}-u_j\|_{\Omega}^{q}$.
Therefore, we obtain
\begin{align*}
    F_q(\hat{u}_q)
    &=
    \frac{1}{J}
    \sum_{j=1}^{J}
    \left\|\hat{u}_q-u_j\right\|_{\Omega}^{q}
    \\
    &\leq
    \frac{1}{J}
    \sum_{j=1}^{J}
    \liminf_{k\to\infty}
    \left\|v_{n_k}-u_j\right\|_{\Omega}^{q}
    \\
    &\leq
    \liminf_{k\to\infty}
    \frac{1}{J}
    \sum_{j=1}^{J}
    \left\|v_{n_k}-u_j\right\|_{\Omega}^{q}
    \\
    &=
    \liminf_{k\to\infty}
    F_q(v_{n_k})
    =
    m_q
    .
\end{align*}
Since $\hat{u}_q\in C$, by the definition of $m_q$, it follows that $F_q(\hat{u}_q) \geq m_q$.
Hence, $F_q(\hat{u}_q) = m_q$.
Thus, $\hat{u}_q$ is a minimizer of \eqref{eq:const-power-frechet-q-mean}.

\paragraph{Proof of (iii)}
Define $m_q := \inf_{(u,\theta)\in\operatorname{Graph}(\mathcal{C})} F_q(u,\theta)$ ($m_q<\infty$ because $\mathcal{C}(\theta)\neq\emptyset$ for each $\theta\in\Theta$ and $\Theta\neq\emptyset$) and take a minimizing sequence $((v_n,\theta_n) \mid n\in\mathbb N)\subset\operatorname{Graph}(\mathcal{C})$ satisfying $F_q(v_n,\theta_n)\to m_q$.
Since the convergent sequence $(F_q(v_n,\theta_n) \mid n\in\mathbb N)$ is bounded, there exists
$V<\infty$ such that $F_q(v_n,\theta_n)\leq V$ for all $n\in\mathbb N$.
Since $u_1$ is continuous on the compact set $\Theta$, we have $M := \sup_{\theta\in\Theta}\|u_1(\theta)\|_{\Omega} < \infty$.
In the same way as the proof of (ii), we obtain
\begin{align*}
    F_q(v_n,\theta_n)
    &\geq
    \frac{1}{J}\|v_n-u_1(\theta_n)\|_{\Omega}^q
    \geq
    \frac{1}{J}
    \left(
        \|v_n\|_{\Omega}
        -
        M
    \right)_+^q
    ,
\end{align*}
and, hence,
\begin{align*}
    \left\|v_n\right\|_{\Omega}
    \leq
    M
    +
    \left(J V\right)^{1/q}
    \quad
    \text{for all}
    ~
    n \in \mathbb{N}
    .
\end{align*}
This means $(v_n \mid n\in\mathbb N)$ is bounded in $\Omega$.

Since Hilbert space is reflexive, there exist a subsequence and some $\hat{u}_q\in\Omega$ where $v_{n_k} \rightharpoonup \hat{u}_q$.
Furthermore, since $\Theta\subset\mathbb R^{p}$ is compact, we can take a further subsequence, denoted by $(({v}_{n_k},{\theta}_{n_k}) \mid k\in\mathbb N)$, and some $\hat{\theta} \in \Theta$ such that
\begin{align*}
    v_{n_k}
    &\rightharpoonup
    \hat{u}_q
    ,
    \quad \text{and} \quad
    \theta_{n_k}
    \to
    \hat
    \theta
    .
\end{align*}
Since $v_{n_k}\in \mathcal{C}(\theta_{n_k})$ for every $k\in\mathbb N$, the weakly sequential closedness of $\operatorname{Graph}(\mathcal{C})$ yields $(\hat{u}_q,\hat\theta)\in\operatorname{Graph}(\mathcal{C})$.

For each $j=1,\ldots,J$, since $u_j$ is continuous, we have $u_j(\theta_{n_k}) \to u_j(\hat\theta)$ in $\Omega$, and strong convergence implies weak convergence (Theorem 3.5 of \cite{brezis2011functional}).
Therefore, we have $v_{n_k}-u_j(\theta_{n_k}) \rightharpoonup \hat{u}_q - u_j(\hat\theta)$.
Similarly to the proof of (ii), by the weak sequential lower semicontinuity of $\|\cdot\|_{\Omega}$ and the monotonicity of $t\mapsto t^{q}$ on $t\in[0,\infty)$, we have
\begin{align*}
    F_q(\hat{u}_q,\hat\theta)
    \leq
    \liminf_{k\to\infty} F_q(v_{n_k},\theta_{n_k})
    =
    m_q
    .
\end{align*}
We also have $F_q(\hat{u}_q,\hat\theta)\geq m_q$ because $(\hat{u}_q,\hat\theta)\in\operatorname{Graph}(\mathcal{C})$.
Therefore, we obtain $F_q(\hat{u}_q,\hat\theta) = m_q$, and, hence, $(\hat{u}_q,\hat\theta)$ is a minimizer of \eqref{eq:extended-const-power-frechet-q-mean}.
\end{proof}

\subsection{Proof of Proposition~\ref{prop:uniqueness}}

We prove Proposition~\ref{prop:uniqueness} as follows.

\begin{proof}[Proof of Proposition~\ref{prop:uniqueness}.]
Suppose that $\hat{u}_q$ and $\tilde{u}_q$ are both minimizers of \eqref{eq:const-power-frechet-q-mean} and let $m_q := F_q(\hat{u}_q) = F_q(\tilde{u}_q)$.
By Proposition~\ref{prop:existence_general}(ii), such $\hat{u}_q$ and $\tilde{u}_q$ exist.
Define
\begin{align*}
    w
    :=
    \frac{1}{2}\left(\hat{u}_q + \tilde{u}_q\right)
    ,
\end{align*}
which satisfies $w \in C$ because $C$ is convex.
For each $j=1,\ldots,J$, define $a_j := \|\hat{u}_q - u_j\|_{\Omega}$ and $b_j := \|\tilde{u}_q - u_j\|_{\Omega}$.
Then, by the triangle inequality, we have
\begin{align}
    \left\|w - u_j\right\|_{\Omega}
    &= 
    \frac{1}{2}\left(\left\|(\hat{u}_q - u_j) + (\tilde{u}_q - u_j)\right\|_{\Omega}\right)
    \nonumber
    \\
    &\leq
    \frac{1}{2}\left(\left\|\hat{u}_q - u_j\right\|_{\Omega} + \left\|\tilde{u}_q - u_j\right\|_{\Omega}\right)
    =
    \frac{1}{2}\left(a_j + b_j\right)
    ,
    \label{eq:midpoint_triangle}
\end{align}
and, since $t\mapsto t^{q}$ is nondecreasing and convex on $t \in [0,\infty)$ for $q>1$, we have
\begin{align}
    \left\|w - u_j\right\|_{\Omega}^{q}
    \leq
    \left(\frac{1}{2}\left(a_j + b_j\right)\right)^{q}
    \leq
    \frac{1}{2} \left(a_j^{q} + b_j^{q}\right)
    .
    \label{eq:midpoint_convexity}
\end{align}
Therefore, we obtain
\begin{align*}
    F_q(w)
    &=
    \frac{1}{J}\sum_{j=1}^{J} \left\|w - u_j\right\|_{\Omega}^{q}
    \\
    &\leq
    \frac{1}{2}\left(\frac{1}{J}\sum_{j=1}^{J}a_j^q + \frac{1}{J}\sum_{j=1}^{J}b_j^q\right) 
    \\
    &=
    \frac{1}{2}\left(F_q(\hat{u}_q) + F_q(\tilde{u}_q)\right)
    =
    m_q
    .
\end{align*}
Since $w \in C$, from the definition of $m_q$, it follows that $F_q(w) \geq m_q$.
Thus, we have $F_q(w) = m_q$.
Therefore, the inequalities in \eqref{eq:midpoint_triangle} and
\eqref{eq:midpoint_convexity} hold with equality for every $j=1,\ldots,J$.

We consider the second inequality in \eqref{eq:midpoint_convexity}.
Since $t\mapsto t^{q}$ is strictly convex on $t \in [0,\infty)$ for $q>1$, the equality
\begin{align*}
    \left(\frac{1}{2}\left(a_j + b_j\right)\right)^{q}
    =
    \frac{1}{2} \left(a_j^{q} + b_j^{q}\right)
    ,
    \quad \text{for}~j=1,\ldots,J,
\end{align*}
implies $a_j = b_j$ for $j=1,\ldots,J$.
Furthermore, using \eqref{eq:midpoint_triangle}, we obtain
\begin{align*}
    \left\|w - u_j\right\|_{\Omega}
    =
    a_j
    ,
    \quad
    \text{for}
    ~
    j = 1,\ldots,J
    .
\end{align*}
On the other hand, we have
\begin{align*}
    \left\|w - u_j\right\|_{\Omega}^{2}
    &= 
    \left(\frac{1}{2}\left(\left\|(\hat{u}_q - u_j) + (\tilde{u}_q - u_j)\right\|_{\Omega}\right)\right)^{2}
    \\
    &= 
    \frac{1}{4}\left(
    2\left\|\hat{u}_q - u_j\right\|_{\Omega}^2
    +
    2\left\|\tilde{u}_q - u_j\right\|_{\Omega}^2
    -
    \left\|(\hat{u}_q - u_j) - (\tilde{u}_q - u_j)\right\|_{\Omega}^{2}
    \right)
    \\
    &=
    \frac{1}{2}a_j^{2}
    +
    \frac{1}{2}b_j^{2}
    -
    \frac{1}{4}\left\|\hat{u}_q - \tilde{u}_q\right\|_{\Omega}^{2}
    .
\end{align*}
Since $a_j = b_j$ and $\|w - u_j\|_{\Omega} = a_j$, we obtain
\begin{align*}
    a_j^{2}
    =
    a_j^{2}
    -
    \frac{1}{4}\left\|\hat{u}_q - \tilde{u}_q\right\|_{\Omega}^{2}
    ,
\end{align*}
and, hence, $\|\hat{u}_q - \tilde{u}_q\|_{\Omega} = 0$.
Therefore, $\hat{u}_q = \tilde{u}_q$, and the minimizer of
\eqref{eq:const-power-frechet-q-mean} is unique.
\end{proof}

\subsection{Proof of Proposition~\ref{prop:robustness}}

We prove Proposition~\ref{prop:robustness} as follows.

\begin{proof}[Proof of Proposition~\ref{prop:robustness}.]
Fix an arbitrary $u_0 \in C$ and define
\begin{align*}
    G_n(u)
    :=
    \frac{1}{J}
    \left(
        \|u-u_J^{(n)}\|_{\Omega}
        +
        \sum_{j=1}^{J-1}\|u-u_j\|_{\Omega}
    \right)
\end{align*}
for $u \in C$.
By the definition of $\hat{u}_1^{(n)}$, we have $G_n(\hat{u}_1^{(n)}) \leq G_n(u_0)$, namely,
\begin{align*}
    \left\|\hat{u}_1^{(n)}-u_J^{(n)}\right\|_{\Omega}
    +
    \sum_{j=1}^{J-1}\left\|\hat{u}_1^{(n)}-u_j\right\|_{\Omega}
    \leq
    \left\|u_0-u_J^{(n)}\right\|_{\Omega}
    +
    \sum_{j=1}^{J-1}\left\|u_0-u_j\right\|_{\Omega}
    .
\end{align*}
By the triangle inequality, we have
\begin{align*}
    \left\|u_0-u_J^{(n)}\right\|_{\Omega}
    -
    \left\|\hat{u}_1^{(n)}-u_J^{(n)}\right\|_{\Omega}
    \leq
    \left\|\hat{u}_1^{(n)}-u_0\right\|_{\Omega}
    .
\end{align*}
Therefore, we obtain
\begin{align}
    \sum_{j=1}^{J-1}\left\|\hat{u}_1^{(n)}-u_j\right\|_{\Omega}
    \leq
    \left\|\hat{u}_1^{(n)}-u_0\right\|_{\Omega}
    +
    \sum_{j=1}^{J-1}\left\|u_0-u_j\right\|_{\Omega}
    .
    \label{eq:robust_reduced}
\end{align}
Since
\begin{align*}
     (J-1)\left\|\hat{u}_1^{(n)}\right\|_{\Omega}
    -
    \sum_{j=1}^{J-1}\left\|u_j\right\|_{\Omega}
    \leq \sum_{j=1}^{J-1}\left\|\hat{u}_1^{(n)}-u_j\right\|_{\Omega}
    ,
\end{align*}
and 
\begin{align*}
    \left\|\hat{u}_1^{(n)}-u_0\right\|_{\Omega}
    \leq 
    \left\|\hat{u}_1^{(n)}\right\|_{\Omega}
    +
    \left\|u_0\right\|_{\Omega}
    ,
\end{align*}
by the triangle inequality, \eqref{eq:robust_reduced} yields
\begin{align*}
    (J-1)\left\|\hat{u}_1^{(n)}\right\|_{\Omega}
    -
    \sum_{j=1}^{J-1}\left\|u_j\right\|_{\Omega}
    \leq
    \left\|\hat{u}_1^{(n)}\right\|_{\Omega}
    +
    \left\|u_0\right\|_{\Omega}
    +
    \sum_{j=1}^{J-1}\|u_0-u_j\|_{\Omega}
    .
\end{align*}
Since $J \geq 3$, i.e., $J-2 \geq 1$, it follows that
\begin{align*}
    \left\|\hat{u}_1^{(n)}\right\|_{\Omega}
    \leq
    \frac{1}{J-2}
    \left(
        \left\|u_0\right\|_{\Omega}
        +
        \sum_{j=1}^{J-1}\|u_0-u_j\|_{\Omega}
        +
        \sum_{j=1}^{J-1}\left\|u_j\right\|_{\Omega}
    \right)
    \quad
    \text{for}
    ~
    n\in\mathbb N
    .
\end{align*}
This upper bound of $\left\|\hat{u}_1^{(n)}\right\|_{\Omega}$ depends only on $u_1,\ldots,u_{J-1}$ and on the fixed point $u_0 \in C$, and it is independent of $n$.
Therefore, we have $\sup_{n\in\mathbb N}\|\hat{u}_1^{(n)}\|_{\Omega}<\infty$.
\end{proof}

\subsection{Proof of Corollary~\ref{cor:existence_consensus_curve}}

Define the objective function
\begin{align*}
    Q_q(y,\boldsymbol{\delta})
    :=
    \frac{1}{J}
    \sum_{j=1}^{J}
    d\left(
        y,
        \mathcal T_{\delta_j}y_j
    \right)^q
    .
\end{align*}

We first establish the continuity of the shift operator in $\mathcal Y$.

\begin{lemma}
\label{lem:shift_continuity}
For any $y\in\mathcal Y$, the map $\mathbb{R} \ni \eta \mapsto \mathcal{T}_{\eta}y \in \mathcal{Y}$
is continuous.
\end{lemma}

\begin{proof}[Proof of Lemma~\ref{lem:shift_continuity}.]
Let $y=(E,I)^{\top} \in \mathcal{Y}$ and fix $\eta_0 \in \mathbb{R}$.
It is sufficient to prove the statement for each component of $(E,I)$, and we prove the statement for $E$.

As stated in Section \ref{sec:deffuncspace}, $H^1([0,T])$ is continuously embedded into $C([0,T])$, and hence we identify $E^{\mathrm{ext}} \in C(\mathbb{R})$.
$E^{\mathrm{ext}}$ is constant on $(-\infty,0)$ and $(T,\infty)$, and continuous on compact set $[0,T]$.
Since continuous functions on a compact set are bounded and uniformly continuous, and constant functions are also bounded and uniformly continuous, $E^{\mathrm{ext}}$ is bounded and uniformly continuous on $\mathbb{R}$.
The weak derivative of $E^{\mathrm{ext}}$ is
\begin{align*}
    \left(E^{\mathrm{ext}}\right)'
    =
    E'\mathbf{1}_{(0,T)}
    \in
    L^{2}(\mathbb{R})
    ,
\end{align*}
where
\begin{align*}
    \mathbf{1}_{(0, T)}(t)
    :=
    \begin{cases}
        1, & t \in(0, T)
        ,
        \\
        0, & t \notin(0, T)
        .
    \end{cases}
\end{align*}
By the definition of the shift operator, we have $(\mathcal T_{\eta}E)(t)=E^{\mathrm{ext}}(t+\eta)$ for all $t\in[0,T]$, and $(\mathcal T_{\eta}E)'(t)=(E^{\mathrm{ext}})'(t+\eta)$ for almost every $t\in[0,T]$.

We evaluate the two parts of the norm separately.
For the function $E^{\mathrm{ext}}$ part, by the uniform continuity on $\mathbb{R}$, we have
\begin{align*}
    &\int_{0}^{T}
    \left|
        E^{\mathrm{ext}}(t+\eta)-E^{\mathrm{ext}}(t+\eta_0)
    \right|^{2}
    dt
    \\
    &\quad
    \leq
    T \sup_{s\in\mathbb R}
    \left|
        E^{\mathrm{ext}}(s+\eta-\eta_0)-E^{\mathrm{ext}}(s)
    \right|^{2}
    \to
    0
\end{align*}
as $\eta \to \eta_0$.
For the derivative $\left(E^{\mathrm{ext}}\right)'$ part, we have
\begin{align*}
    &\int_{0}^{T}
    \left|
        \left(E^{\mathrm{ext}}\right)'(t+\eta)
        -
        \left(E^{\mathrm{ext}}\right)'(t+\eta_0)
    \right|^{2}
    dt
    \\
    &\quad
    \leq
    \int_{\mathbb R}
    \left|
        \left(E^{\mathrm{ext}}\right)'(s+\eta-\eta_0)
        -
        \left(E^{\mathrm{ext}}\right)'(s)
    \right|^{2}
    ds
    \to
    0
    ,
\end{align*}
as $\eta \to \eta_0$ (Lemma 4.3 of \cite{brezis2011functional}).
Combining these results for $E$ and $I$, we have
\begin{align*}
    \left\|
        \mathcal T_{\eta}y-\mathcal T_{\eta_0}y
    \right\|_{\mathcal Y}
    \to
    0
    \quad
    \text{as}
    ~
    \eta\to\eta_0
    .
\end{align*}
\end{proof}

We next establish compactness of the parameter space, weak sequential closedness of the feasible set, and that the feasible set is non-empty.

\begin{lemma}
\label{lem:F_weakly_closed}
The following statements hold:
\begin{enumerate}
    \item[(i)] $[\sigma_{\min},\sigma_{\max}] \times [\gamma_{\min},\gamma_{\max}] \times \Delta$ is a compact set.
    \item[(ii)] Take sequences $(y_n |~ n \in \mathbb{N}) \subset \mathcal{Y}$, $(\sigma_n |~ n \in \mathbb{N}) \subset [\sigma_{\min},\sigma_{\max}]$, $(\gamma_n |~ n \in \mathbb{N}) \subset [\gamma_{\min},\gamma_{\max}]$, and $(\boldsymbol{\delta}_n |~ n \in \mathbb{N}) \subset \Delta$.
    Suppose that $y_n \rightharpoonup y$ in $\mathcal Y$, $\sigma_n \to \sigma$, $\gamma_n \to \gamma$, $\boldsymbol{\delta}_n \to \boldsymbol{\delta}$, and $(y_n,\sigma_n,\gamma_n,\boldsymbol{\delta}_n) \in \mathcal F$ for all $n\in\mathbb N$.
    Then, we have $(y,\sigma,\gamma,\boldsymbol{\delta}) \in \mathcal F$.
    \item[(iii)] $\mathcal{F} \neq \emptyset$.
\end{enumerate}
\end{lemma}

\begin{proof}[Proof of Lemma~\ref{lem:F_weakly_closed}.]
We prove (i), (ii), and (iii) separately.

\paragraph{Proof of (i)}
Since $\Delta$ is a closed subset of the compact set
$[-\delta_{\max},\delta_{\max}]^J$, the product set $[\sigma_{\min},\sigma_{\max}] \times [\gamma_{\min},\gamma_{\max}] \times \Delta$ is compact.
That is, (i) holds.

\paragraph{Proof of (ii)}
Let $y_n = (E_n, I_n)^\top$ and $y = (E, I)^\top$.
Since $y_n \rightharpoonup y$ in $\mathcal{Y}$, we have $E_n \rightharpoonup E$ in $H^1([0,T])$ and $I_n \rightharpoonup I$ in $H^1([0,T])$.
As stated in Section \ref{sec:deffuncspace}, $H^1([0,T])$ is continuously embedded into $C([0,T])$.
For $u \in H^1([0,T])$, let $\operatorname{ev}_t(u)=u(t)$ denote point evaluation at each fixed $t\in[0,T]$.
Since $\operatorname{ev}_t$ is a continuous linear functional, by the definition of weak convergence, we have $E_n(t) = \operatorname{ev}_t(E_n) \to \operatorname{ev}_t(E) = E(t)$ for every $t\in[0,T]$.
Similarly, we obtain $I_n(t) \to I(t)$ for every $t\in[0,T]$.
Since $E_n(t)\geq0$ and $I_n(t)\geq0$, we have $E(t) \geq 0$ and $I(t) \geq 0$ for every $t\in[0,T]$.
For every fixed $t\in[0,T]$, the map
\begin{align*}
    H^1([0,T])
    \ni
    u
    \mapsto
    \int_0^t u(s)\,ds
    \in
    \mathbb R
    ,
\end{align*}
is also a continuous linear functional.
Hence, we have
\begin{align*}
    \int_0^t I_n(s)\,ds
    \to
    \int_0^t I(s)\,ds
    ,
\end{align*}
for every $t\in[0,T]$.
Similarly, since
\begin{align*}
    N
    -
    E_n(t)
    -
    I_n(t)
    -
    \gamma_n
    \int_0^t I_n(s)\,ds
    \geq
    0
    ,
\end{align*}
for every $t\in[0,T]$ and $\gamma_n\to\gamma$, we have
\begin{align*}
    N
    -
    E(t)
    -
    I(t)
    -
    \gamma
    \int_0^t I(s)\,ds
    \geq
    0
    ,
\end{align*}
for every $t\in[0,T]$.
For the differential equation constraints, since $I_n\rightharpoonup I$ in $H^1([0,T])$, we have $I_n' \rightharpoonup I'$ in $L^2([0,T])$.
Similarly, $E_n \rightharpoonup E$ in $L^2([0,T])$ and $I_n \rightharpoonup I$ in $L^2([0,T])$.
Since $\sigma_n\to\sigma$ and $\gamma_n\to\gamma$, it follows that $\sigma_n E_n - \gamma_n I_n \rightharpoonup \sigma E - \gamma I$ in $L^2([0,T])$.
Since $I_n' = \sigma_n E_n - \gamma_n I_n$ in $L^2([0,T])$ for every $n \in \mathbb{N}$, we have
\begin{align*}
    I'
    =
    \sigma E
    -
    \gamma I
    ,
\end{align*}
in $L^2([0,T])$.

Finally, since $\boldsymbol{\delta}_n\in\Delta$ for every $n \in \mathbb{N}$, $\boldsymbol{\delta}_n\to\boldsymbol{\delta}$, and $\Delta$ is a closed set, we obtain
$\boldsymbol{\delta} \in \Delta$.
Similarly, we have $\sigma \in [\sigma_{\min},\sigma_{\max}]$ and $\gamma \in [\gamma_{\min},\gamma_{\max}]$.

Therefore, we have $(y,\sigma,\gamma,\boldsymbol{\delta}) \in \mathcal F$.

\paragraph{Proof of (iii)}
Take arbitrary $\sigma_0 \in [\sigma_{\min},\sigma_{\max}]$ and $\gamma_0 \in [\gamma_{\min},\gamma_{\max}]$.
Let $y_0 := (0,0)^\top \in \mathcal Y$ and $\boldsymbol{\delta}_0 := (0,\ldots,0) \in \Delta$.
Then we have
\begin{align*}
    \left\{\begin{array}{l}
        E_0 = 0, \quad I_0 = 0,
        \\
        N - E_0 - I_0 - \gamma_0 \int_0^\cdot I_0(s)\,ds = N \geq
        0
        ,
        \\
        I_0' = 0 = \sigma_0 E_0 - \gamma_0 I_0
        .
    \end{array}\right.
\end{align*}
Therefore, it holds that $(y_0,\sigma_0,\gamma_0,\boldsymbol{\delta}_0) \in \mathcal F$ and hence $\mathcal F\neq\emptyset$.
\end{proof}

We now prove Corollary~\ref{cor:existence_consensus_curve}.

\begin{proof}[Proof of Corollary~\ref{cor:existence_consensus_curve}.]
We apply Proposition~\ref{prop:existence_general}(iii) with $\Omega = \mathcal Y$, $\Theta = [\sigma_{\min},\sigma_{\max}] \times [\gamma_{\min},\gamma_{\max}] \times \Delta$, $\theta = (\sigma,\gamma,\boldsymbol{\delta})$, $u_j(\theta) := \mathcal T_{\delta_j} y_j$, and $\mathcal{C}(\theta) := \left\{y \in \mathcal Y \mid (y,\sigma,\gamma,\boldsymbol{\delta}) \in \mathcal F\right\}$, for $j=1,\ldots,J$.
Then, $\operatorname{Graph}(\mathcal{C}) = \mathcal F$ and $F_q(y,\theta) = Q_q(y,\boldsymbol{\delta})$.
We verify the assumptions of Proposition~\ref{prop:existence_general}(iii) as follows.

First, as stated in Section~\ref{sec:deffuncspace}, $\mathcal Y$ is a Hilbert space with $d(y_1,y_2) = \|y_1-y_2\|_{\mathcal Y}$.
By Lemma~\ref{lem:F_weakly_closed}(i), $\Theta$ is a compact set, and it is non-empty because $\boldsymbol{0} \in \Delta$.
Second, by Lemma~\ref{lem:F_weakly_closed}(iii), we have $\mathcal{F} \neq \emptyset$.
Moreover, the trajectory $y_0 := (0,0)^\top$ taken in the proof of Lemma~\ref{lem:F_weakly_closed}(iii) satisfies $y_0 \in \mathcal{C}(\theta)$ for every $\theta \in \Theta$, since the constraints of $\mathcal{F}$ hold for $y_0$ with any $(\sigma,\gamma) \in [\sigma_{\min},\sigma_{\max}] \times [\gamma_{\min},\gamma_{\max}]$ and any $\boldsymbol{\delta} \in \Delta$.
Therefore, $\mathcal{C}$ has non-empty values.
Third, for each $j=1,\ldots,J$, the map $\theta \mapsto \delta_j$ is continuous, and the map $\eta \mapsto \mathcal T_{\eta} y_j$ is continuous by Lemma~\ref{lem:shift_continuity}.
Therefore, $\theta \mapsto u_j(\theta) = \mathcal T_{\delta_j} y_j$ is continuous on $\Theta$.
Fourth, by Lemma~\ref{lem:F_weakly_closed}(ii), $\operatorname{Graph}(\mathcal{C}) = \mathcal F$ is weakly sequentially closed as defined in Definition~\ref{def:wsc_graph}.

Since all the assumptions of Proposition~\ref{prop:existence_general}(iii) are satisfied, there exists a minimizer $(\hat y_q, \hat\theta) = (\hat y_q, \hat\sigma, \hat\gamma, \hat{\boldsymbol{\delta}}) \in \operatorname{Graph}(\mathcal{C}) = \mathcal F$ of $F_q(y,\theta) = Q_q(y,\boldsymbol{\delta})$.
Thus, a minimizer of \eqref{eq:joint_shift_frechet} exists.
\end{proof}

\subsection{Proof of Corollary~\ref{cor:robustness_consensus_curve}}

We prove Corollary~\ref{cor:robustness_consensus_curve} as follows.

\begin{proof}[Proof of Corollary~\ref{cor:robustness_consensus_curve}.]
Fix $n \in \mathbb{N}$ and let $y_0 := (0,0)^\top \in \mathcal Y$.
As shown in the proof of Lemma~\ref{lem:F_weakly_closed}(iii), we have
$(y_0,\hat\sigma^{(n)},\hat\gamma^{(n)},\hat{\boldsymbol\delta}^{(n)}) \in \mathcal F$.
Therefore, similarly to the proof of Proposition \ref{prop:robustness}, by the definition of $\hat y^{(n)}$, we have
\begin{align*}
    \left\|\hat y^{(n)}-\mathcal T_{\hat\delta_J^{(n)}}y_J^{(n)}\right\|_{\mathcal Y}
    +
    \sum_{j=1}^{J-1}\left\|\hat y^{(n)}-\mathcal T_{\hat\delta_j^{(n)}}y_j\right\|_{\mathcal Y}
    \leq
    \left\|\mathcal T_{\hat\delta_J^{(n)}}y_J^{(n)}\right\|_{\mathcal Y}
    +
    \sum_{j=1}^{J-1}\left\|\mathcal T_{\hat\delta_j^{(n)}}y_j\right\|_{\mathcal Y}
    .
\end{align*}
In the same way as the proof of Proposition~\ref{prop:robustness} with $u_0 = y_0$, we obtain
\begin{align*}
    (J-2)\left\|\hat y^{(n)}\right\|_{\mathcal Y}
    \leq
    2\sum_{j=1}^{J-1}\left\|\mathcal T_{\hat\delta_j^{(n)}}y_j\right\|_{\mathcal Y}
    .
\end{align*}
Since the map $\eta \mapsto \mathcal T_{\eta}y_j$ is continuous by Lemma~\ref{lem:shift_continuity} and $[-\delta_{\max},\delta_{\max}]$ is compact, we have
\begin{align*}
    M_j
    :=
    \sup_{|\eta|\leq\delta_{\max}}\left\|\mathcal T_{\eta}y_j\right\|_{\mathcal Y}
    < \infty
    ,
    \quad
    \text{for}~j=1,\ldots,J-1
    .
\end{align*}
Therefore, since $\hat\delta_j^{(n)} \in [-\delta_{\max},\delta_{\max}]$ and $J-2\geq 1$, it follows that
\begin{align*}
    \left\|\hat y^{(n)}\right\|_{\mathcal Y}
    \leq
    \frac{2}{J-2}\sum_{j=1}^{J-1}M_j
    \quad
    \text{for}
    ~
    n\in\mathbb N
    .
\end{align*}
This upper bound does not depend on $n$, and, hence, we have $\sup_{n\in\mathbb N}\|\hat y^{(n)}\|_{\mathcal Y}<\infty$.
\end{proof}

\section{Block Optimization for $\boldsymbol{c}$ and $(\sigma,\gamma)$}
\label{sec:simplified_optimization}

For fixed $(\sigma,\gamma)$, the constraints \eqref{eq:bspline_constraints} are linear constraints in $\boldsymbol{c}$.
Hence, the optimization problem \eqref{eq:rep_update_given_shift} reduces to a quadratic programming (QP) problem when $q=2$, and to a second-order cone programming (SOCP) problem when $q=1$.
This eliminates the need to rely on general sequential quadratic programming (SQP) solvers.
In this subsection, we describe a two-level algorithm that (i) solves the inner problem by QP/SOCP and (ii) updates $(\sigma,\gamma)$.

\paragraph{(i) Convex Program for Fixed $(\sigma, \gamma)$.}
Fix $(\sigma,\gamma)$.
Given $\boldsymbol{\delta}^{(r)}$ and the coefficient vectors $\{\boldsymbol{c}_j\}_{j=1}^J$, consider the optimization problem:
\begin{align}
    \hat{\boldsymbol{c}}(\sigma,\gamma)
    :=
    \argmin_{\boldsymbol{c}\in\mathbb{R}^{2K}}
    \left\{
        \frac{1}{J}\sum_{j=1}^J D(\boldsymbol{c},\boldsymbol{c}_j)^q
    \right\}
    \quad
    \text{s.t.}\quad\text{constraints}~\eqref{eq:bspline_constraints}.
    \label{eq:inner_problem_fixed_sigma_gamma}
\end{align}
We now describe specific formulations for (1) $q=1$, (2) $q=2$, and (3) $q \neq 1$, $q \neq 2$.

\begin{itemize}
    \item[(P1)]
    When $q=1$, the objective function in \eqref{eq:inner_problem_fixed_sigma_gamma} is non-differentiable at $\boldsymbol{c}=\boldsymbol{c}_j$.
    Due to the positive semi-definiteness of $\mathbf{G}^{\mathcal{Y}}$, there exists a matrix $\mathbf{L}$ satisfying $\mathbf{G}^{\mathcal{Y}}=\mathbf{L}^\top\mathbf{L}$.
    Then, we have $D(\boldsymbol{c},\boldsymbol{c}_j)=\|\mathbf{L}(\boldsymbol{c}-\boldsymbol{c}_j)\|_2$, and \eqref{eq:inner_problem_fixed_sigma_gamma} has the SOCP formulation:
    \begin{align*}
        \argmin_{\boldsymbol{c}\in\mathbb{R}^{2K},\ \boldsymbol{u}\in\mathbb{R}^J}
        \ \frac{1}{J}\sum_{j=1}^J u_j
        \quad
        \text{s.t.}\quad
        \|\mathbf{L}(\boldsymbol{c}-\boldsymbol{c}_j)\|_2 \leq u_j,\ \text{for}~j=1,\ldots,J,
        \quad
        \text{constraints}~\eqref{eq:bspline_constraints}
        .
    \end{align*}
    \item[(P2)]
    When $q=2$, the objective function in \eqref{eq:inner_problem_fixed_sigma_gamma} becomes quadratic:
    \begin{align*}
        \frac{1}{J}\sum_{j=1}^J D(\boldsymbol{c},\boldsymbol{c}_j)^2
        =
        \frac{1}{J}\sum_{j=1}^J
        (\boldsymbol{c}-\boldsymbol{c}_j)^\top \mathbf{G}^{\mathcal{Y}}(\boldsymbol{c}-\boldsymbol{c}_j)
        =
        (\boldsymbol{c}-\bar{\boldsymbol{c}})^\top \mathbf{G}^{\mathcal{Y}}(\boldsymbol{c}-\bar{\boldsymbol{c}})
        +
        \text{const}
        ,
    \end{align*}
    where $\bar{\boldsymbol{c}}:={\sum_{j=1}^J \boldsymbol{c}_j} / J$.
    Therefore, \eqref{eq:inner_problem_fixed_sigma_gamma} is equivalent to the quadratic program
    \begin{align*}
        \hat{\boldsymbol{c}}(\sigma,\gamma)
        =
        \argmin_{\boldsymbol{c}\in\mathbb{R}^{2K}}
        (\boldsymbol{c}-\bar{\boldsymbol{c}})^\top \mathbf{G}^{\mathcal{Y}}(\boldsymbol{c}-\bar{\boldsymbol{c}})
        \quad
        \text{s.t.}\quad\text{constraints}~\eqref{eq:bspline_constraints}.
    \end{align*}
    \item[(P3)]
    For $q \neq 1$ and $q \neq 2$, we can apply an IRLS (iteratively reweighted least squares)-type method that reduces each update to a QP.
    Given an iterate $\boldsymbol{c}^{(s)}$, define
    \begin{align*}
        w_j^{(s)}
        :=
        \left\{ D(\boldsymbol{c}^{(s)},\boldsymbol{c}_j)^2 + \varepsilon \right\}^{\frac{q}{2}-1},
        \quad \text{for}~j=1,\ldots,J,
    \end{align*}
    with a small $\varepsilon>0$.
    We obtain the update $\boldsymbol{c}^{(s+1)}$ by the weighted QP:
    \begin{align*}
        \boldsymbol{c}^{(s+1)}
        =
        \argmin_{\boldsymbol{c}\in\mathbb{R}^{2K}}
        \sum_{j=1}^J w_j^{(s)}
        (\boldsymbol{c}-\boldsymbol{c}_j)^\top \mathbf{G}^{\mathcal{Y}}(\boldsymbol{c}-\boldsymbol{c}_j)
        \quad
        \text{s.t.}\quad\text{constraints}~\eqref{eq:bspline_constraints}
    \end{align*}
\end{itemize}

\paragraph{(ii) Profile Optimization for Updating $(\sigma, \gamma)$.}
We next profile out $\boldsymbol{c}$ and optimize in the two-dimensional parameter space.
Define the profile objective
\begin{align*}
    V_q(\sigma,\gamma)
    :=
    \min_{\boldsymbol{c}\in\mathbb{R}^{2K}}
    \left\{
        \frac{1}{J}\sum_{j=1}^J D(\boldsymbol{c},\boldsymbol{c}_j)^q
    \right\}
    \quad
    \text{s.t.}\quad\text{constraints}~
    \eqref{eq:bspline_constraints},
\end{align*}
and obtain
\begin{align}
    (\hat{\sigma},\hat{\gamma})
    &=
    \argmin_{(\sigma,\gamma) \in [\sigma_{\min},\sigma_{\max}] \times [\gamma_{\min},\gamma_{\max}]} V_q(\sigma,\gamma),
    \label{eq:outer_profile_minimization}
    \\
    \hat{\boldsymbol{c}}
    &:=
    \hat{\boldsymbol{c}}(\hat{\sigma},\hat{\gamma}).
    \nonumber
\end{align}

We note that the population constraint yields an explicit upper bound for $\gamma$.
Let
\begin{align*}
    E_m := (\mathbf{B}\boldsymbol{c}^{(E)})_m,
    \quad
    I_m := (\mathbf{B}\boldsymbol{c}^{(I)})_m,
    \quad
    \mathcal{I}_m := (\mathbf{\Phi}\boldsymbol{c}^{(I)})_m = \int_{0}^{t_m} I(s)\,ds.
\end{align*}
If $\mathcal{I}_m>0$, then $S(t_m)=N-E_m-I_m-\gamma \mathcal{I}_m\geq 0$ implies
\begin{align*}
    \gamma \leq \frac{N-E_m-I_m}{\mathcal{I}_m}
    .
\end{align*}
Consequently, any feasible solution satisfies
\begin{align}
    \gamma \leq
    \min_{m:~\mathcal{I}_m>0}
    \frac{N-(\mathbf{B}\boldsymbol{c}^{(E)})_m-(\mathbf{B}\boldsymbol{c}^{(I)})_m}{(\mathbf{\Phi}\boldsymbol{c}^{(I)})_m}.
    \label{eq:gamma_bound_global}
\end{align}
In the outer profile optimization for $(\sigma,\gamma)$, we use \eqref{eq:gamma_bound_global} to restrict the candidate range of $\gamma$.

When $q \neq 1$, to solve the outer minimization problem \eqref{eq:outer_profile_minimization} efficiently using gradient-based optimization methods such as L-BFGS-B, we require the gradient of the profile objective function $\nabla V_q(\sigma, \gamma)$.
Since the inner problem is a convex optimization problem with differentiable constraints with respect to parameters, we can apply the envelope theorem to compute the exact gradient analytically using the Lagrange multipliers obtained from the inner solution.
Specifically, let $\hat{\boldsymbol{c}}$ be the optimal solution of the inner problem for fixed $(\sigma, \gamma)$.
We define the Lagrangian of the inner problem as $\mathcal{L}(\boldsymbol{c}, \boldsymbol{\nu}, \boldsymbol{\lambda}, \boldsymbol{\lambda}_E, \boldsymbol{\lambda}_I, \sigma, \gamma)$.
Here, $\boldsymbol{\nu} \in \mathbb{R}^{M+1}$ is the Lagrange multiplier vector associated with the differential equation constraints.
$\boldsymbol{\lambda} \in \mathbb{R}^{M+1}_{\geq 0}$, $\boldsymbol{\lambda}_E \in \mathbb{R}^{M+1}_{\geq 0}$, and $\boldsymbol{\lambda}_I \in \mathbb{R}^{M+1}_{\geq 0}$ are the multipliers for the population upper bound constraint, the non-negativity of $E$, and the non-negativity of $I$, respectively.
The full Lagrangian is given by:
\begin{align*}
    \mathcal{L}
    &=
    \frac{1}{J}\sum_{j=1}^J D(\boldsymbol{c},\boldsymbol{c}_j)^q
    + \boldsymbol{\nu}^\top \left(\mathbf{B}'\boldsymbol{c}^{(I)} - \sigma\mathbf{B}\boldsymbol{c}^{(E)} + \gamma\mathbf{B}\boldsymbol{c}^{(I)} \right)
    \nonumber
    \\
    &
    \quad- \boldsymbol{\lambda}^\top \left( N\boldsymbol{1} - \mathbf{B}\boldsymbol{c}^{(E)} - \mathbf{B}\boldsymbol{c}^{(I)} - \gamma\mathbf{\Phi}\boldsymbol{c}^{(I)} \right)
    - \boldsymbol{\lambda}_E^\top \mathbf{B}\boldsymbol{c}^{(E)}
    - \boldsymbol{\lambda}_I^\top \mathbf{B}\boldsymbol{c}^{(I)}
    .
\end{align*}
According to the envelope theorem, the gradient of the value function $V_q(\sigma, \gamma)$ with respect to the parameters is given by the partial derivatives of the Lagrangian evaluated at the optimal primal-dual pair.
Note that the objective function and the non-negativity constraints (associated with $\boldsymbol{\lambda}_E$ and $\boldsymbol{\lambda}_I$) do not explicitly depend on $\sigma$ and $\gamma$.
Thus, their partial derivatives vanish, and we obtain:
\begin{align*}
    \frac{\partial V_q}{\partial \sigma}
    &=
    \frac{\partial \mathcal{L}}{\partial \sigma}\bigg|_{\boldsymbol{c}=\hat{\boldsymbol{c}}}
    =
    -\boldsymbol{\nu}^\top \mathbf{B}\hat{\boldsymbol{c}}^{(E)},
    \\
    \frac{\partial V_q}{\partial \gamma}
    &=
    \frac{\partial \mathcal{L}}{\partial \gamma}\bigg|_{\boldsymbol{c}=\hat{\boldsymbol{c}}}
    =
    \boldsymbol{\nu}^\top \mathbf{B}\hat{\boldsymbol{c}}^{(I)}
    +
    \boldsymbol{\lambda}^\top \mathbf{\Phi}\hat{\boldsymbol{c}}^{(I)}.
\end{align*}

Since $\mathcal{L}$ is non-differentiable if we consider $q=1$, we use $q=1+\varepsilon$ instead in our implementation for efficient profile optimization.
Here, $\varepsilon$ is a sufficiently small positive constant.

\section{Sensitivity Analysis for the Simulation Study}
\label{sec:sensanal}

We conducted sensitivity analyses for other settings of hyperparameters for the constrained power Fr\'echet means.
Specifically, we varied the number of basis functions ($K \in \{15, 30, 60\}$) and the scaling parameter for the derivative part of the metric ($\rho \in \{0.1, 0.5, 1\}$) for several settings, and applied the proposed method to $J=10$ SEIR curves.
The other settings for the experiments were the same as Section~\ref{sec:simulation}.

Across all examined hyperparameters, the obtained consensus curves and the resulting epidemiological parameters remained almost consistent.
As shown in Figure~\ref{fig:trajectories_EI_sensitivity_K} and Table~\ref{tab:estimated_params_sensitivity_K}, varying the number of B-spline basis functions from $K=15$ to $K=60$ did not yield substantial differences in the obtained curves and the point estimates of $\sigma$, $\gamma$, and $\beta$ under the assumption of SEIR-type transmission model.
Similarly, the results were not sensitive to the choice of $\rho$ (Figure~\ref{fig:trajectories_EI_sensitivity_rho} and Table~\ref{tab:estimated_params_sensitivity_rho}), with almost identical parameter estimates.
These findings suggest that the proposed approach is stable with respect to these hyperparameters.

\begin{figure}[htbp]
    \centering
    \begin{subfigure}[t]{0.48\textwidth}
        \centering
        \includegraphics[width=\textwidth]{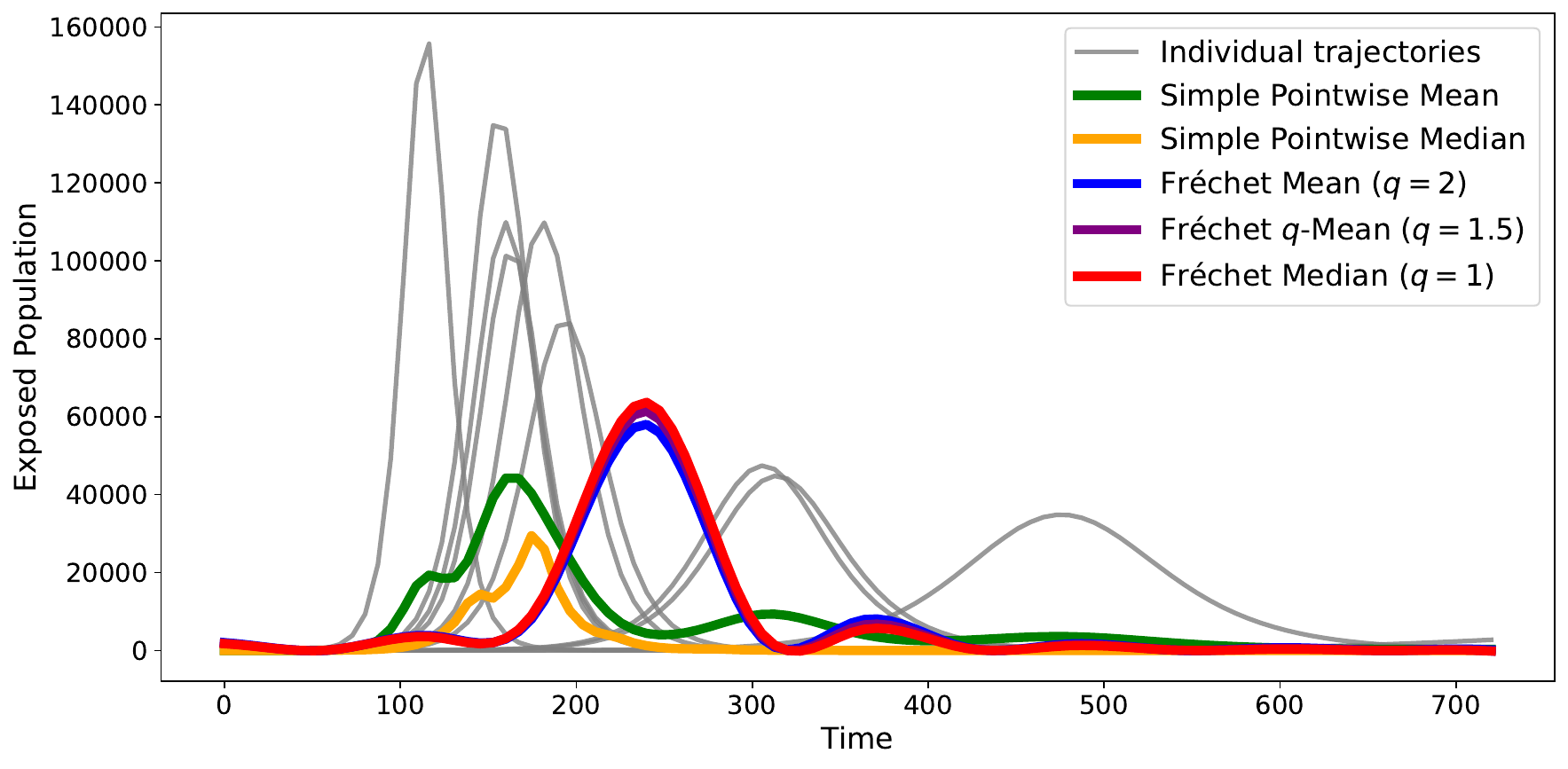}
        \caption{Exposed population ($K=15$)}
        \label{fig:trajectories_E_J10_K15}
    \end{subfigure}
    \hfill
    \begin{subfigure}[t]{0.48\textwidth}
        \centering
        \includegraphics[width=\textwidth]{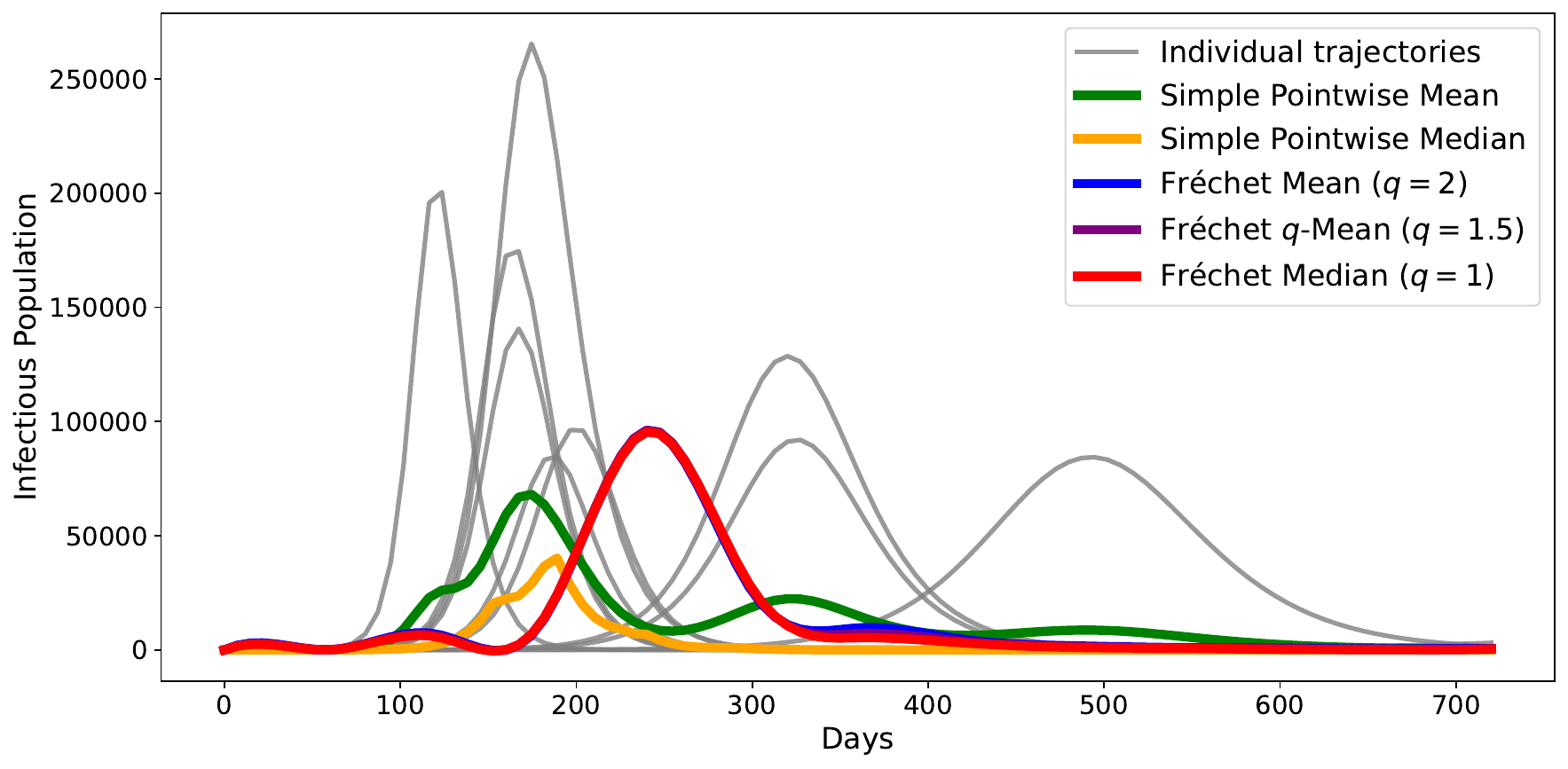}
        \caption{Infectious population ($K=15$)}
        \label{fig:trajectories_I_J10_K15}
    \end{subfigure}
    \hfill
    \begin{subfigure}[t]{0.48\textwidth}
        \centering
        \includegraphics[width=\textwidth]{images/SEIR_E_representative_J10_K30_R1.pdf}
        \caption{Exposed population ($K=30$)}
        \label{fig:trajectories_E_J10_K30}
    \end{subfigure}
    \hfill
    \begin{subfigure}[t]{0.48\textwidth}
        \centering
        \includegraphics[width=\textwidth]{images/SEIR_I_representative_J10_K30_R1.pdf}
        \caption{Infectious population ($K=30$)}
        \label{fig:trajectories_I_J10_K30}
    \end{subfigure}
    \hfill
    \begin{subfigure}[t]{0.48\textwidth}
        \centering
        \includegraphics[width=\textwidth]{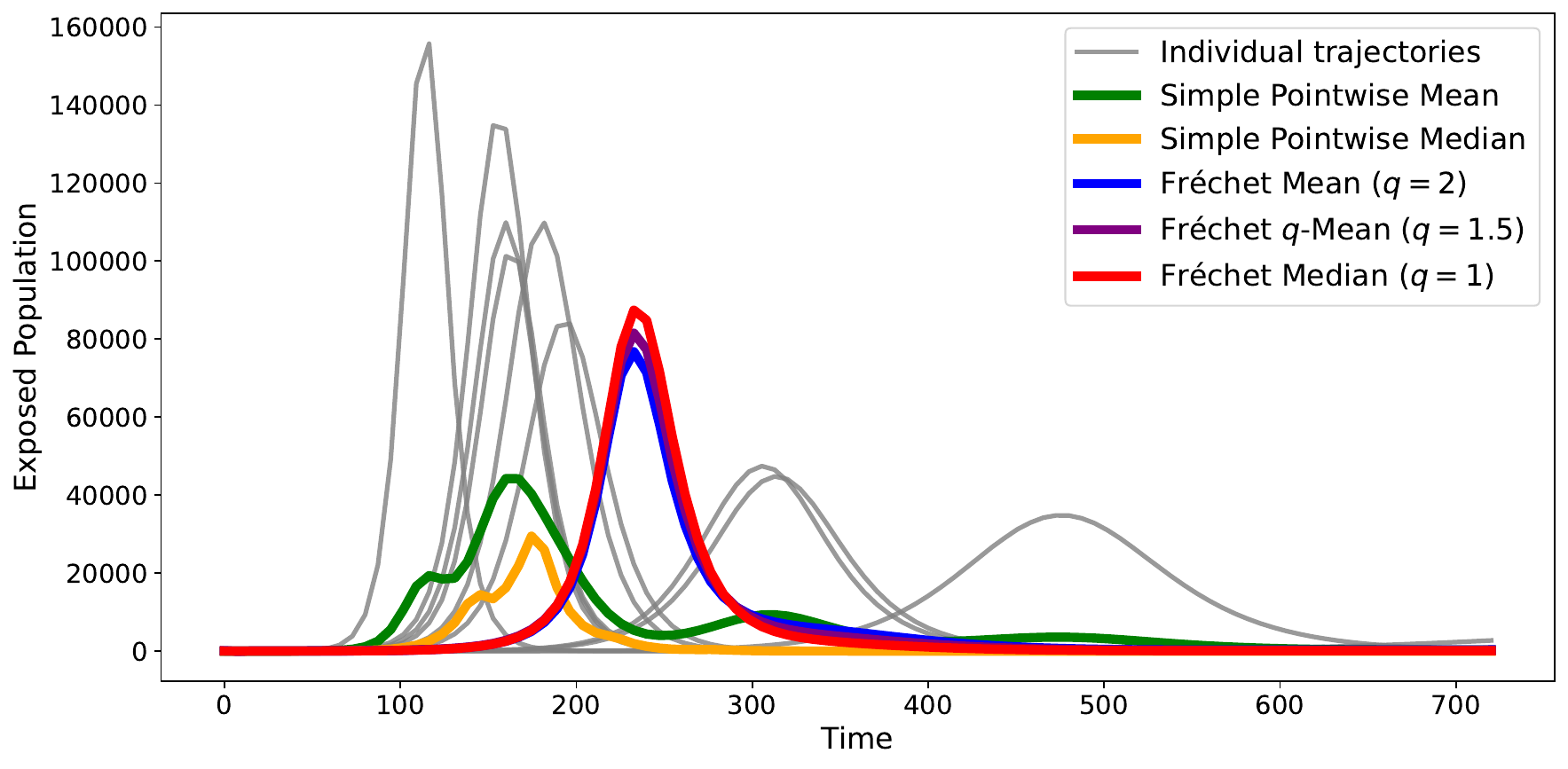}
        \caption{Exposed population ($K=60$)}
        \label{fig:trajectories_E_J10_K60}
    \end{subfigure}
    \hfill
    \begin{subfigure}[t]{0.48\textwidth}
        \centering
        \includegraphics[width=\textwidth]{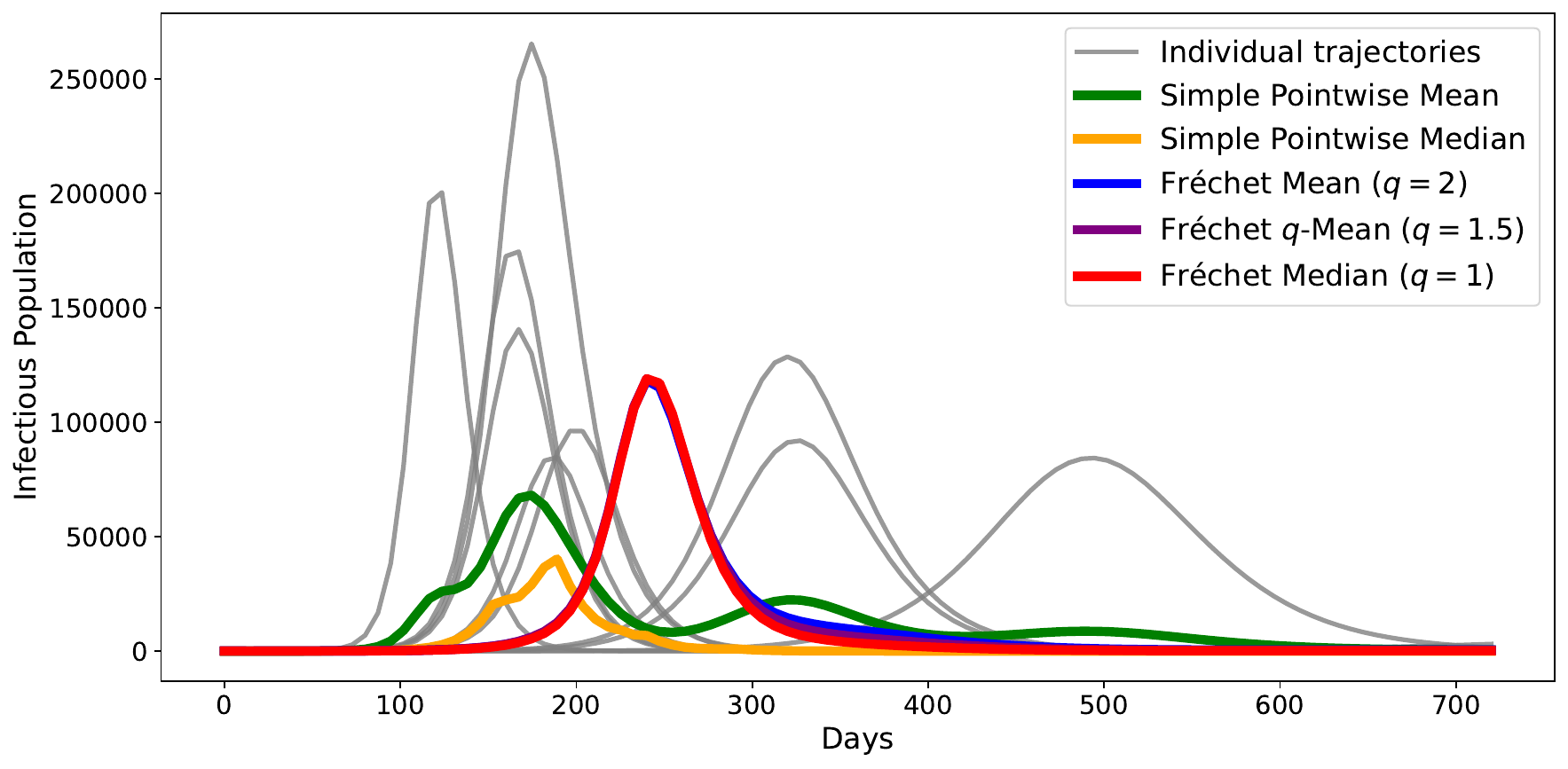}
        \caption{Infectious population ($K=60$)}
        \label{fig:trajectories_I_J10_K60}
    \end{subfigure}
    \caption{Simulated $J=10$ SEIR trajectories (gray) with the Fr\'echet mean (blue), the Fr\'echet median (red), the simple pointwise mean (green), and the simple pointwise median (orange). The power Fr\'echet means were calculated under $\rho=1$.}
    \label{fig:trajectories_EI_sensitivity_K}
\end{figure}

\begin{figure}[htbp]
    \centering
    \begin{subfigure}[t]{0.48\textwidth}
        \centering
        \includegraphics[width=\textwidth]{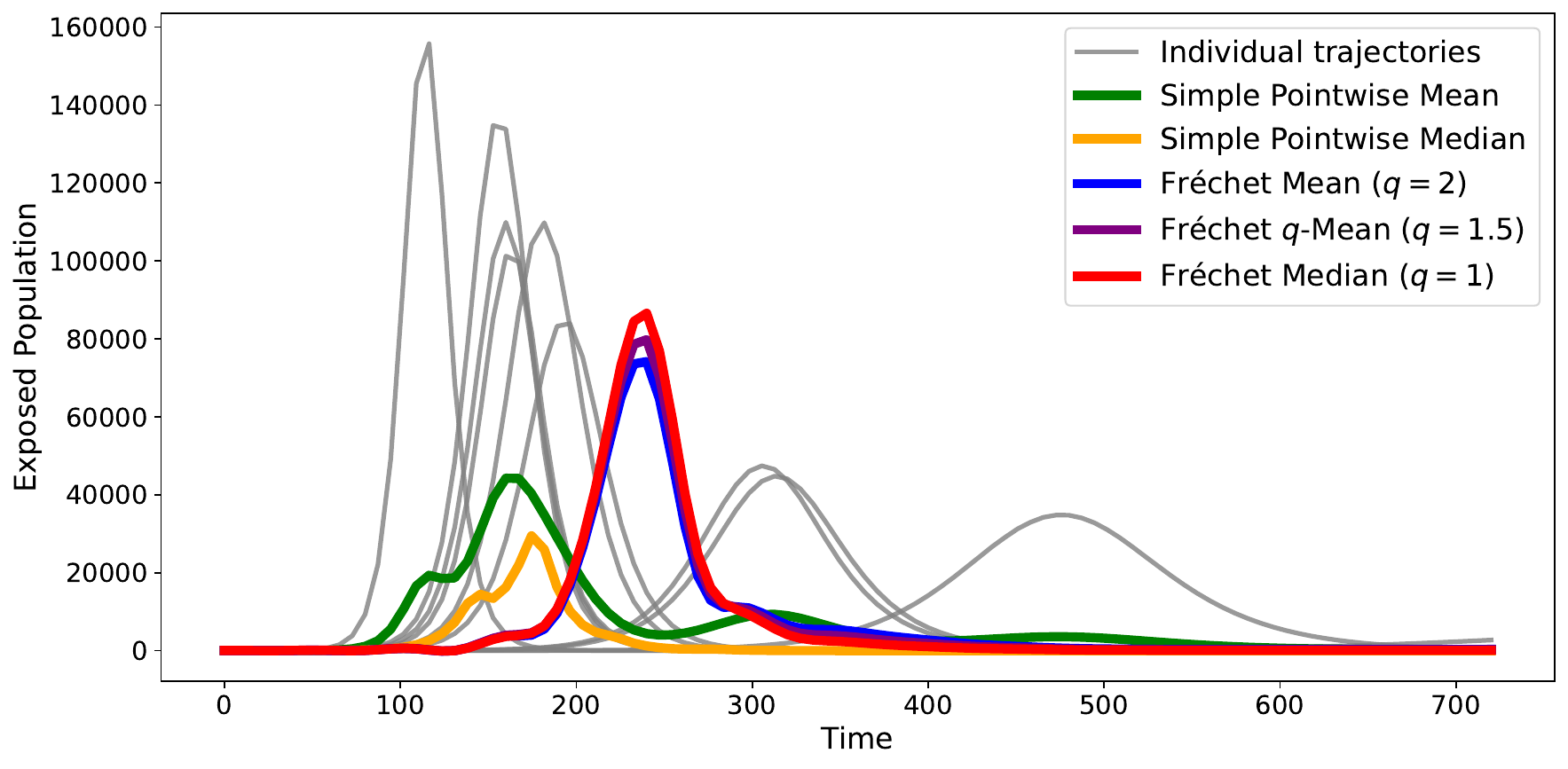}
        \caption{Exposed population ($\rho=0.1$)}
        \label{fig:trajectories_E_J10_R10}
    \end{subfigure}
    \hfill
    \begin{subfigure}[t]{0.48\textwidth}
        \centering
        \includegraphics[width=\textwidth]{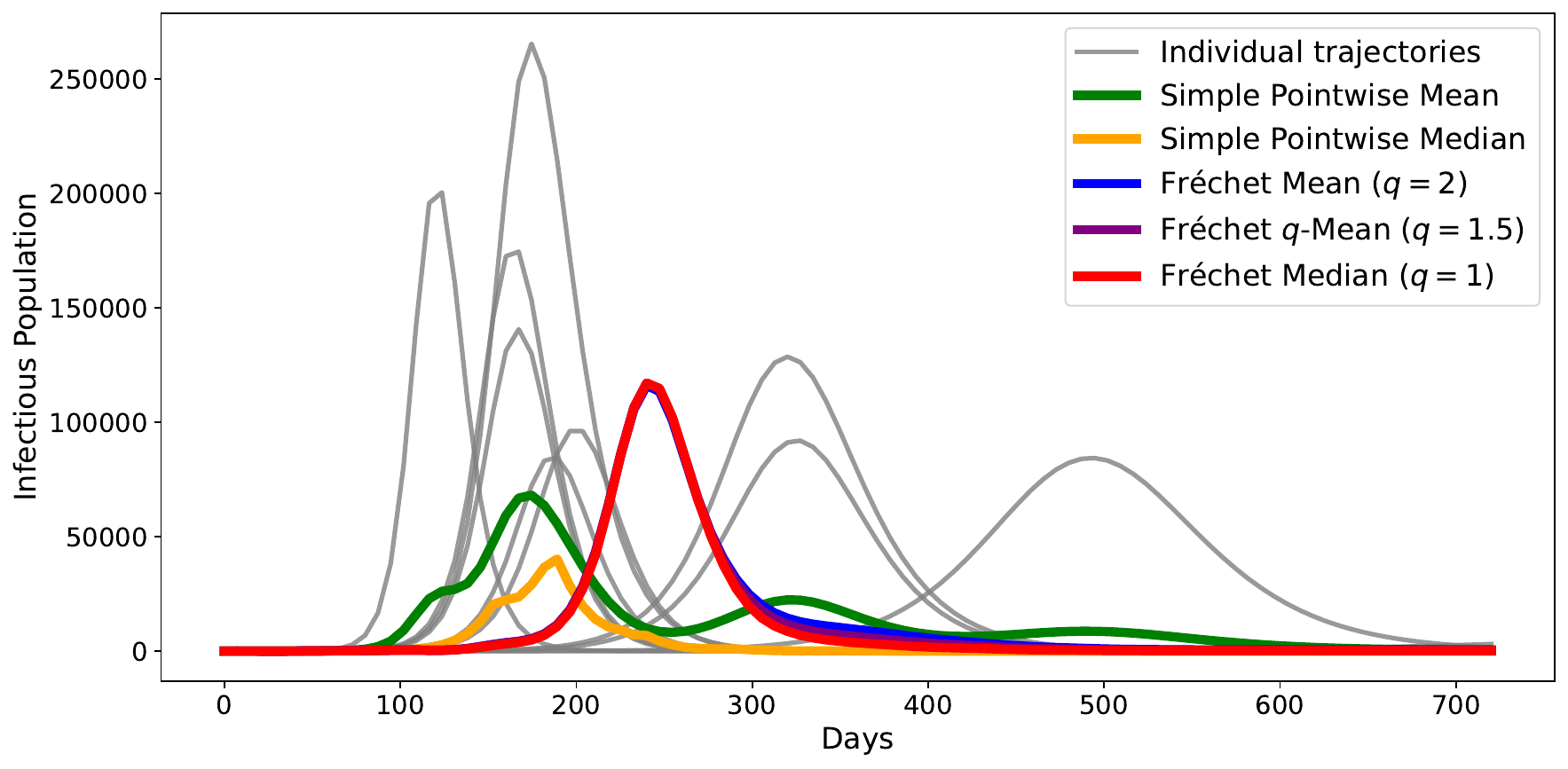}
        \caption{Infectious population ($\rho=0.1$)}
        \label{fig:trajectories_I_J10_R10}
    \end{subfigure}
    \hfill
    \begin{subfigure}[t]{0.48\textwidth}
        \centering
        \includegraphics[width=\textwidth]{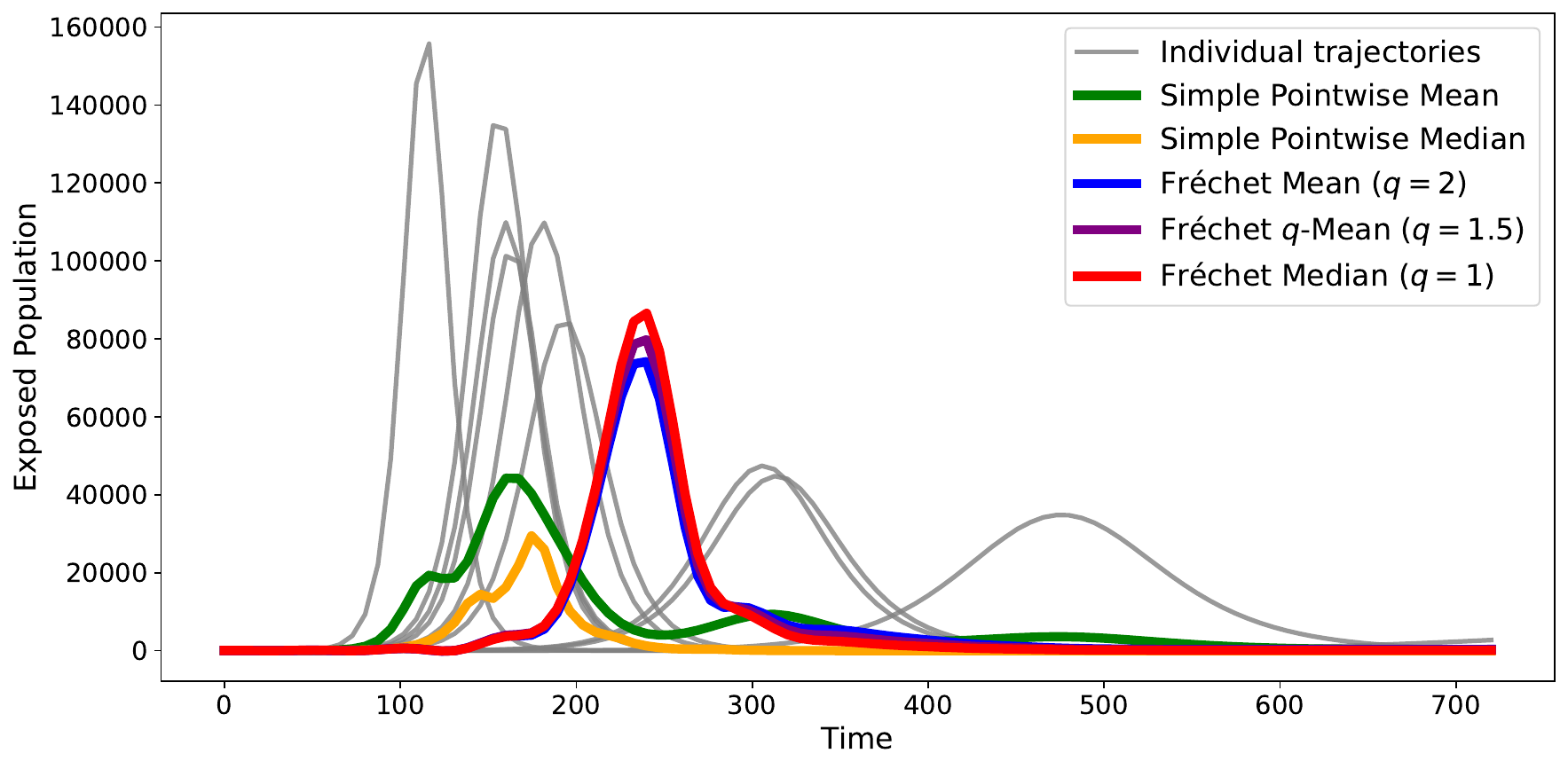}
        \caption{Exposed population ($\rho=0.5$)}
        \label{fig:trajectories_E_J10_R2}
    \end{subfigure}
    \hfill
    \begin{subfigure}[t]{0.48\textwidth}
        \centering
        \includegraphics[width=\textwidth]{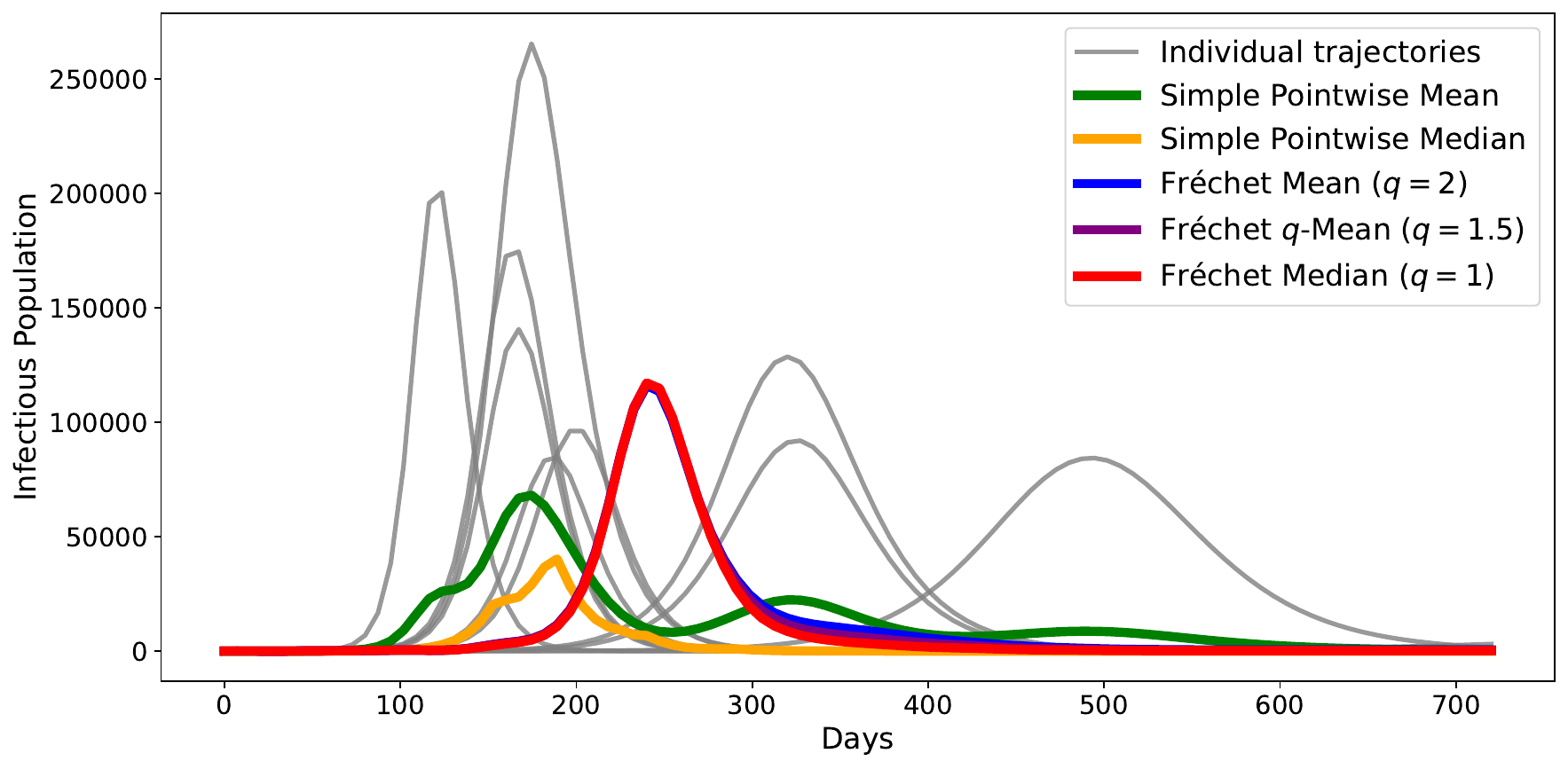}
        \caption{Infectious population ($\rho=0.5$)}
        \label{fig:trajectories_I_J10_R2}
    \end{subfigure}
    \hfill
    \begin{subfigure}[t]{0.48\textwidth}
        \centering
        \includegraphics[width=\textwidth]{images/SEIR_E_representative_J10_K30_R1.pdf}
        \caption{Exposed population ($\rho=1$)}
        \label{fig:trajectories_E_J10_R1}
    \end{subfigure}
    \hfill
    \begin{subfigure}[t]{0.48\textwidth}
        \centering
        \includegraphics[width=\textwidth]{images/SEIR_I_representative_J10_K30_R1.pdf}
        \caption{Infectious population ($\rho=1$)}
        \label{fig:trajectories_I_J10_R1}
    \end{subfigure}
    \caption{Simulated $J=10$ SEIR trajectories (gray) with the Fr\'echet mean (blue), the Fr\'echet median (red), the simple pointwise mean (green), and the simple pointwise median (orange). The power Fr\'echet means were calculated under $K=30$.}
    \label{fig:trajectories_EI_sensitivity_rho}
\end{figure}

\begin{table}[htbp]
\centering
\caption{Estimated parameters and transmission rate under the SEIR model assumption from the consensus trajectories under different settings of $K$ and $q$ for $J=10$. The power Fr\'echet means were calculated under $\rho=1$.}
\label{tab:estimated_params_sensitivity_K}
{\small
\begin{tabular}{cccccccc}
\toprule
$K$ & $q$ & $\sigma$ & $\gamma$ & Incubation ($1/\sigma$) & Infectious ($1/\gamma$) & $\beta$ & $R_0 = \beta/\gamma$\\
\midrule
15 & 1 & \makecell{0.186~(0.017)} & \makecell{0.118~(0.017)} & \makecell{5.411~(0.508)} & \makecell{8.647~(1.340)} & \makecell{0.272~(0.042)} & \makecell{2.299~(0.129)} \\
15 & 1.5 & \makecell{0.184~(0.017)} & \makecell{0.113~(0.016)} & \makecell{5.493~(0.523)} & \makecell{9.027~(1.402)} & \makecell{0.260~(0.041)} & \makecell{2.297~(0.142)} \\
15 & 2 & \makecell{0.182~(0.018)} & \makecell{0.109~(0.016)} & \makecell{5.553~(0.602)} & \makecell{9.418~(1.568)} & \makecell{0.249~(0.042)} & \makecell{2.287~(0.166)} \\
30 & 1 & \makecell{0.200~(0.018)} & \makecell{0.128~(0.021)} & \makecell{5.049~(0.502)} & \makecell{8.048~(1.524)} & \makecell{0.302~(0.075)} & \makecell{2.334~(0.317)} \\
30 & 1.5 & \makecell{0.197~(0.019)} & \makecell{0.121~(0.021)} & \makecell{5.118~(0.557)} & \makecell{8.507~(1.630)} & \makecell{0.284~(0.073)} & \makecell{2.316~(0.310)} \\
30 & 2 & \makecell{0.196~(0.021)} & \makecell{0.116~(0.021)} & \makecell{5.177~(0.642)} & \makecell{8.941~(1.793)} & \makecell{0.269~(0.072)} & \makecell{2.291~(0.312)} \\
60 & 1 & \makecell{0.186~(0.012)} & \makecell{0.120~(0.017)} & \makecell{5.410~(0.361)} & \makecell{8.518~(1.315)} & \makecell{0.260~(0.045)} & \makecell{2.168~(0.209)} \\
60 & 1.5 & \makecell{0.183~(0.012)} & \makecell{0.113~(0.016)} & \makecell{5.502~(0.395)} & \makecell{9.020~(1.368)} & \makecell{0.243~(0.042)} & \makecell{2.148~(0.193)} \\
60 & 2 & \makecell{0.181~(0.014)} & \makecell{0.108~(0.015)} & \makecell{5.576~(0.492)} & \makecell{9.479~(1.509)} & \makecell{0.230~(0.041)} & \makecell{2.126~(0.195)} \\
\bottomrule
\end{tabular}
}
\end{table}

\begin{table}[htbp]
\centering
\caption{Estimated parameters and transmission rate under the SEIR model assumption from the consensus trajectories under different settings of $\rho$ and $q$ for $J=10$. The power Fr\'echet means were calculated under $K=30$.}
\label{tab:estimated_params_sensitivity_rho}
{\small
\begin{tabular}{cccccccc}
\toprule
$\rho$ & $q$ & $\sigma$ & $\gamma$ & Incubation ($1/\sigma$) & Infectious ($1/\gamma$) & $\beta$ & $R_0 = \beta/\gamma$\\
\midrule
0.1 & 1 & \makecell{0.200~(0.018)} & \makecell{0.128~(0.021)} & \makecell{5.050~(0.502)} & \makecell{8.048~(1.523)} & \makecell{0.302~(0.075)} & \makecell{2.334~(0.317)} \\
0.1 & 1.5 & \makecell{0.197~(0.019)} & \makecell{0.121~(0.021)} & \makecell{5.117~(0.550)} & \makecell{8.506~(1.625)} & \makecell{0.284~(0.073)} & \makecell{2.316~(0.310)} \\
0.1 & 2 & \makecell{0.196~(0.021)} & \makecell{0.116~(0.021)} & \makecell{5.178~(0.642)} & \makecell{8.943~(1.793)} & \makecell{0.268~(0.072)} & \makecell{2.291~(0.311)} \\
0.5 & 1 & \makecell{0.200~(0.018)} & \makecell{0.128~(0.021)} & \makecell{5.050~(0.502)} & \makecell{8.048~(1.524)} & \makecell{0.302~(0.075)} & \makecell{2.334~(0.317)} \\
0.5 & 1.5 & \makecell{0.197~(0.019)} & \makecell{0.121~(0.021)} & \makecell{5.118~(0.557)} & \makecell{8.507~(1.630)} & \makecell{0.284~(0.073)} & \makecell{2.316~(0.310)} \\
0.5 & 2 & \makecell{0.196~(0.021)} & \makecell{0.116~(0.021)} & \makecell{5.177~(0.642)} & \makecell{8.942~(1.793)} & \makecell{0.269~(0.072)} & \makecell{2.291~(0.311)} \\
1.0 & 1 & \makecell{0.200~(0.018)} & \makecell{0.128~(0.021)} & \makecell{5.049~(0.502)} & \makecell{8.048~(1.524)} & \makecell{0.302~(0.075)} & \makecell{2.334~(0.317)} \\
1.0 & 1.5 & \makecell{0.197~(0.019)} & \makecell{0.121~(0.021)} & \makecell{5.118~(0.557)} & \makecell{8.507~(1.630)} & \makecell{0.284~(0.073)} & \makecell{2.316~(0.310)} \\
1.0 & 2 & \makecell{0.196~(0.021)} & \makecell{0.116~(0.021)} & \makecell{5.177~(0.642)} & \makecell{8.941~(1.793)} & \makecell{0.269~(0.072)} & \makecell{2.291~(0.312)} \\
\bottomrule
\end{tabular}
}
\end{table}

\bibliography{bibliography}

@article{aron1984seasonality,
    author    = {Aron, J.~L. and Schwartz, I.~B.},
    title     = {Seasonality and period-doubling bifurcations in an epidemic model},
    journal   = {Journal of Theoretical Biology},
    volume    = {110},
    number    = {4},
    pages     = {665--679},
    year      = {1984},
    publisher = {Elsevier},
}

@article{bigot2013frechet,
    author  = {Bigot, J.},
    title   = {Fr{\'e}chet means of curves for signal averaging and application to {ECG} data analysis},
    journal = {The Annals of Applied Statistics},
    volume  = {7},
    number  = {4},
    pages   = {2384--2401},
    year    = {2013},
}

@article{boschi2021functional,
    author    = {Boschi, T. and Di Iorio, J. and Testa, L. and Cremona, M.~A. and Chiaromonte, F.},
    title     = {Functional data analysis characterizes the shapes of the first {COVID}-19 epidemic wave in {Italy}},
    journal   = {Scientific Reports},
    volume    = {11},
    pages     = {17054},
    year      = {2021},
    publisher = {Nature Publishing Group UK London},
}

@article{chowell2020real,
    author    = {Chowell, G. and Luo, R. and Sun, K. and Roosa, K. and Tariq, A. and Viboud, C.},
    title     = {Real-time forecasting of epidemic trajectories using computational dynamic ensembles},
    journal   = {Epidemics},
    volume    = {30},
    pages     = {100379},
    year      = {2020},
    publisher = {Elsevier},
}

@article{etbaigha2018seir,
    author    = {Etbaigha, F. and Willms, A.~R. and Poljak, Z.},
    title     = {An {SEIR} model of influenza {A} virus infection and reinfection within a farrow-to-finish swine farm},
    journal   = {PLoS One},
    volume    = {13},
    number    = {9},
    pages     = {e0202493},
    year      = {2018},
    publisher = {Public Library of Science San Francisco, CA USA},
}

@article{fox2024optimizing,
    author  = {Fox, S.~J. and Kim, M. and Meyers, L.~A. and Reich, N.~G. and Ray, E.~L.},
    title   = {Optimizing disease outbreak forecast ensembles},
    journal = {Emerging Infectious Diseases},
    volume  = {30},
    number  = {9},
    pages   = {1967--1969},
    year    = {2024},
}

@article{he2020seir,
    author    = {He, S. and Peng, Y. and Sun, K.},
    title     = {{SEIR} modeling of the {COVID}-19 and its dynamics},
    journal   = {Nonlinear Dynamics},
    volume    = {101},
    pages     = {1667--1680},
    year      = {2020},
    publisher = {Springer},
}

@article{kermack1927contribution,
    author    = {Kermack, W.~O. and McKendrick, A.~G.},
    title     = {A contribution to the mathematical theory of epidemics},
    journal   = {Proceedings of the Royal Society of London. Series A, Containing Papers of a Mathematical and Physical Character},
    volume    = {115},
    number    = {772},
    pages     = {700--721},
    year      = {1927},
    publisher = {The Royal Society London},
}

@article{li2020early,
    author    = {Li, Q. and Guan, X. and Wu, P. and Wang, X. and Zhou, L. and Tong, Y. and Ren, R. and Leung, K.~S. and Lau, E.~H. and Wong, J.~Y. and Xing, X. and Xiang, N. and Wu, Y. and Li, C. and Chen, Q. and Li, D. and Liu, T. and Zhao, J. and Liu, M. and Tu, W. and Chen, C. and Jin, L. and Yang, R. and Wang, Q. and Zhou, S. and Wang, R. and Liu, H. and Luo, Y. and Liu, Y. and Shao, G. and Li, H. and Tao, Z. and Yang, Y. and Deng, Z. and Liu, B. and Ma, Z. and Zhang, Y. and Shi, G. and Lam, T.~T. and Wu, J.~T. and Gao, G.~F. and Cowling, B.~J. and Yang, B. and Leung, G.~M. and Feng, Z.},
    title     = {Early transmission dynamics in {Wuhan}, {China}, of novel coronavirus--infected pneumonia},
    journal   = {New England Journal of Medicine},
    volume    = {382},
    number    = {13},
    pages     = {1199--1207},
    year      = {2020},
    publisher = {Mass Medical Soc},
}

@article{mcandrew2021adaptively,
    author    = {McAndrew, T. and Reich, N.~G.},
    title     = {Adaptively stacking ensembles for influenza forecasting},
    journal   = {Statistics in Medicine},
    volume    = {40},
    number    = {30},
    pages     = {6931--6952},
    year      = {2021},
    publisher = {Wiley Online Library},
}

@article{mcgowan2019collaborative,
    author    = {McGowan, C.~J. and Biggerstaff, M. and Johansson, M. and Apfeldorf, K.~M. and Ben-Nun, M. and Brooks, L. and Convertino, M. and Erraguntla, M. and Farrow, D.~C. and Freeze, J. and Ghosh, S. and Hyun, S. and Kandula, S. and Lega, J. and Liu, Y. and Michaud, N. and Morita, H. and Niemi, J. and Ramakrishnan, N. and Ray, E.~L. and Reich, N.~G. and Riley, P. and Shaman, J. and Tibshirani, R. and Vespignani, A. and Zhang, Q. and Reed, C. and {Influenza Forecasting Working Group}},
    title     = {Collaborative efforts to forecast seasonal influenza in the {United} {States}, 2015--2016},
    journal   = {Scientific Reports},
    volume    = {9},
    pages     = {683},
    year      = {2019},
    publisher = {Nature Publishing Group UK London},
}

@article{petersen2019frechet,
    author    = {Petersen, A. and M{\"u}ller, H.-G.},
    title     = {Fr{\'e}chet regression for random objects with {Euclidean} predictors},
    journal   = {The Annals of Statistics},
    volume    = {47},
    number    = {2},
    pages     = {691--719},
    year      = {2019},
    publisher = {Institute of Mathematical Statistics},
}

@article{ray2018prediction,
    author    = {Ray, E.~L. and Reich, N.~G.},
    title     = {Prediction of infectious disease epidemics via weighted density ensembles},
    journal   = {PLoS Computational Biology},
    volume    = {14},
    number    = {2},
    pages     = {e1005910},
    year      = {2018},
    publisher = {Public Library of Science San Francisco, CA USA},
}

@article{ray2023comparing,
    author    = {Ray, E.~L. and Brooks, L.~C. and Bien, J. and Biggerstaff, M. and Bosse, N.~I. and Bracher, J. and Cramer, E.~Y. and Funk, S. and Gerding, A. and Johansson, M.~A. and Rumack, A. and Wang, Y. and Zorn, M. and Tibshirani, R.~J. and Reich, N.~G.},
    title     = {Comparing trained and untrained probabilistic ensemble forecasts of {COVID}-19 cases and deaths in the {United} {States}},
    journal   = {International Journal of Forecasting},
    volume    = {39},
    number    = {3},
    pages     = {1366--1383},
    year      = {2023},
    publisher = {Elsevier},
}

@article{reich2019accuracy,
    author    = {Reich, N.~G. and McGowan, C.~J. and Yamana, T.~K. and Tushar, A. and Ray, E.~L. and Osthus, D. and Kandula, S. and Brooks, L.~C. and Crawford-Crudell, W. and Gibson, G.~C. and Moore, E. and Silva, R. and Biggerstaff, M. and Johansson, M.~A. and Rosenfeld, R. and Shaman, J.},
    title     = {Accuracy of real-time multi-model ensemble forecasts for seasonal influenza in the {US}},
    journal   = {PLoS Computational Biology},
    volume    = {15},
    number    = {11},
    pages     = {e1007486},
    year      = {2019},
    publisher = {Public Library of Science San Francisco, CA USA},
}

@article{taylor2022interval,
    author    = {Taylor, K.~S. and Taylor, J.~W.},
    title     = {Interval forecasts of weekly incident and cumulative {COVID}-19 mortality in the {United} {States}: {A} comparison of combining methods},
    journal   = {PLoS One},
    volume    = {17},
    number    = {3},
    pages     = {e0266096},
    year      = {2022},
    publisher = {Public Library of Science San Francisco, CA USA},
}

@article{taylor2023combining,
    author    = {Taylor, J.~W. and Taylor, K.~S.},
    title     = {Combining probabilistic forecasts of {COVID}-19 mortality in the {United} {States}},
    journal   = {European Journal of Operational Research},
    volume    = {304},
    number    = {1},
    pages     = {25--41},
    year      = {2023},
    publisher = {Elsevier},
}

@article{wang2016functional,
    author    = {Wang, J.-L. and Chiou, J.-M. and M{\"u}ller, H.-G.},
    title     = {Functional data analysis},
    journal   = {Annual Review of Statistics and Its Application},
    volume    = {3},
    pages     = {257--295},
    year      = {2016},
    publisher = {Annual Reviews},
}

@incollection{white2010mathematical,
    author    = {White, P.~J. and Enright, M.~C.},
    title     = {Chapter 5 - {Mathematical} models in infectious disease epidemiology},
    booktitle = {Infectious Diseases (Third Edition)},
    pages     = {70--75},
    year      = {2010},
    publisher = {Mosby},
}

@article{yamana2016superensemble,
    author    = {Yamana, T.~K. and Kandula, S. and Shaman, J.},
    title     = {Superensemble forecasts of dengue outbreaks},
    journal   = {Journal of the Royal Society Interface},
    volume    = {13},
    number    = {123},
    pages     = {20160410},
    year      = {2016},
    publisher = {The Royal Society},
}

@article{yamana2017individual,
    author    = {Yamana, T.~K. and Kandula, S. and Shaman, J.},
    title     = {Individual versus superensemble forecasts of seasonal influenza outbreaks in the {United} {States}},
    journal   = {PLoS Computational Biology},
    volume    = {13},
    number    = {11},
    pages     = {e1005801},
    year      = {2017},
    publisher = {Public Library of Science San Francisco, CA USA},
}

@book{brezis2011functional,
    author    = {Br{\'e}zis, H.},
    title     = {Functional Analysis, Sobolev Spaces and Partial Differential Equations},
    publisher = {Springer},
    year      = {2011},
}

@book{reed2012methods,
    author    = {Reed, M.},
    title     = {Methods of Modern Mathematical Physics: Functional Analysis},
    publisher = {Elsevier},
    year      = {2012},
}

@article{griette2021can,
    author   = {Griette, Q. and Demongeot, J. and Magal, P.},
    title    = {What can we learn from {COVID-19} data by using epidemic models with unidentified infectious cases?},
    journal  = {Mathematical Biosciences and Engineering},
    volume   = {19},
    number   = {1},
    pages    = {537--594},
    year     = {2022},
    issn     = {1551-0018},
    doi      = {10.3934/mbe.2022025},
    url      = {https://www.aimspress.com/article/doi/10.3934/mbe.2022025},
}

@article{weitz2015modeling,
    author    = {Weitz, J.~S. and Dushoff, J.},
    title     = {Modeling post-death transmission of {Ebola}: Challenges for inference and opportunities for control},
    journal   = {Scientific Reports},
    volume    = {5},
    number    = {1},
    pages     = {8751},
    year      = {2015},
    publisher = {Nature Publishing Group UK London},
}

@article{gerberry2009seiqr,
    author    = {Gerberry, D.~J. and Milner, F.~A.},
    title     = {An {SEIQR} model for childhood diseases},
    journal   = {Journal of Mathematical Biology},
    volume    = {59},
    number    = {4},
    pages     = {535--561},
    year      = {2009},
    publisher = {Springer},
}

@article{dubey2024metric,
    author    = {Dubey, P. and Chen, Y. and M{\"u}ller, H.-G.},
    title     = {Metric statistics: Exploration and inference for random objects with distance profiles},
    journal   = {The Annals of Statistics},
    volume    = {52},
    number    = {2},
    pages     = {757--792},
    year      = {2024},
    publisher = {Institute of Mathematical Statistics},
}

@article{muller2016peter,
    author    = {M{\"u}ller, H.-G.},
    title     = {{Peter Hall}, functional data analysis and random objects},
    journal   = {The Annals of Statistics},
    pages     = {1867--1887},
    year      = {2016},
    publisher = {JSTOR},
}

@book{ramsay2005functional,
    author    = {Ramsay, J.~O. and Silverman, B.~W.},
    title     = {Functional Data Analysis},
    publisher = {Springer},
    year      = {2005},
    edition   = {2nd},
    series    = {Springer Series in Statistics},
    address   = {New York, NY},
    doi       = {10.1007/b98888},
    isbn      = {978-0-387-40080-8},
}

@book{hsing2015theoretical,
    author    = {Hsing, T. and Eubank, R.~L.},
    title     = {Theoretical Foundations of Functional Data Analysis, With an Introduction to Linear Operators},
    publisher = {Wiley Online Library},
    volume    = {997},
    year      = {2015},
}

@article{schotz2022strong,
    author    = {Sch{\"o}tz, C.},
    title     = {Strong laws of large numbers for generalizations of {Fr{\'e}chet} mean sets},
    journal   = {Statistics},
    volume    = {56},
    number    = {1},
    pages     = {34--52},
    year      = {2022},
    publisher = {Taylor \& Francis},
}

@article{allen2021early,
    author  = {Allen, K. and Parry, A.~E. and Glass, K.},
    title   = {Early reports of epidemiological parameters of the {COVID-19} pandemic},
    journal = {Western Pacific Surveillance and Response Journal: WPSAR},
    volume  = {12},
    number  = {2},
    pages   = {65},
    year    = {2021},
}

@article{wan2020seir,
    author  = {Wan, K. and Chen, J. and Lu, C. and Dong, L. and Wu, Z. and Zhang, L.},
    title   = {When will the battle against novel coronavirus end in {Wuhan}: {A} {SEIR} modeling analysis},
    journal = {Journal of Global Health},
    volume  = {10},
    number  = {1},
    pages   = {011002},
    year    = {2020},
}

@article{dagpunar2020sensitivity,
    author    = {Dagpunar, J.~S.},
    title     = {Sensitivity of {UK} {COVID-19} deaths to the timing of suppression measures and their relaxation},
    journal   = {Infectious Disease Modelling},
    volume    = {5},
    pages     = {525--535},
    year      = {2020},
    publisher = {Elsevier},
}

@article{carcione2020simulation,
    author    = {Carcione, J.~M. and Santos, J.~E. and Bagaini, C. and Ba, J.},
    title     = {A simulation of a {COVID-19} epidemic based on a deterministic {SEIR} model},
    journal   = {Frontiers in Public Health},
    volume    = {8},
    pages     = {230},
    year      = {2020},
    publisher = {Frontiers Media SA},
}

@article{peirlinck2020outbreak,
    author    = {Peirlinck, M. and Linka, K. and Sahli Costabal, F. and Kuhl, E.},
    title     = {Outbreak dynamics of {COVID-19} in {China} and the {United States}},
    journal   = {Biomechanics and Modeling in Mechanobiology},
    volume    = {19},
    number    = {6},
    pages     = {2179--2193},
    year      = {2020},
    publisher = {Springer},
}

@article{wu2020nowcasting,
    author    = {Wu, J.~T. and Leung, K. and Leung, G.~M.},
    title     = {Nowcasting and forecasting the potential domestic and international spread of the {2019-nCoV} outbreak originating in {Wuhan}, {China}: A modelling study},
    journal   = {The Lancet},
    volume    = {395},
    number    = {10225},
    pages     = {689--697},
    year      = {2020},
    publisher = {Elsevier},
}

@article{tang2020estimation,
    author    = {Tang, B. and Wang, X. and Li, Q. and Bragazzi, N.~L. and Tang, S. and Xiao, Y. and Wu, J.},
    title     = {Estimation of the transmission risk of the {2019-nCoV} and its implication for public health interventions},
    journal   = {Journal of Clinical Medicine},
    volume    = {9},
    number    = {2},
    pages     = {462},
    year      = {2020},
    publisher = {MDPI},
}
\bibliographystyle{apalike}

\end{document}